\documentclass[conference]{IEEEtran}
\usepackage{booktabs} 
\usepackage{times} 
\usepackage{color} 
\usepackage{balance} 
\usepackage{url} 
\usepackage[all]{nowidow}
\usepackage{tabularx}
\usepackage{siunitx} 
\usepackage{tabularx,subfigure,multirow, graphicx}
\usepackage{epsfig}
\usepackage{epstopdf} 
\usepackage{enumitem}
\usepackage{caption}
\usepackage[a-2b]{pdfx}
\usepackage{pifont} 
\usepackage{makecell}

\usepackage{hyperref} 

\usepackage{amsthm,amsmath}  

\newcolumntype{L}[1]{>{\raggedright\arraybackslash}p{#1}}
\newcolumntype{C}[1]{>{\centering\arraybackslash}p{#1}}
\newcolumntype{R}[1]{>{\raggedleft\arraybackslash}p{#1}}

\usepackage[linesnumbered,ruled,vlined]{algorithm2e}
\SetKwRepeat{Do}{do}{while}
\SetCommentSty{mycommfont}

\long\def\comment#1{}

\newcommand{\nop}[1]{}

\newtheorem{theorem}{\bf Theorem}[section]

\newtheorem{example}{\bf Example}

\theoremstyle{remark}

\theoremstyle{definition}
\newtheorem{definition}{\bf Definition}
\IEEEoverridecommandlockouts
\usepackage{cite}
\usepackage{amsmath,amssymb,amsfonts}
\usepackage{algorithmic}
\usepackage{graphicx}
\usepackage{textcomp}
\usepackage{xcolor}

\newcommand{\revisedtext}[1]{{\color{black}#1}}

\def\BibTeX{{\rm B\kern-.05em{\sc i\kern-.025em b}\kern-.08em
    T\kern-.1667em\lower.7ex\hbox{E}\kern-.125emX}}
\begin{document}

\title{SINDI: An Efficient Index for Sparse Vector Approximate Maximum Inner Product Search
}

\author{
	Ruoxuan Li{\small$~^{1}$}, Xiaoyao Zhong{\small$~^{2}$}, Jiabao Jin{\small$~^{2}$}, Peng Cheng{\small$~^{3,1}$}, Wangze Ni{\small$~^{4}$}, Zhitao Shen{\small$~^{2}$}, Wei Jia{\small$~^{2}$},\\ 
	Xiangyu Wang{\small$~^{2}$}, Heng Tao Shen{\small$~^{3}$}, Jingkuan Song{\small$~^{3}$}
	\\
	\fontsize{10}{10}\itshape
	$~^{1}$East China Normal University, Shanghai, China; 
	$~^{2}$Ant Group, Shanghai, China;\\
	$~^{3}$Tongji University, Shanghai, China;
	$~^{4}$Zhejiang University, Hangzhou China
	\fontsize{9}{9}\upshape\\
	rxlee@stu.ecnu.edu.cn, zhongxiaoyao.zxy@antgroup.com, jinjiabao.jjb@antgroup.com,\\ 
	cspcheng@tongji.edu.cn, niwangze@zju.edu.cn, zhitao.szt@antgroup.com, jw94525@antgroup.com, \\
	wxy407827@antgroup.com, shenhengtao@hotmail.com, jingkuan.song@gmail.com
}

\maketitle

\begin{abstract}
Sparse vector Maximum Inner Product Search (MIPS) is crucial in multi-path retrieval for Retrieval-Augmented Generation (RAG). Recent inverted index-based and graph-based algorithms have achieved high search accuracy with practical efficiency. However, their performance in production environments is often limited by redundant distance computations and frequent random memory accesses. Furthermore, the compressed storage format of sparse vectors hinders the use of SIMD acceleration.
    In this paper, we propose the \emph{sparse inverted non-redundant distance index} (\textsc{Sindi}), which incorporates three key optimizations:   
    (i) Efficient Inner Product Computation: \textsc{Sindi} leverages SIMD acceleration and eliminates redundant identifier lookups, enabling batched inner product computation;  
    (ii) Memory-Friendly Design: \textsc{Sindi} replaces random memory accesses to original vectors with sequential accesses to inverted lists, substantially reducing memory-bound latency.
    (iii) Vector Pruning: \textsc{Sindi} retains only the high-value non-zero entries of vectors, improving query throughput while maintaining accuracy.
    We evaluate \textsc{Sindi} on multiple real-world datasets. Experimental results show that \textsc{Sindi} achieves state-of-the-art performance across datasets of varying scales, languages, and models. On the \textsc{MsMarco} dataset, when Recall@50 exceeds 99\%, \textsc{Sindi} delivers single-thread query-per-second (QPS) improvements ranging from $4.2\times$ to $26.4\times$ compared with \textsc{Seismic} and \textsc{PyANNS}. Notably, \textsc{Sindi} has been integrated into Ant Group's open-source vector search library, \emph{VSAG}.
\end{abstract}



\section{Introduction}
\label{sec:introduction}

Recently, retrieval-augmented generation (RAG)~\cite{blendedrag, maxrag, neuripsrag, parametricrag, accrag, enhancerag} \revisedtext{has} become one of \revisedtext{the} most successful information retrieval framework attracting attention from research communities and industry. Usually, texts are embedded into dense vectors (i.e., no dimension of the vector is zero entry) in RAG, then retrieved through approximate nearest neighbor search (ANNS) on their corresponding dense vectors~\cite{vsag, enhanced_graph, yukun, mingyu}.

\revisedtext{To enhance the RAG framework, researchers have found that complementing dense vector-based RAG with sparse vector retrieval yields superior accuracy and recall~\cite{ma2025lightretriever, sallinen2025mmoremassivemultimodalopen, ahmad2025transformertafsirqias2025, fensore2025evaluatinghybridretrievalaugmented}. Unlike dense vectors, sparse vectors (where only a small fraction of dimensions are non-zero) are generated by specific models (e.g., \textsc{Splade}~\cite{distill-splade,  splade, efficiencyspalde,  towardssplade}) to preserve semantic information while enabling precise lexical matching~\cite{ma2025lightretriever}. In the enhanced RAG framework, dense vectors capture holistic semantic similarity, while sparse vectors ensure exact term recall. This synergy translates into significant application-level improvements: in the evaluation of AntGroup production RAG system, integrating sparse vectors into a strong hybrid baseline (BM25 + Dense) boosted \textbf{Recall@3 by 18.5\%} (from 55.60\% to 74.10\%) and \textbf{Recall@10 by 3.7\%} (from 85.20\% to 88.90\%). We illustrate the process of this enhanced RAG workflow in the following example:}

\begin{figure}[t]
	\centering
	{\includegraphics[width=0.48\textwidth]{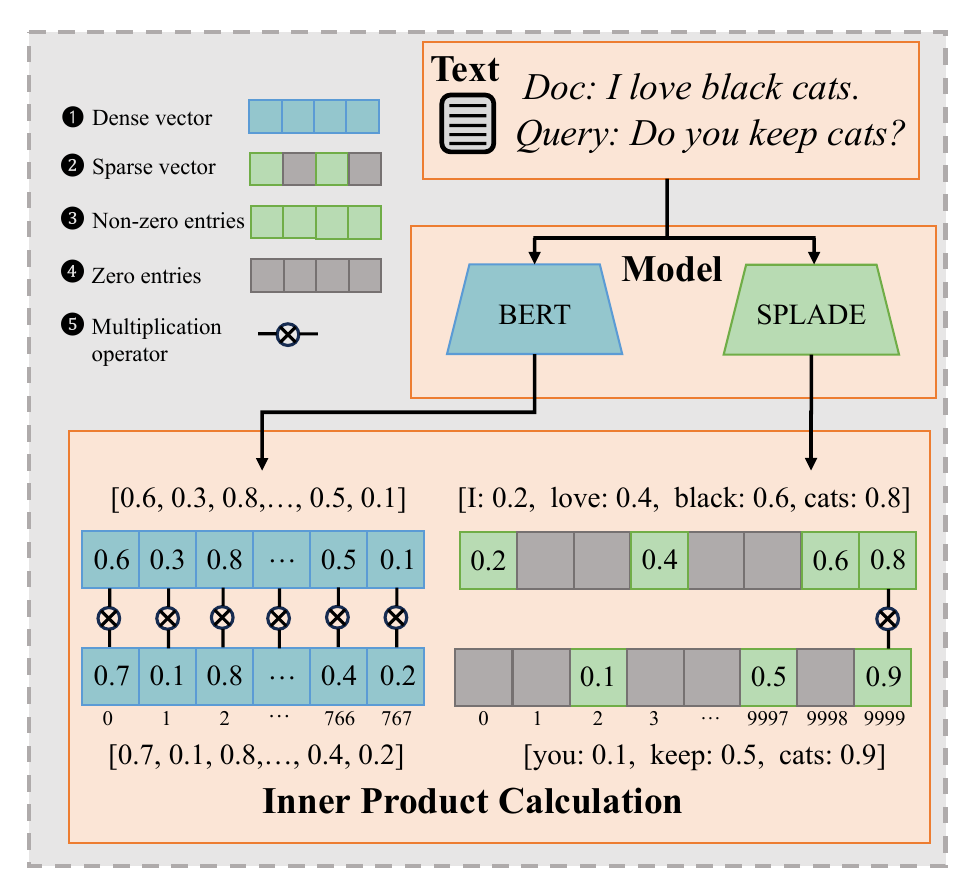}}
	\caption{\small Example of Dense and Sparse Vector Representations and Inner Product Calculations.}\vspace{-2ex}
	\label{fig:sv_vs_dv}
\end{figure}

\begin{example}	
	\label{ex:sv_vs_dv}
	\textbf{Precise lexical matching.} In the retriever stage of RAG, queries and documents are compared to select top-$k$ candidates. Dense vectors capture semantic similarity, while sparse vectors support exact term matching. For example, ``I love black cats'' is tokenized into ``i'', ``love'', ``black'', and ``cats'', with ``cats'' receiving the highest weight (0.8). Queries containing ``cats'' precisely match documents where this token scores high.
	\textbf{Challenges in inner-product computation.} Dense vectors are stored contiguously, enabling parallel SIMD dot-products. Sparse vectors have high dimensionality but compactly store only non-zero entries, causing two bottlenecks:
	(1) \emph{ID lookup overhead}: Matching shared non-zero dimensions requires traversing all entries. Even if only one dimension matches, all must be scanned. 
	(2) \emph{No SIMD acceleration}: Unaligned storage prevents parallel SIMD processing.
\end{example}

The similarity retrieval problem for sparse vectors is formally known as the Maximum Inner Product Search (MIPS)~\cite{ip-1, ip-2, ip-3, ip-4, ip-5}, which aims to identify the top-$k$ vectors in a dataset that have the largest inner product value with a given query vector. However, due to the curse of dimensionality~\cite{indyk1998approximate}, performing exact MIPS in high-dimensional spaces is computationally prohibitive. To mitigate this issue, we focus on Approximate Maximum Inner Product Search (AMIPS)~\cite{ip-3, ip-6, ip-7}, which trades a small amount of recall for significantly improved search efficiency.

Many algorithms~\cite{seismic, sosia, bmp, pyanns} for AMIPS employ techniques like inverted indices~\cite{seismic}, proximity graphs~\cite{pyanns}, and hash partitioning~\cite{sosia}. They improve efficiency by grouping similar vectors, reducing the candidates examined during queries.

Despite reducing the search space, existing approaches still face two major performance bottlenecks:  
(i) \textit{Distance computation cost}: Matching non-zero dimensions between a query and a document incurs substantial identifier lookup overhead, and the inner product computation cannot be effectively accelerated using SIMD instructions.  
(ii) \textit{Random memory access cost}: During query processing, data are accessed in a random manner, and the variable lengths of sparse vectors further complicate direct access in memory.

To address the aforementioned challenges, we propose \textsc{Sindi}, a \textit{Sparse Inverted Non-redundant Distance-calculation Index} for efficient sparse vector search. The main contributions of this paper are: (i) \textit{Value-Storing Inverted Index}: \textsc{Sindi} stores both vector identifiers and values in the inverted index, enabling direct access during query processing; (ii) \textit{Efficient Inner Product Computation}: \textsc{Sindi} eliminates redundant overhead in identifying common dimensions and fully exploits SIMD acceleration. By grouping values under the same dimension, \textsc{Sindi} enables batched inner product computation during queries; (iii) \textit{Cache-Friendly Design}: \textsc{Sindi} reduces random memory accesses by avoiding fetches of original vectors. Instead, it sequentially accesses inverted lists for specific dimensions, thereby lowering cache miss rates. (iv) \textit{Vector Mass Pruning}: \textsc{Sindi} retains only high-value non-zero entries in vectors, effectively reducing the search space and improving query throughput while preserving accuracy.

\begin{table}[t] 
	\renewcommand{\arraystretch}{1.4} 
	\footnotesize
	\caption{Comparison to Existing Algorithms.}
	\vspace{0.5em}
	\label{tab:cmp_to_alg_transposed}
	\begin{tabular}{lccc}
		\toprule
		& \textbf{\textsc{Sindi}(ours)} & \textsc{Seismic} & \textsc{Py}\textsc{Ann}s \\
		\midrule
		Distance Complexity & $O\Bigl( \frac{\Vert q \Vert}{s} \Bigr)$
		& $O(\Vert q \Vert + \Vert x \Vert)$ & $O(\Vert q \Vert+ \Vert x \Vert)$ \\
		Memory Friendly                & \checkmark   & \ding{55}       & \ding{55} \\
		SIMD Support                 & \checkmark  & \ding{55}        & \ding{55} \\
		QPS (Recall@50=99\%)             & 241         & 58               & 24 \\
		Construction Time(s)         & 58          & 220              & 4163 \\
		\bottomrule
	\end{tabular}
\end{table}

We compare \textsc{Sindi} with several state-of-the-art methods on the \textsc{MsMarco} dataset (8.8M scale) in Table~\ref{tab:cmp_to_alg_transposed}. 
\revisedtext{Regarding the computational cost, \textsc{Sindi} achieves an \textbf{amortized per-candidate time complexity} of $O\left(\frac{\lVert q \rVert}{s}\right)$ (symbols defined in Table~\ref{tab:symbols}), establishing a robust performance lower bound across different data distributions.} 
In contrast, traditional inverted index and graph-based algorithms typically scale with $O(\lVert q \rVert + \lVert x \rVert)$. The detailed derivation is given in \S~\ref{subsec:efficient-distance-computation}.

In summary, the contributions of this paper are as follows:  
\begin{itemize}[leftmargin=*,labelsep=0.5em]
	\item We present \textsc{Sindi}, a novel value-storing inverted index described in \S\ref{subsec:value-storing-inverted-index} and \S\ref{subsec:efficient-distance-computation}, which reduces redundant distance computation and random memory accesses. We further introduce a \textit{Window Switch} strategy in \S\ref{subsec:cache-optimization} to support large-scale datasets.
	\item We propose \textit{Vector Mass Pruning} in \S\ref{sec:approximate-inverted-index} to decrease the search space and improve query speed while maintaining accuracy.
	\item We evaluate \textsc{Sindi} on multi-scale, multilingual datasets in \S\ref{sec:experimental}, demonstrating $4\times\sim26\times$ higher single-thread \revisedtext{Queries Per Second (QPS)} than \textsc{PyANNS} and \textsc{Seismic} at over 99\% Recall@50, and achieving 8.8M-scale index construction in 60 seconds with minimal cost.
\end{itemize}

\begin{table}[t]
	\centering\footnotesize
	\caption{Summary of Symbols}
	\vspace{0.5em}
	\begin{tabular}{ll}
		\toprule
		\textbf{Symbol} & \textbf{Description} \\
		\hline
		$\mathcal{D}$ & base dataset \\
		$d$ & dimension of $\mathcal{D}$ \\
		$\vec{x}, \vec{q}$ & base vector, query vector \\
		$x,\ \Vert x \Vert$ & sparse format of $\vec{x}$; number of non-zero entries in $x$ \\
		$\vec{x}_i ,\, x_i$ & $i$-th base vector and its sparse format \\
		$x_i^j$ & value of $\vec{x}_i$ in dimension $j$ \\
		$s$ & SIMD width (elements per SIMD operation) \\
		$\lambda$ & window size \\
		$\sigma$ & number of windows \\
		$\alpha$ & base vector pruning ratio \\
		$\beta$ & query vector pruning ratio \\
		$\gamma$ & reorder pool size \\
		$I,\ I_j$ & inverted index; inverted list for dimension $j$ \\
		$I_{j, w}$ & $w$-th window of inverted list $I_j$ \\
		$P^j$ & temporary product array on dimension $j$ \\
		$A, A[m]$ & distance array; value at index $m$ \\
		$\Omega(\vec{x}_1, \vec{x}_2)$ & set of common non-zero dimensions of $\vec{x}_1$ and $\vec{x}_2$ \\
		$\delta(\vec{x}_1, \vec{x}_2)$ & inner product of $\vec{x}_1$ and $\vec{x}_2$ \\
		\bottomrule
	\end{tabular}
	\label{tab:symbols}
\end{table}

\section{Preliminaries}
\label{sec:preliminaries}

\subsection{Problem Definition}
\label{subsec:definition}
Sparse vectors differ from dense vectors in that most of their dimensions have zero values. By storing only the non-zero entries, they significantly reduce storage and computation costs. We formalize the definition as follows.

\begin{definition}[Sparse Vector and Non-zero Entries]
	\label{def:sparse_vector}
	\textit{
		Let $\mathcal{D} \subseteq \mathbb{R}^{d}$ be a dataset of $d$-dimensional sparse vectors.  
		For any $\vec{x} \in \mathcal{D}$, let $x$ denote its sparse representation, defined as the set of non-zero entries:
		$
		x = \{\, x^j \mid x^j \neq 0,\; j \in [0, d-1] \,\}.
		$
		Here, $x^j$ denotes the value of $\vec{x}$ in dimension $j$.  
		The notation $\lVert x \rVert$ denotes the number of non-zero entries in $x$.
	}
\end{definition}

To avoid confusion, we illustrate sparse vectors with an example in Figure~\ref{fig:sv_vs_dv}. 
\begin{example}
	Consider the document ``I love black cats'' encoded into a sparse embedding: $[ \text{I}: 0.2,\, \text{love}: 0.4,\, \text{black}: 0.6,\, \text{cats}: 0.8 ]$. The corresponding sparse representation is $x = \{ x^0 = 0.2,\, x^3 = 0.4,\, x^{9998} = 0.6,\, x^{9999} = 0.8 \}$, where $\lVert x \rVert = 4$.
\end{example}

Since the similarity measure in this work is based on the inner product, we formally define it as follows.

\begin{definition}[Inner Product on Sparse Vectors]
	\label{def:sparse_inner_product}
	\textit{
		Let $\vec{x}_1, \vec{x}_2 \in \mathcal{D}$, and let $x_1$ and $x_2$ denote their sparse representations.  
		Define the set of common non-zero dimensions as  
		$
		\Omega(\vec{x}_1, \vec{x}_2) 
		= \{\, j \mid x_1^j \in x_1 \ \wedge\ x_2^j \in x_2 \,\}.
		$
		The inner product between $\vec{x}_1$ and $\vec{x}_2$ is then given by
		$
		\delta(\vec{x}_1, \vec{x}_2) 
		= \sum_{j \in \Omega(\vec{x}_1, \vec{x}_2)} x_1^j \cdot x_2^j .
		$
	}
\end{definition}

Given the formal definition of the inner product for sparse vectors, we now define the Sparse Maximum Inner Product Search (Sparse-MIPS) , which finds the vector in the dataset that maximizes this similarity measure with the query.

\begin{definition}[Sparse Maximum Inner Product Search]
	\textit{
		Given a sparse dataset
		$\mathcal{D} \subseteq \mathbb{R}^{d}$ and a query point $\vec{q} \in \mathbb{R}^{d}$, the Sparse Maximum Inner Product Search (Sparse-MIPS) returns a vector $\vec{x}^{\ast} \in \mathcal{D}$ that has the maximum inner product with $\vec{q}$, i.e.,
		$
		\vec{x}^{\ast}
		= \underset{\vec{x} \in \mathcal{D}}{\arg\max}\;
		\delta(\vec{x},\, \vec{q}).
		\label{eq:mips}
		$
	}
	\label{def:mips}
\end{definition}

For small datasets, exact Sparse-MIPS can be obtained by scanning all vectors. For large-scale high-dimensional collections, this is prohibitively expensive, and Approximate Sparse-MIPS mitigates the cost by trading a small loss in accuracy for much higher efficiency.

\begin{definition}[Approximate Sparse Maximum Inner Product Search]
	\textit{
		Given a sparse dataset $\mathcal{D} \subseteq \mathbb{R}^{d}$, a query point $\vec{q}$,
		and an approximation ratio $c \in (0, 1]$, let $\vec{x}^{\ast} \in \mathcal{D}$ be the vector that has the maximum inner product with $\vec{q}$.  
		\revisedtext{A $c$-Maximum Inner Product Search ($c$-MIPS)} returns a point $\vec{x} \in \mathcal{D}$ satisfying
		$
		\delta(\vec{q},\, \vec{x}) \;\ge\;
		c \cdot \delta(\vec{q},\, \vec{x}^{\ast}).
		$
	}
	\label{def:amips}
\end{definition}

In practice, $c$-Sparse-MIPS methods can reduce query latency by orders of magnitude compared with exact search, making them preferable for large-scale, real-time applications such as web search and recommender systems.
 
For ease of reference, the main notations and their meanings are summarized in Table~\ref{tab:symbols}, which will be referred to throughout the rest of the paper.

\subsection{Existing Solutions}
\label{discussion}

Representative algorithms for the AMIPS problem on sparse vectors include the inverted-index based \textsc{Seismic}~\cite{seismic}, the graph based \textsc{PyANNS}~\cite{pyanns}. \textsc{Seismic} constructs an inverted list based on vector dimensions. \textsc{PyANNS} creates a proximity graph where similar vectors are connected as neighbors.

\begin{figure}[t!]\vspace{-2ex}
	\centering
	{\includegraphics[width=0.48\textwidth]{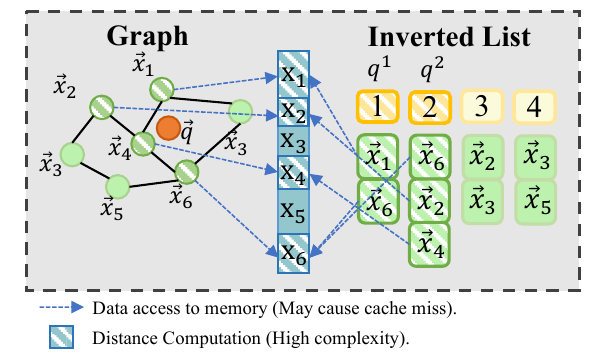}}
	\caption{\small The Bottleneck of the Graph Index and Inverted Index During Searching Process.}
	\label{fig:existing_solutions}\vspace{-2ex}
\end{figure}

\begin{example}
	Figure~\ref{fig:existing_solutions} illustrates a proximity graph and an inverted index constructed for $\vec{x}_{1}$ to $\vec{x}_{6}$.  
	Consider a query vector $\vec{q}$ with two non-zero entries $q^{1}$ and $q^{2}$.  
	In the proximity graph, when the search reaches $\vec{x}_4$, the algorithm computes distances between $\vec{q}$ and all its neighbors, sequentially accessing $x_{1}$, $x_{2}$, $x_{4}$, and $x_{6}$ from memory.  
	In the inverted index, the algorithm traverses the posting lists for dimensions $1$ and $2$, accessing $x_{1}$, $x_{6}$, $x_{2}$, and $x_{4}$.  
	Since vector access during search is essentially random, this incurs substantial random memory access overhead.
	Moreover, because $\lVert x \rVert$ varies across vectors, the distance computation between $\vec{q}$ and $\vec{x}$ has a time complexity of $O(\lVert q \rVert + \lVert x \rVert)$
\end{example}

\noindent \textit{\underline{Redundant Distance Computations}.}  
Sparse vectors incur high distance computation cost due to (i) the explicit lookup needed to identify the common dimensions $\Omega(\vec{x}, \vec{q})$ between a document $\vec{x}$ and a query $\vec{q}$, resulting in complexity $O(\lVert q \rVert + \lVert x \rVert)$, and (ii) the inability of existing algorithms to exploit SIMD acceleration for inner product computation.  
Profiling 6980 queries on the \textsc{MsMarco} dataset (1M vectors) using perf~\cite{perf} and VTune~\cite{vtune} shows that \textsc{PyANNS} spent 83.3\% of CPU cycles on distance calculation.

\noindent \textit{\underline{Random Memory Accesses.}} The inefficiency of memory access in existing algorithms can be attributed to two main factors.  First, \textsc{Seismic} organize similar data into the same partition. To improve accuracy, vectors are replicated across multiple partitions. This replication breaks the alignment between storage layout and query traversal order, preventing cache-friendly sequential access. During retrieval, the index returns candidate vector IDs, which incur random memory accesses to fetch data, leading to frequent cache misses. Moreover, $\lVert x \rVert$ varies across sparse vectors, requiring offset table lookups to locate each vector’s data.  
In our measurements, \textsc{Seismic} averaged 5168 random vector accesses per query (5.1\,MB), with an L3 cache miss rate of 67.68\%.

\section{Full Precision Inverted Index}
\label{sec:full-precision-inverted-index}

This section introduces full-precision \textsc{Sindi}, an inverted index designed for sparse vector retrieval. Its advantages are organized along three aspects: index structure, distance computation, and cache optimization. 

\begin{itemize}[leftmargin=*,labelsep=0.5em]
	\item In \S~\ref{subsec:value-storing-inverted-index}, \textsc{Sindi} constructs a value-based inverted index by storing both vector identifiers and their corresponding dimension values in posting lists. This design eliminates the redundant dimension-matching overhead present in traditional inverted indexes.
	\item In \S~\ref{subsec:efficient-distance-computation}, \textsc{Sindi} employs a two-phase search process involving \textit{product computation} and \textit{accumulation}. By using SIMD instructions in multiplication, it reduces query complexity from $O(\Vert q \Vert + \Vert x \Vert)$ to $O\Bigl(\frac{\Vert q \Vert}{s}\Bigr)$. \revisedtext{This yielding superior CPU utilization compared to state-of-the-art methods and improves query throughput}.
	
	\item In \S~\ref{subsec:cache-optimization}, to reduce cache misses, \textsc{Sindi} uses a Window Switch strategy, partitioning posting lists into fixed-size segments ($\lambda$) sharing a distance array. This minimizes memory overhead, and as shown theoretically and in Figure~\ref{fig:window}, an optimal $\lambda$ reduces access costs.
\end{itemize}

\subsection{Value-storing Inverted Index}
\label{subsec:value-storing-inverted-index}

Redundant inner product computations arise because identifying the common non-zero dimensions $\Omega(\vec{q}, \vec{x})$ between a query vector $\vec{q}$ and a document vector $\vec{x}$ requires scanning many irrelevant entries outside their intersection.  
We observe that the document identifiers retrieved from traversing an inverted list correspond precisely to the dimensions in $\Omega(\vec{x}, \vec{q})$.  
Therefore, when accessing a document $\vec{x}$ from the list of dimension~$j$, we can simultaneously retrieve its value $x^j$, thereby enabling direct computation of the inner product without incurring the overhead of finding $\Omega(\vec{q}, \vec{x})$.

\begin{figure}[t]
	\centering
	{\includegraphics[width=0.48\textwidth]{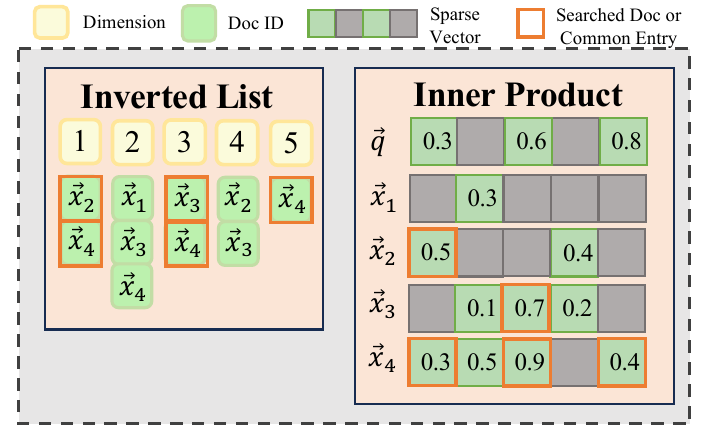}}
	\caption{\small Overlap of Inverted List Entries and Common Non-Zero Dimensions in Inner-Product Computation.}
	\label{fig:motivation}\vspace{-2ex}
\end{figure}

\begin{example}
	Figure~\ref{fig:motivation} illustrates the inverted lists constructed for vectors $x_1$ to $x_5$.  
	When a query $q$ arrives, it sequentially probes the inverted lists for dimensions~1, 3, and~5.
	In the base dataset, only dimensions in $\Omega(\vec{q}, \vec{x})$ need to be traversed during the inner product computation.  
	For example, although $x_4$ has a value in dimension~2, this dimension is not in $\Omega(\vec{q}, x_4)$ and thus never accessed.
	We further observe that the document identifiers retrieved from the inverted lists overlap exactly with those used to determine the common non-zero dimensions for the inner product.
	This implies that the products of these non-zero entries can be computed during document retrieval, thereby ensuring that only dimensions in $\Omega(\vec{q}, \vec{x})$ are involved in the inner product computation.
\end{example}

Inspired by these observations, we extend each inverted list $I_j$ to store not only the identifier $i$ of the vector $\vec{x}_i$, but also the value $x_i^j$ in dimension~$j$.  
In our notation, we simply write $x_i^j$ in $I_j$ to denote this stored value, with the subscript $i$ implicitly encoding the associated vector ID.  
This value-storing design eliminates the cost of explicitly locating $\Omega(\vec{x}_i, \vec{q})$ during inner product computation, as well as the additional random memory access needed to fetch $x_i^j$ from the original vector.

\subsection{Efficient Distance Computation}
\label{subsec:efficient-distance-computation}
With the value-storing inverted index, the inner product between a query $\vec{q}$ and candidate vector $\vec{x}$ can be computed without explicitly identifying $\Omega(\vec{q}, \vec{x})$.  
During search, \textsc{Sindi} organizes the computation into two stages: \emph{product computation} and \emph{accumulation}. In the first stage, as each relevant posting list $I_j$ is traversed, the products $q^j \times x_i^j$ are computed (using SIMD instructions when possible) and stored sequentially in a temporary array $P^j$.  
Here, $P^j$ holds the products for $I_j$, with each entry $P^j[t]$ corresponding to the $t$-th entry in $I_j$. Therefore, $P^j$ has the same length as $I_j$. In the second stage, the values in $P^j$ are aggregated into a preallocated distance array $A$ of length $\lVert \mathcal{D} \rVert$, where each entry $A[i]$ stores the accumulated score for vector $\vec{x}_i$.  
This arrangement enables $O(1)$ time per accumulation into $A$, and naturally exploits SIMD parallelism.  

The following example illustrates the detailed steps of the search procedure introduced above.

\begin{example} 
	Figure~\ref{fig:example_of_sindi} illustrates the search procedure.  
	Since $\lVert \mathcal{D} \rVert = 9$, the distance array $A$ is initialized with $\text{size}(A) = 9$ and all elements set to $0$.  
	The query $\vec{q}$ contains three non-zero components, $q^1$, $q^5$, and $q^8$, so only the inverted lists $I_1$, $I_5$, and $I_8$ are traversed.
	Consider $\vec{x}_4$ as an example.  
	From $I_1$, we obtain $x_4^1 = 6.8$, and the product $q^1 \times x_4^1 = 17.0$ (this multiplication can be SIMD-accelerated) is temporarily stored in $P[0]$.  
	Accumulating $P[0]$ into $A[4]$ gives $A[4] = 0 + 17.0 = 17.0$.  
	Similarly, from $I_5$ we compute $x_4^5 \times q^5 = 14.0$, store it in $P[2]$, and add it to obtain $A[4] = 31.0$.  
	The computation for $I_8$ proceeds analogously.  
	Finally, $A[4]$ becomes $36.1$, which equals $\delta(\vec{x}_4, \vec{q})$.
	Although $\vec{x}_4$ has another non-zero entry $x_4^2$, it is not in $\Omega(\vec{x}_4, \vec{q})$ and thus does not contribute to the inner product.  
	The same accumulation process applies to $\vec{x}_2$, $\vec{x}_3$, $\vec{x}_6$, and $\vec{x}_7$.  
	Eventually, $A[4]$ holds the largest value, so the nearest neighbor of $\vec{q}$ is $\vec{x}_4$.
\end{example}

\begin{figure}[!t]\centering
	{\includegraphics[width=0.4\textwidth]
		{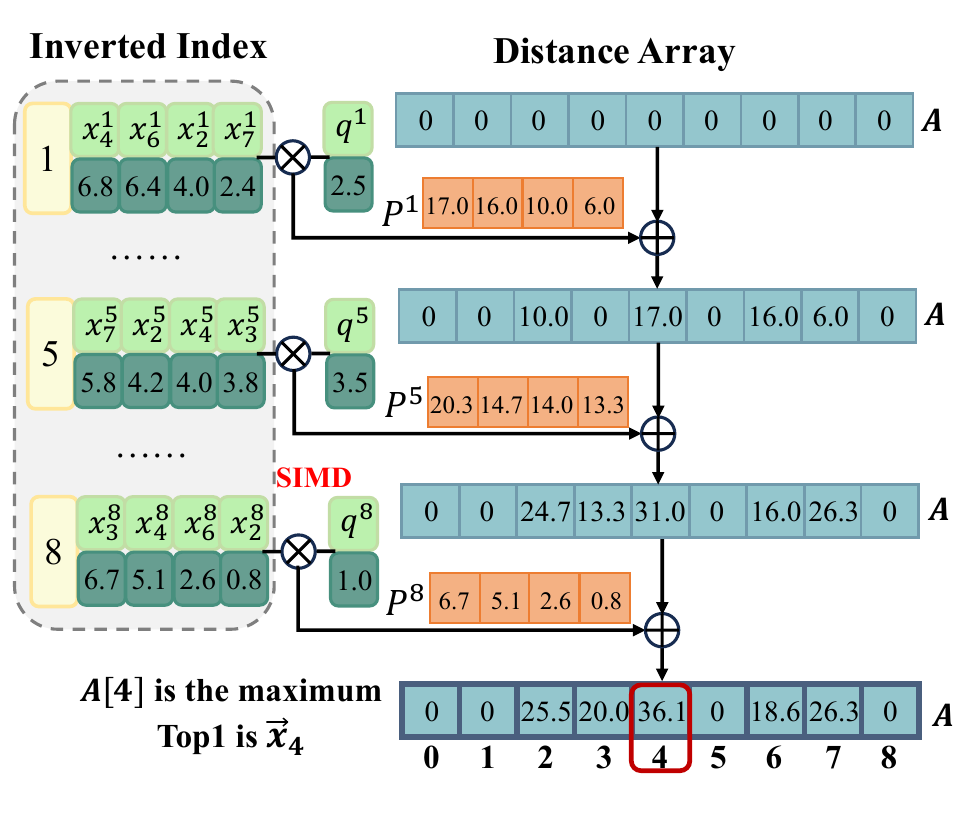}}\vspace{-1ex}
	\caption{\small An Example of \textsc{Sindi} and Query Process.}
	\vspace{-2ex} 
	\label{fig:example_of_sindi}
\end{figure}

Using SIMD instructions, \textsc{Sindi} processes the $j$-th inverted list $I_j$ in batches, multiplying $q^j$ with each $x_i^j$ it contains and writing the results sequentially into $P^j$.  
This approach maximizes CPU utilization and reduces the complexity of the inner product computation from $O(\lVert q \rVert + \lVert x \rVert)$ to $O\!\left(\frac{\lVert q \rVert}{s}\right)$, where $s$ denotes the number of elements processed per SIMD operation.  
The complexity of \textsc{Sindi}'s distance computation is derived as follows:

\begin{theorem}[Amortized Time Complexity of \textsc{Sindi} Distance Computation]
	\label{thm:sindi-complexity}
	Let $\mathcal{D}$ be the dataset, $I$ its inverted index, and $I_j$ the posting list for dimension~$j$.  
	Let $\mathcal{I}^j = \{ x_i^j \mid x_i^j \neq 0 \}$ denote the set of non-zero entries in $I_j$.
	
	Given a query vector $\vec{q}$, let $\mathcal{J} = \{\, j \mid q^j \neq 0 \,\}$ be the set of query dimensions with non-zero entries.  
	Let $\mathcal{X} = \{\, \vec{x}_i \mid \exists j \in \mathcal{J},\; x_i^j \neq 0 \,\}$ denote the set of candidate vectors retrieved by $\vec{q}$.
	
	Let $s$ be the number of elements that can be processed simultaneously using SIMD instructions.  
	Then the amortized per-vector time complexity of computing the inner product between $\vec{q}$ and all $\vec{x}_i \in \mathcal{X}$ is $\Theta\!\left( \frac{\lVert q \rVert}{s} \right)$.
\end{theorem}

\begin{proof}
	The total number of non-zero entries accessed in all posting lists for $\mathcal{J}$ is
	$
	\sum_{j \in \mathcal{J}} \lVert \mathcal{I}^j \rVert.
	$
	Since $s$ entries can be processed in parallel using SIMD, the total time is
	$
	T_{\text{total}} = \sum_{j \in \mathcal{J}} \frac{\lVert \mathcal{I}^j \rVert}{s}.
	$
	Amortizing over all $\lVert \mathcal{X} \rVert$ candidates gives
	$
	T = \frac{T_{\text{total}}}{\lVert \mathcal{X} \rVert}
	= \frac{\sum_{j \in \mathcal{J}} \lVert \mathcal{I}^j \rVert}{s \cdot \lVert \mathcal{X} \rVert}.
	$
	For any $\vec{x}_i \in \mathcal{X}$, we have $q \cap x_i \subseteq q$, implying that $\Vert \Omega(\vec{x}_i, \vec{q}) \Vert \le \Vert q \Vert$. Therefore:
	$
	T \le \frac{\sum_{\vec{x}_i \in \mathcal{X}} \Vert q \Vert}{s \cdot \Vert \mathcal{X} \Vert}
	= \frac{\Vert q \Vert \cdot \Vert \mathcal{X} \Vert}{s \cdot \Vert \mathcal{X} \Vert}
	= \frac{\Vert q \Vert}{s}.
	$
	Substituting into $T$ yields
	$
	T \le \frac{\sum_{\vec{x}_i \in \mathcal{X}} \lVert q \rVert}{s \cdot \lVert \mathcal{X} \rVert}
	= \frac{\lVert q \rVert}{s}.
	$
	Hence, the amortized per-vector complexity is
	$
	\Theta\!\left( \frac{\lVert q \rVert}{s} \right).
	$
\end{proof}

\subsection{Cache Optimization}
\label{subsec:cache-optimization}

When the dataset size $\mathcal{D}$ reaches the million scale, the distance array $A$ becomes correspondingly large.  
Random accesses to such a long array cause frequent cache misses, making query performance highly memory-bound.  
In addition, allocating a full-length distance array for every query incurs substantial memory overhead.  

To address these issues, \textsc{Sindi} adopts the \emph{Window Switch} strategy, which partitions the vector ID space into fixed-size windows and restricts each query to accessing IDs within a single window at a time.
With a shorter distance array per window, the accessed entries are located within a much more compact memory region. This substantially improves spatial locality, and the resulting access pattern closely resembles a sequential scan, enabling hardware prefetching and reducing cache misses.

\subsubsection{Window Switch}
During index construction, \textsc{Sindi} partitions the dataset $\mathcal{D}$ into contiguous ID segments, referred to as \emph{windows}.  
The window size is denoted by $\lambda$ ($0 < \lambda \le \lVert \mathcal{D} \rVert$), and the total number of windows is $\sigma = \left\lceil \frac{\lVert \mathcal{D} \rVert}{\lambda} \right\rceil$.  
The $w$-th window contains vectors from $\vec{x}_{w\lambda}$ to $\vec{x}_{(w+1)\lambda - 1}$, and the window index of vector $\vec{x}_i$ is $\left\lfloor \tfrac{i}{\lambda} \right\rfloor$.  
Each inverted list $I_j$ is partitioned in the same way, so every list has $\sigma$ windows.  
We denote the $w$-th window of the $j$-th inverted list by $I_{j,w}$ ($0 \le w < \sigma$), whose entries are the non-zero $x_i^j$ in dimension~$j$ for vectors in that ID range. Each $x_i^j$ implicitly carries the identifier $i$ through its subscript.

At query time, the length of the distance array $A$ is set to the window size $\lambda$, and all windows share this same $A$.  
Within a window, each vector $\vec{x}_i$ is mapped to a unique entry $A[i \bmod \lambda]$.

The search procedure for the $w$-th window proceeds in two steps:  

\textbf{(1) Inner product computation.}  
For each scanned posting list $I_{j,w}$, compute the products $q^j \times x_i^j$ for all its entries, using SIMD instructions when possible.  
These products are written sequentially into a temporary array $P^j$, which has the same length as $I_{j,w}$ and is index-aligned with it — i.e., $P^j[t]$ stores the product for the $t$-th posting in $I_{j,w}$.  
The accumulation stage then adds each $P^j[t]$ into the corresponding $A[i \bmod \lambda]$ using $O(1)$ time per update.

\textbf{(2) Heap update.}  
After processing all query dimensions for the current window, $A$ holds the final scores for that window's candidates.  
Scan $A$ to insert the top-scoring vectors into a minimum heap $H$, which maintains the vector IDs and scores for the results to be returned.  
Here, each $A[t]$ corresponds to vector $\vec{x}_{t + \lambda \times w}$, so the global ID can be recovered from the local index~$t$.

Note that \textit{Window Switch} only changes the order of list entries scanned, without altering the number of arithmetic operations.  
Thus, the time complexity of distance computation remains $O\!\left( \frac{\lVert q \rVert}{s} \right)$.

\begin{algorithm}[t]\small
	\DontPrintSemicolon
	\KwIn{A sparse dataset $\mathcal{D}$ and dimension $d$, window size $\lambda$}
	\KwOut{Inverted list $I$}
	\For{$j \in \{0,...,d-1\}$}{
		$\mathcal{X} \leftarrow \{\, \vec{x}_i \in \mathcal{D} \mid x_i^j \neq 0 \,\}$; \\
		\ForEach{$\vec{x}_i \in \mathcal{X}$}{
			$w \leftarrow \lfloor \frac{i}{\lambda} \rfloor$\;
			$I_{j,w}.append(x_i^j)$\;
		}
		
	}
	\Return{$I$}\;
	\caption{\textsc{PreciseSindiConstruction}}
	\label{algo:full_precision_ivf_construction}
\end{algorithm}

\subsubsection{Construction and Search}
The construction process of the full-precision \textsc{Sindi} index with \textit{Window Switch} is shown in Algorithm~\ref{algo:full_precision_ivf_construction}.  
Given a sparse vector dataset $\mathcal{D}$ of dimension $d$, the algorithm iterates over each dimension $j$ (Line~1).  
For each $j$, it collects all vectors $\vec{x}_i \in \mathcal{D}$ having a non-zero entry $x_i^j$ into a temporary set $\mathcal{X}$ (Line~2).  
These vectors are then appended to the corresponding window $I_{j,w}$ of $I_j$ based on their IDs, where $w = \lfloor i / \lambda \rfloor$ (Lines~3--5).  
After processing all dimensions, the inverted index $I$ is returned (Line~6).  
The time complexity of the construction process is $O(\lVert\mathcal{D}\rVert \lVert\bar{x}\rVert)$, where $\lVert \bar{x} \rVert = \frac{\sum_{\vec{x}_i \in \mathcal{D}} \lVert x_i \rVert}{\lVert\mathcal{D}\rVert}$ denotes the average number of nonzero entries per vector.

The search process for the full-precision \textsc{Sindi} index is summarized in Algorithm~\ref{algo:full_precision_ivf_search}, and consists of three stages: product computation, accumulation, and heap update.  
Given a query $\vec{q}$, inverted index $I$, and target recall size $k$, a distance array $A$ of length $\lambda$ is initialized to zeros (Lines~1–2) and an empty min-heap $H$ is created (Line~3).  
The outer loop traverses all windows $w \in \{0, \dots, \sigma-1\}$ (Line~4).  
For each non-zero query component $q^j$ (Line~5), SIMD-based batched multiplication is performed with all $x_i^j$ in $I_{j,w}$, and the results are stored sequentially into a temporary product array $P^j$ aligned with $I_{j,w}$ (Line~6).  
Each $x_i^j$ is then retrieved (Lines~7–8), its mapped index $m = i \bmod \lambda$ computed (Line~9), and $P^j[t]$ accumulated into $A[m]$ (Line~10).  
After all $q^j$ for the current window are processed, the heap update stage begins (Line~12): each $A[m]$ that exceeds the heap minimum or when $H$ has fewer than $k$ entries is inserted into $H$ with its global ID $(m + \lambda w)$ and score $A[m]$ (Lines~13–14), removing the smallest if size exceeds $k$ (Lines~15–16).  
$A[m]$ is then reset to zero for the next window (Line~17).  
When all windows are processed, $H$ contains up to $k$ vector IDs with their full-precision distances to $\vec{q}$, which are returned as the final result (Line~20).

\noindent\textbf{Complexity.} Let $l = \frac{\sum_{q^j \in q} |I_j|}{\lVert q \rVert}$ denote the average number of nonzero entries in the traversed lists.  
With \textit{Window Switch}, the total number of postings visited is still $\lVert q \rVert \, l$, so the time complexity of a full-precision query is $O\!\left(\frac{\lVert q \rVert \, l}{s}\right)$, where $s$ is the SIMD batch size. Hence, the computational cost is independent of the window size $\lambda$.

\begin{algorithm}[t]\small
	\DontPrintSemicolon
	\KwIn{Query $\vec{q}$, an inverted list $I$, and $k$}
	\KwOut{At most $k$ points in $\mathcal{D}$}
	\For{$m \in \{0,...,\lambda-1\}$}{
		$A[m] \leftarrow 0$
	}
	$H \leftarrow$ empty min-heap \\
	
	\For{$w \in \{0,...,\sigma-1\}$ }{
		\ForEach{$q^j \in q$}{
			$P^j \leftarrow$ \textsc{Simd}Product($q^j, \, I_{j,w}$); \\
			\For{$t \in \{0, \dots, I_{j,w}.size() - 1\}$}{
				$x_i^j \leftarrow I_{j,w}[t]$; \\
				$m \leftarrow i \bmod \lambda$; \\
				$A[m] \leftarrow A[m] + P^j[t]$; \\
			}
		}
		\For{$m \in \{0,...,\lambda-1\}$}{
			\uIf{$A[m] > H.min()$ \textbf{or} $H.len() < k$}{
				$H.insert(m+\lambda \times w,A[m])$ \\
			}
			\uIf{$H.len() > k$}{
				$H.pop()$ \\
			}
			$A[m] \leftarrow 0$ \\
		}
	}
	\Return{$H$}\;
	\caption{\textsc{PreciseSindiSearch}}
	\label{algo:full_precision_ivf_search}
\end{algorithm}

\subsection{Analysis of Window Size's Impact on Performance}

While the \textit{Window Switch} strategy does not alter the overall computational complexity, two types of memory-access costs are sensitive to the window size~$\lambda$:

\begin{itemize}[leftmargin=*,labelsep=0.5em]
	\item \textbf{Random-access cost to the distance array.}  
	During accumulation, each partial score $P^j[t]$ is added to $A[i \bmod \lambda]$, producing a random write pattern over $A$.  
	When $\lambda$ decreases, the length of $A$ becomes smaller, more of it fits in cache, and this random-access cost decreases due to fewer cache misses.
	
	\item \textbf{Cache eviction to sub-list cost when switching windows.}  
	Under \textit{Window Switch}, the search iterates over multiple posting sub-lists $I_{j,w}$ for different dimensions $j$ within the same window $w$.  
	Switching from one dimension's sub-list to another may evict previously cached sub-list data from memory.  
	When $\lambda$ decreases, the number of windows $\sigma = \frac{\lVert\mathcal{D}\rVert}{\lambda}$ increases, leading to more frequent loading and eviction of these sub-lists, and hence increasing this cost.
\end{itemize}

Selecting $\lambda$ thus requires balancing the reduced random-access cost to the distance array against the increased cache-eviction cost to inverted sub-lists. The following example illustrates this trade-off:

\revisedtext{
	\begin{example}  
		Figure~\ref{fig:window} reports experimental results for the full-precision \textsc{Sindi} on the SPLADE-FULL and AntSparse-10M datasets.  
		For each dataset, we executed queries under different window sizes $\lambda$ and measured the QPS.  
		We also used the Intel VTune Profiler~\cite{vtune} to record \textit{memory bound} metrics for two types of memory accesses: random accesses to the distance array (arising from accumulation writes) and cache eviction to sub-lists (occurring when switching dimensions within a window).  
		Here, \textit{memory bound} denotes the percentage of execution time stalled due to memory accesses.  
		For the SPLADE-FULL dataset, as $\lambda$ increases from $1$K to $8.8$M, the memory bound from distance array accesses increases monotonically, whereas that from sub-list cache evictions decreases monotonically.  
		The total memory-bound latency reaches its minimum near $\lambda\approx 150$K, corresponding to the highest query throughput.  
		The AntSparse-10M dataset exhibits the same trend, confirming the existence of an optimal window size.
	\end{example}
}

Based on this, the memory-access latency for queries can be modeled by a double power-law~\cite{bertsekas2021data, hennessy2011computer}:
\begin{equation}
	T_{\text{mem}}(\lambda) = c_{rand} \lambda^{+\alpha} + c_{evict} \lambda^{-\beta} + C_0,
	\label{eq:window-mem}
\end{equation}
where:
\begin{itemize}[leftmargin=*,labelsep=0.5em]
	\item $\lambda$ is the window size;
	\item $c_{rand}\lambda^{+\alpha}$ models the increasing cost of random accesses to the distance array as $\lambda$ grows;
	\item $c_{evict}\lambda^{-\beta}$ models the decreasing cost of cache eviction to sub-lists with larger $\lambda$;
	\item $C_0$ is the baseline memory-access cost unrelated to $\lambda$.
\end{itemize}

This function reaches its minimum at $\lambda^* = \left(\frac{c_{evict}\beta}{c_{rand}\alpha}\right)^{\frac{1}{\alpha+\beta}}$. When $\lambda \ll \lambda^*$, $T_{\text{mem}}$ is dominated by the sub-list eviction term and decreases as $\lambda$ increases; when $\lambda \gg \lambda^*$, the distance-array term dominates and $T_{\text{mem}}$ grows with $\lambda$. The optimum $\lambda^*$ occurs when these two terms are balanced.

\revisedtext{
	\begin{example}
		Figure~\ref{fig:window} presents the theoretical QPS curves, where parameters $(\alpha,\beta)$ are estimated via log--log regression~\cite{draper1998applied} and $(c_{rand},c_{evict},C_0)$ via least-squares fitting~\cite{lawson1995solving} of the double power-law model. For SPLADE-FULL, the model predicts an optimal $\lambda^* \approx 1.2 \times 10^5$, while for AntSparse-10M, it predicts $\lambda^* \approx 5.1 \times 10^4$. Both predictions align with the robust high-throughput interval $\lambda \in [50,000, 100,000]$. This broad plateau confirms that within this range, performance is stable, rendering precise point-wise tuning unnecessary despite minor hardware-level variability. More details are shown in Appendix B.
	\end{example}
}

\begin{figure}[!t]
	\centering
	\subfigure{
		\scalebox{0.5}[0.5]{\includegraphics{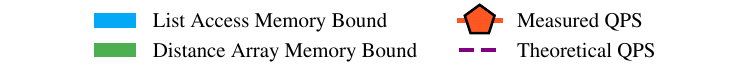}}}\hfill\\\vspace{-3ex}
	\addtocounter{subfigure}{-1}
	\subfigure[][{\scriptsize SPLADE-FULL}]{
		\scalebox{0.15}[0.15]{\includegraphics{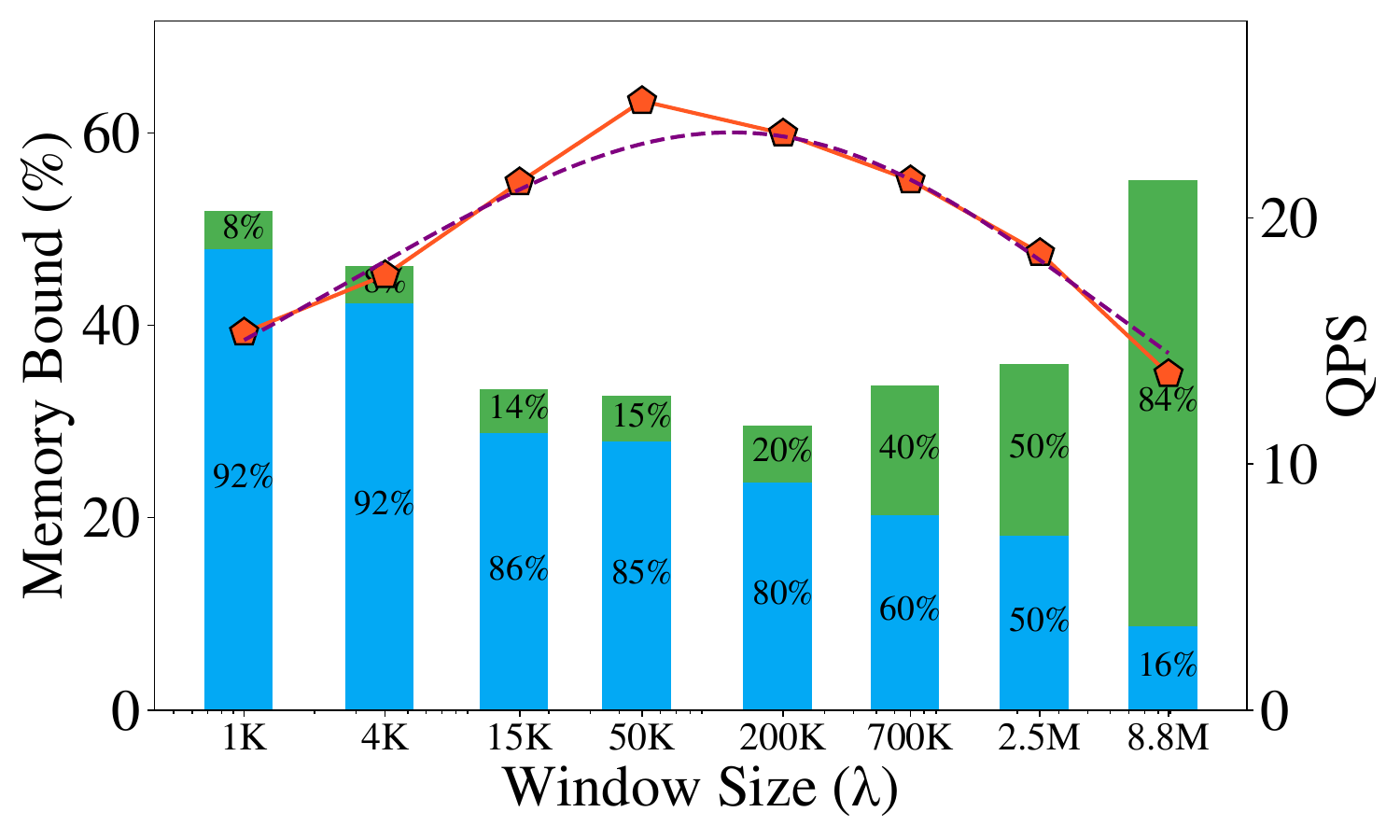}}
		\label{subfig:splade-full-window}}\vspace{-2ex}
	\hfill
	\subfigure[][{\scriptsize AntSparse-10M}]{
		\scalebox{0.15}[0.15]{\includegraphics{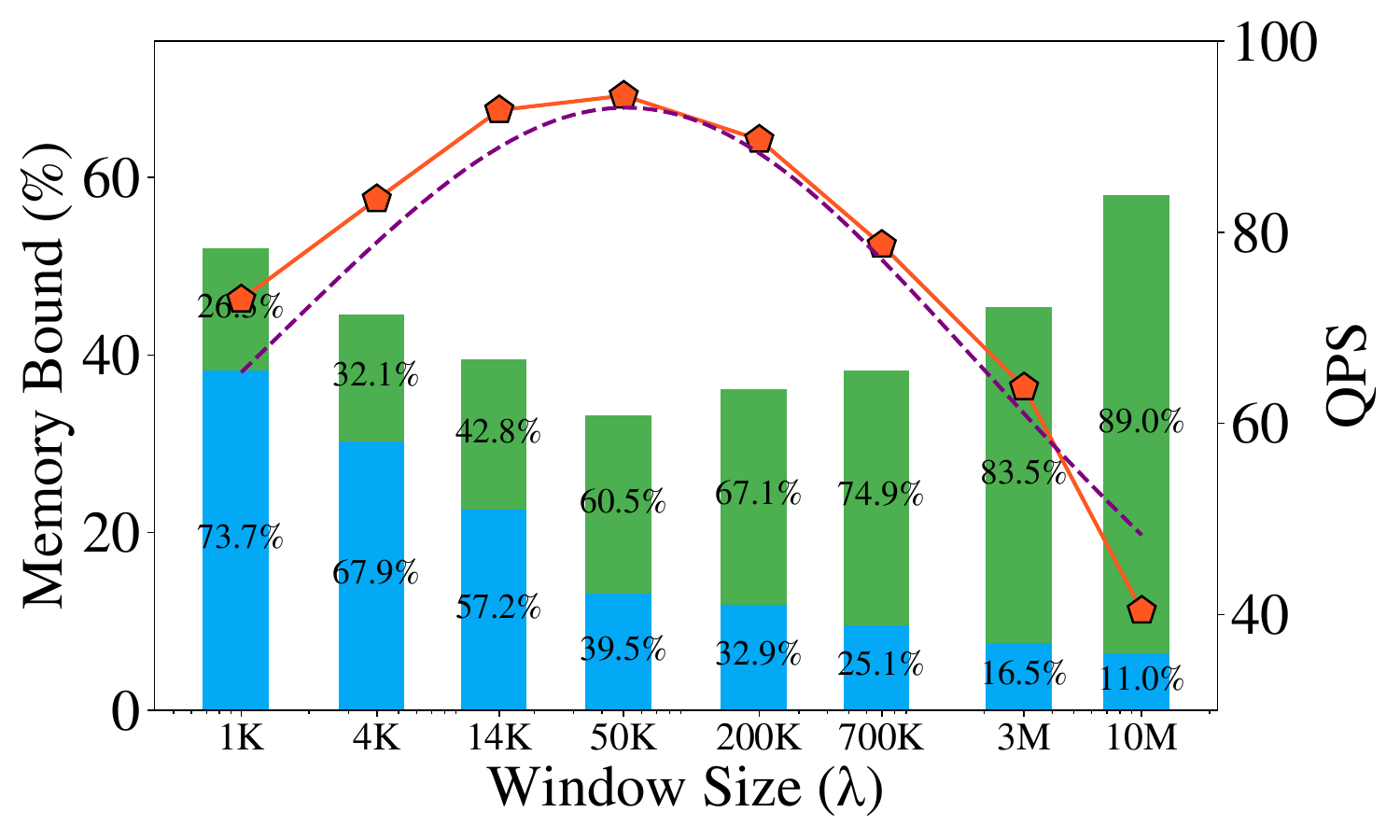}}
		\label{subfig:antsparse-10m-window}}
	\vspace{1ex} 
	
	\caption{\revisedtext{\small Impact of Window Size on Query Throughput and Memory Accesses.}}
	\vspace{-2ex} 
	\label{fig:window}
\end{figure}

\section{Approximate Inverted Index}
\label{sec:approximate-inverted-index}

This section focuses on optimizing the query process through pruning and reordering. 
In our framework, pruning serves as a form of coarse retrieval, reducing the number of non-zero entries or inverted-list lengths to quickly generate a compact candidate set. 
\textit{Reordering} then plays the role of fine-grained ranking by computing exact inner products for this candidate set. Combining these two stages significantly improves query throughput while keeping the loss in recall negligible.

\subsection{Pruning Strategies}
\label{subsec:pruning}
A key property of sparse vectors is that a small number of high-valued non-zero entries can preserve the majority of the vector's overall information content~\cite{splade}. 
This distribution pattern typically results from the training objectives of sparse representation models. 
For instance, \textsc{SPLADE} often concentrates the most informative components into a limited set of high-weight dimensions. Conversely, many low-valued non-zero entries correspond to common stopwords (e.g., ``is'', ``the''), which can be safely removed with minimal impact on retrieval quality.

Let $\vec{x}_i'$ denote the pruned version of document $\vec{x}_i$, 
and $\vec{q}'$ the pruned version of query $\vec{q}$.  
Let $l$ be the average posting-list length before pruning and $l'$ the average posting–list length after pruning.  
The reduction in computational cost achieved by pruning is
$
\Vert q \Vert \cdot l \;-\; \Vert q' \Vert \cdot l',
$

For a given pruning operator $\phi$, we define the \emph{inner product error} for document $\vec{x}_i$ as
$
e^{(\phi)}_i = \delta(\vec{x}_i, \vec{q}) - 
\delta\!\bigl(\phi(\vec{x}_i), \phi(\vec{q})\bigr),
$
where $\delta(\cdot,\cdot)$ denotes the exact inner product and $\phi(\cdot)$ applies the pruning transformation.  
The total inner product error over the dataset $\mathcal{D}$ is then
$
\varepsilon^{(\phi)} = \sum_{\vec{x}_i \in \mathcal{D}} e^{(\phi)}_i
$

Smaller $\Vert x'_i \Vert$ yields higher query throughput, typically incurs a larger inner product error.
Therefore, pruning must be designed with a trade-off between efficiency and accuracy. 
The following example shows that it is possible to retain a subset of high–value non-zero entries while incurring merely a small loss in inner product accuracy.

\begin{example}
	Figure~\ref{subfig:ip_error} reports the inner product error measured on a $100$K-scale dataset when retaining different proportions of the largest non-zero entries from both document and query vectors. 
	The results show that the error drops rapidly as the retaining ratio increases from 0.1 to 0.3, and becomes almost negligible once the ratio exceeds 0.5. 
	This indicates a saturation effect, where further increasing the retaining ratio yields minimal additional gains in accuracy.
\end{example}

As discussed in Section~\ref{sec:full-precision-inverted-index}, 
for full-precision \textsc{Sindi} the upper bound for computing the inner product between $\vec{x}_i$ and $\vec{q}$ is $\varTheta\!\left( \frac{\Vert q \Vert}{s} \right)$. 
For a given $\vec{q}$, the overall query complexity is $O\!\left( \frac{\Vert q \Vert \, l}{s} \right)$, where $l$ is the average number of inverted lists traversed. 
Thus, reducing query latency requires decreasing $l$, $\Vert x \Vert$, and $\Vert q \Vert$. 
This can be achieved via three approaches: list pruning, document pruning, and query pruning. 
List and document pruning are performed during index construction, while query pruning is applied at query time. 
In this work, we focus on the construction stage, since query pruning and document pruning are both forms of vector pruning. 
We compare three strategies—\textit{List Pruning (LP)}, \textit{Vector Number Pruning (VNP)}, and \textit{Mass Ratio Pruning (MRP)}—and analyze their respective strengths and weaknesses.

\revisedtext{
	As discussed in Section~\ref{sec:full-precision-inverted-index}, 
	for full-precision \textsc{Sindi} the upper bound for computing the inner product between $\vec{x}_i$ and $\vec{q}$ is $\varTheta\!\left( \frac{\Vert q \Vert}{s} \right)$. 
	For a given $\vec{q}$, the overall query complexity is $O\!\left( \frac{\Vert q \Vert \, l}{s} \right)$, where $l$ is the average number of inverted lists traversed. 
	Thus, reducing query latency requires decreasing $l$, $\Vert x \Vert$, and $\Vert q \Vert$. This can be achieved via retaining entries with the largest magnitude. 
	While magnitude-based pruning is an established concept—exemplified by \textit{List Pruning (LP)} in \textsc{Seismic}~\cite{seismic} and utilized in \textsc{Bmp}~\cite{bmp}—\textsc{Sindi}'s \textit{Mass Ratio Pruning (MRP)} introduces a distinct objective. Unlike prior methods that employ pruning primarily to shrink the search space (filtering), MRP is architected to minimize \textbf{inner product error} via an adaptive cumulative mass ratio, ensuring high-quality candidates for reordering. In the following, we systematically compare LP, \textit{Vector Number Pruning (VNP)}, and MRP to analyze their respective trade-offs.
}

\begin{figure}[!t]
	\centering
	\subfigure[][{\scriptsize Inner Product Error}]{
		\scalebox{0.17}[0.17]{\includegraphics{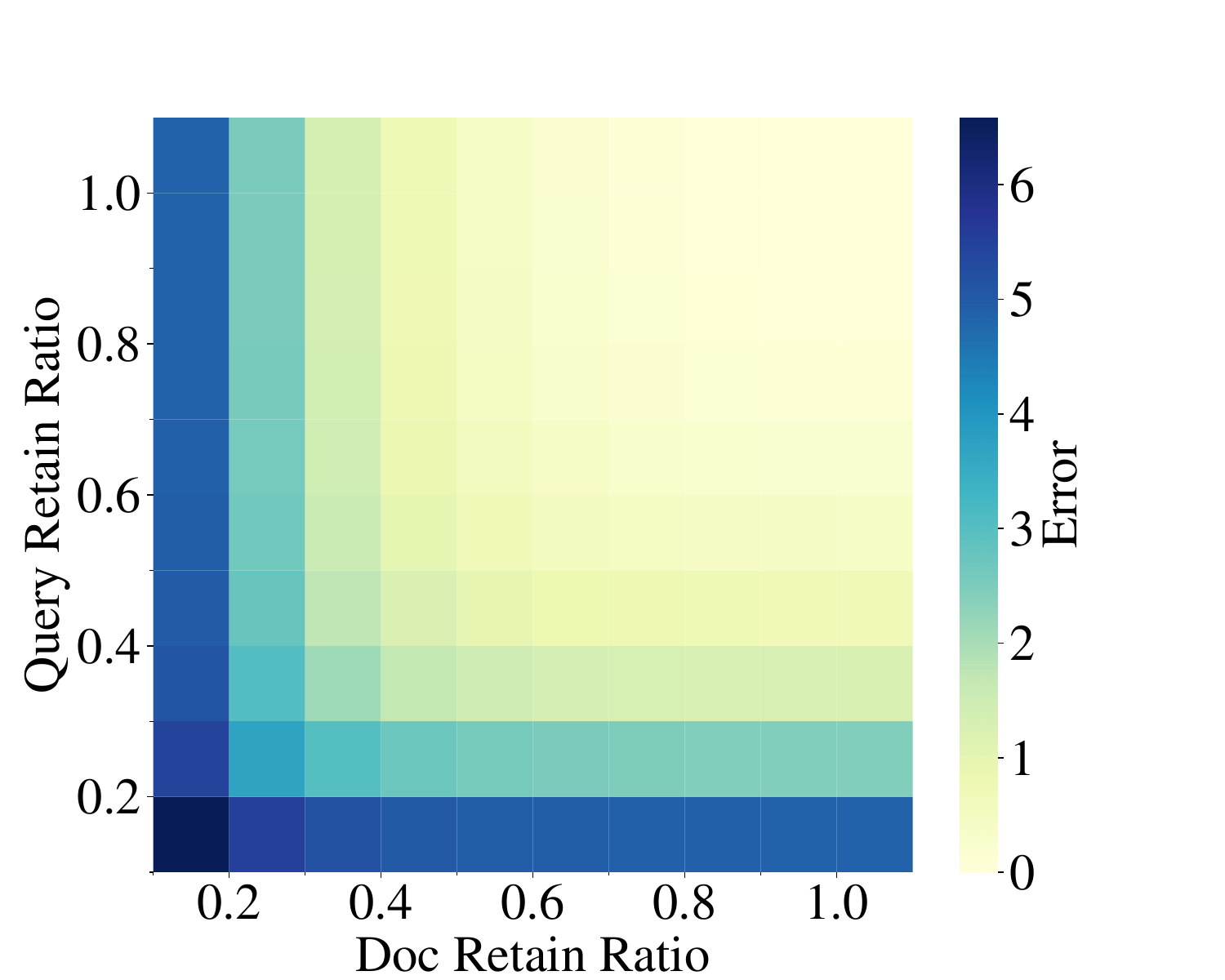}}
		\label{subfig:ip_error}}
	\subfigure[][{\scriptsize Recall Comparision}]{
		\scalebox{0.17}[0.17]{\includegraphics{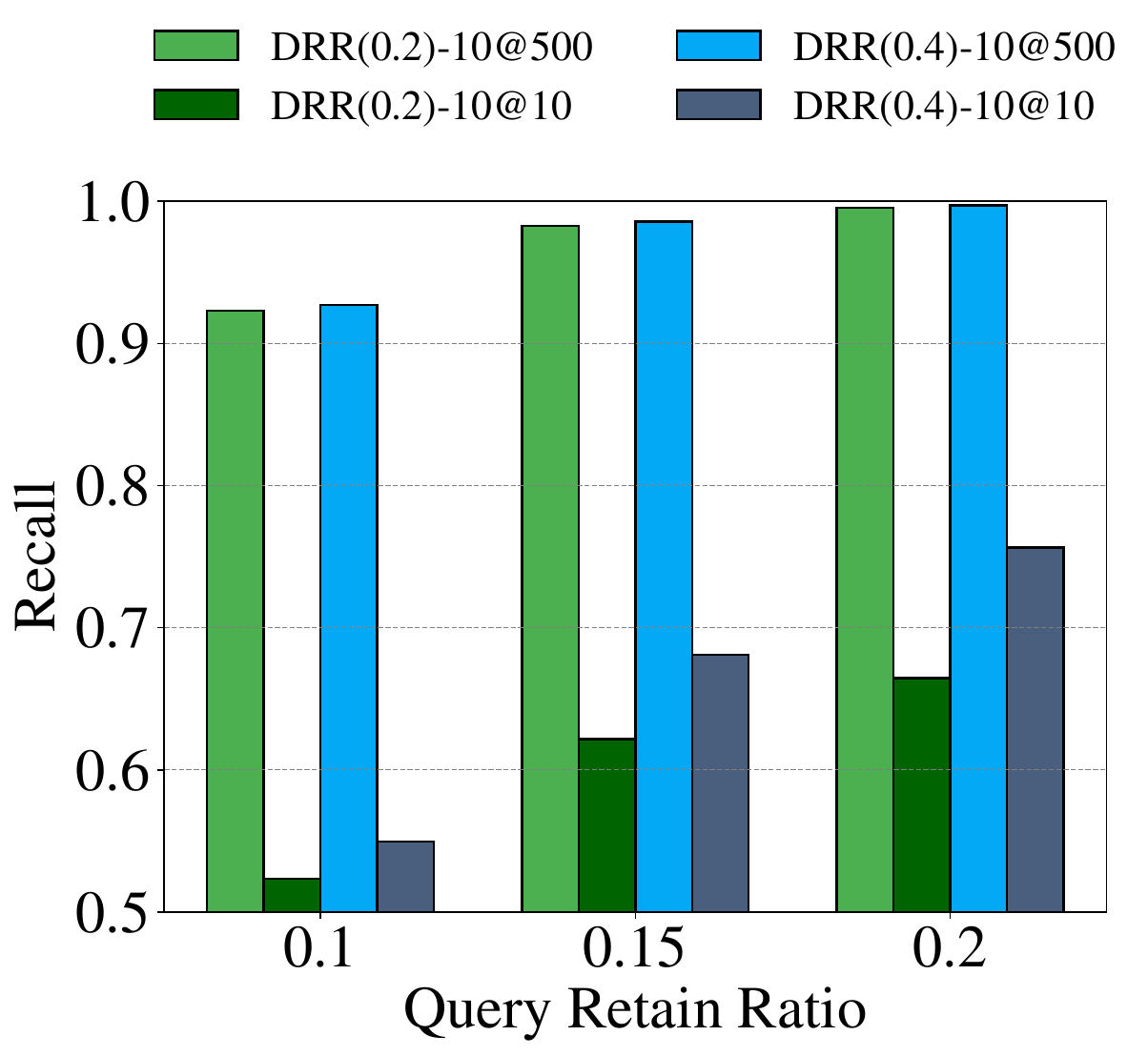}}
		\label{subfig:recall_cmp}}
	\vspace{-1ex} 
	
	\caption{\small Intuition of Pruning and Reorder}
	\vspace{-2ex} 
	\label{fig:recall}
\end{figure}

\noindent\textbf{\textit{List Pruning (LP)}.}
LP operates at the inverted–list level: for each dimension $j$, it retains only the non–zero entries with the largest absolute values in $I_j$, limiting the list length to $l'$.  
Since the size of $I_j$ varies across dimensions, some high–value $|x_i^j|$ entries in longer lists may be removed, while lower–value entries in shorter lists may be kept.  
After pruning, each document vector $\vec{x}_i$ becomes $\phi_{\mathrm{LP}}(\vec{x}_i)$, containing only the coordinates that survive the list truncation.

\noindent\textbf{\textit{Vector Number Pruning (VNP)}.} 
VNP applies the pruning operator $\phi_{\mathrm{VNP}}$ at the vector level. 
For each document vector $\vec{x}_i$, $\phi_{\mathrm{VNP}}(\vec{x}_i)$ retains the $vn$ non–zero entries with the largest absolute values, ensuring $\Vert \phi_{\mathrm{VNP}}(\vec{x}_i) \Vert = vn$.  
Since $\Vert \vec{x}_i \Vert$ varies across vectors, high–value $\vert x_i^j \vert$ entries that contribute substantially to the inner product may still be removed under this fixed–size scheme.

\noindent\textbf{\textit{Mass Ratio Pruning (MRP)}.} 
MRP applies the pruning operator $\phi_{\mathrm{MRP}}$ based on the cumulative sum of the absolute values of a vector’s non-zero entries.  
For each document vector $\vec{x}_i$, $\phi_{\mathrm{MRP}}(\vec{x}_i)$ ranks all non-zero entries in descending order of absolute value and retains the shortest prefix whose cumulative sum reaches a fraction $\alpha$ of the vector’s total mass.  
This adaptive scheme removes low–value entries that contribute little to the inner product, while allowing vectors with different value distributions to keep variable numbers of entries, thereby reducing inverted–list size without enforcing a uniform length limit.

To formally introduce MRP, we first define the mass of a vector.

\begin{definition}[Mass of a Vector]
	\label{def:mass}
	Let $\vec{x}\in\mathbb{R}^{d}$ be a vector. The \emph{mass} of $\vec{x}$ is defined as the sum of the absolute values of $x$'s non-zero entries:
	$
	mass(\vec{x})
	\;=\;
	\sum_{x^j \in x} \lvert x^j \rvert .
	$
\end{definition}

We define the vector obtained after Mass Ratio Pruning as the $\alpha$-mass subvector.

\revisedtext{
	\begin{definition}[$\alpha$-Mass Subvector]
		\label{def:alpha-mass-subvector}
		Let $\vec{x} \in \mathbb{R}^{d}$ and let $\pi$ be a permutation
		that orders the non-zero entries of $\vec{x}$ by non–increasing absolute value, 
		i.e., $\vert x^{\pi_{j}} \vert \ge \vert x^{\pi_{j+1}} \vert$.  
		For a constant $\alpha \in (0,1]$, let $1 \le r \le \Vert x \Vert$ be the smallest
		integer satisfying
		$
		\sum_{j=1}^{r} \vert x^{\pi_{j}} \vert
		\;\ge\;
		\alpha \times mass(\vec{x}).
		$
		The $\alpha$-mass subvector \cite{seismic}, denoted $\alpha\text{-}mass(\vec{x})$, 
		is the vector whose non-zero entries are $\{ x^{\pi_{j}} \}_{j=1}^{r}$.
	\end{definition}
}

\begin{figure}[t!]\centering
	{\includegraphics[width=0.48\textwidth]
		{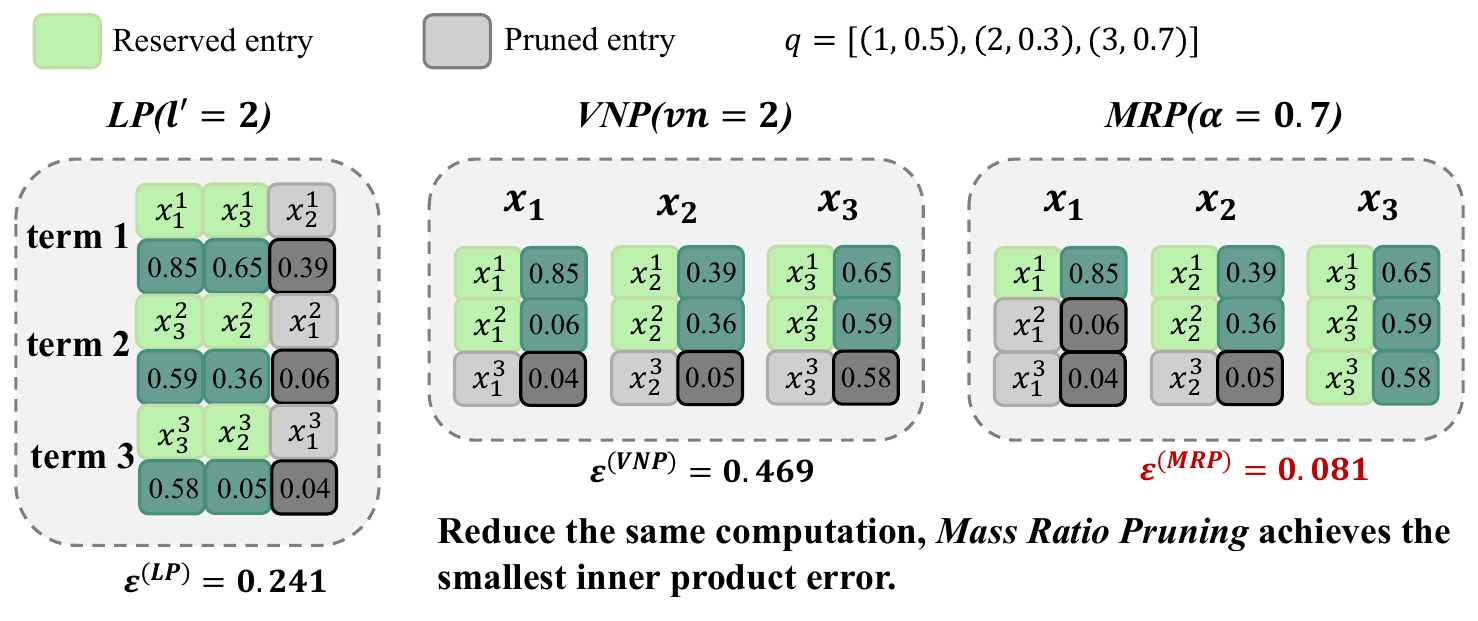}}
	\caption{\small An Example of \textit{List Pruning}, \textit{Vector Number Pruning} and \textit{Mass Ratio Pruning}.}
	\vspace{-2ex} 
	\label{fig:example_of_pruning}
\end{figure}

\begin{example}
	Figure~\ref{fig:example_of_pruning} illustrates the three pruning methods applied to sparse vectors $\vec{x}_1$, $\vec{x}_2$, and $\vec{x}_3$. 
	\textit{List Pruning} prunes each inverted list to size $l' = 2$, 
	\textit{Vector Number Pruning} retains the top $vn=2$ entries of each vector, 
	and \textit{Mass Ratio Pruning} prunes each $\vec{x}_i$ to its $\alpha\text{-}mass(\vec{x}_i)$ with $\alpha=0.7$.  
	The figure shows:  
	(i) all three strategies result in the same reduction in computation, 
	$\Vert q \Vert l - \Vert q' \Vert l' = 9 - 6 = 3$;  
	(ii) \textit{Mass Ratio Pruning} yields the smallest inner product error.  
	This is because \textit{List Pruning} cannot retain the larger value $x_2^1$ when each list is limited to two vectors,  
	and \textit{Vector Number Pruning} drops $x_3^3$.  
	In contrast, \textit{Mass Ratio Pruning} prioritizes high–value entries, thereby minimizing error.
\end{example}


\begin{algorithm}[t]\small
	\DontPrintSemicolon
	\KwIn{Sparse dataset $\mathcal{D}$ of dimension $d$; window size $\lambda$; pruning ratio $\alpha$}
	\KwOut{Inverted index $I$}
	
	$\mathcal{D}' \leftarrow \emptyset$
	
	\ForEach{$\vec{x}_i \in \mathcal{D}$}{
		$\vec{x}'_i \leftarrow \alpha\text{-}mass(\vec{x}_i)$ \\
		$\mathcal{D'} \leftarrow \mathcal{D'} \, \cup  \vec{x}'_i$
	}
	
	$I =$ \textsc{PreciseSindiConstruction}$(\mathcal{D’}, d, \lambda)$ \\
	\Return{$I$ and $\mathcal{D}$}\;
	\caption{\textsc{ApproximateSindiConstruction}}
	\label{algo:approximate_ivf_construction}
\end{algorithm}

Algorithm~\ref{algo:approximate_ivf_construction} outlines the construction of the approximate \textsc{Sindi} index. 
Given a sparse dataset $\mathcal{D}$ of maximum dimension $d$, window size $\lambda$, and pruning ratio $\alpha$, 
the algorithm first initializes an empty set $\mathcal{D}'$ to store pruned vectors (Line~1).  
For each vector $\vec{x}_i \in \mathcal{D}$ (Line~2), its $\alpha$-mass subvector 
$\alpha\text{-}mass(\vec{x}_i)$ is computed and assigned to $\vec{x}'_i$ (Line~3), 
which is then added to $\mathcal{D}'$.  
The remaining steps invoke \textsc{PreciseSindiConstruction} \revisedtext{(Algorithm~\ref{algo:full_precision_ivf_construction})} on $\mathcal{D}'$ (Line~5), 
followed by returning both the inverted index $I$ and the original dataset $\mathcal{D}$ (Line~6) 
so that reordering can be performed during retrieval.

Figure~\ref{subfig:ip_error} shows that retaining under half of a query’s non-zero entries reduces the inner product error to nearly zero, 
cutting the search space by more than half. 
This suggests that posting lists from a few high–value non-zero entries in a query already cover most of the recall. 
Therefore, \textsc{Sindi} applies \textit{Mass Ratio Pruning} to queries: 
given $\beta \in (0,1]$, the pruned query is denoted $\beta\text{-}mass(\vec{q})$ and is used in coarse retrieval.

\subsection{Reordering}
\label{subsec:reordering}

Retaining only a small portion of non-zero entries preserves most of the inner product but may disrupt the partial order of full inner products. Using such pruned results directly for recall degrades accuracy. Nevertheless, experiments show that with enough candidates, the true nearest neighbors are often included. Figure~\ref{subfig:recall_cmp} reports Recall 10@500 and Recall 10@10 under different pruning ratios for documents and queries. Retaining 20\% of document entries and 15\% of query entries yields Recall 10@500 = 0.98 but Recall 10@10 = 0.63. This motivates a two-step strategy:  
(1) perform coarse recall with the pruned index to retrieve $\gamma$ candidates into a min-heap $H$;  
(2) compute the full inner products for all candidates in $H$ to refine the final top-$k$ results for efficient AMIPS.
The second stage is \textit{reordering}.


\begin{algorithm}[t]\small
	\DontPrintSemicolon
	\KwIn{Query $\vec{q}$, inverted index $I$, query prune ratio $\beta$, reorder number $\gamma$, and $k$}
	\KwOut{Top-$k$ points in $\mathcal{D}$ (at most $k$)}
	
	$H \leftarrow$ empty min-heap\;
	$R \leftarrow$ empty max-heap\;
	
	$\vec{q}' \leftarrow \beta\text{-}mass(\vec{q})$\;
	
	$H \leftarrow$ \textsc{PreciseSindiSearch}$(\vec{q}', I, \gamma)$\;
	
	\While{$H \neq \emptyset$}{
		$i, dis \leftarrow H.pop()$\;
		$dis' \leftarrow 1 - \delta(\vec{x}_i, \vec{q})$\;
		\uIf{$dis' < R.max()$ \textbf{or} $R.len() < k$}{
			$R.insert(i, dis')$\;
		}
		\uIf{$R.len() > k$}{
			$R.pop()$\;
		}
	}
	
	\Return{$R$}\;
	\caption{\textsc{ApproximateSindiSearch}}
	\label{algo:approximate_ivf_search}
\end{algorithm}

Algorithm~\ref{algo:approximate_ivf_search} shows the search procedure of the approximate \textsc{Sindi} index. 
Given a query $\vec{q}$, inverted index $I$, pruning ratio $\beta$, reordering size $\gamma$, and target $k$, the algorithm initializes an empty min-heap $H$ for coarse candidates and a max-heap $R$ for the final top-$k$ results (Lines~1--2). 
It first derives the $\beta$-mass subvector $\vec{q}'$ (Line~3) and invokes \textsc{PreciseSindiSearch} \revisedtext{(Algorithm~\ref{algo:full_precision_ivf_search})} to obtain $\gamma$ coarse candidates stored in $H$ (Line~4). 
While $H$ is not empty (Line~5), the best coarse candidate $(i,\,dis)$ is popped (Line~6), its exact distance to $\vec{q}$ is computed as $dis'$ (Line~7), and $(i,\,dis')$ is inserted into $R$ if it improves the current set or if $R$ contains fewer than $k$ elements (Line~8). 
If $R$ exceeds size $k$, its worst element is removed (Lines~9--10). 
After all candidates are processed, $R$ is returned as the final result (Line~12).

\section{Experimental Study}
\label{sec:experimental}

\subsection{Experimental Settings}

\textbf{Datasets.} Table~\ref{tab:datasets} summarizes the datasets used in our experiments, covering:  
(i) English datasets from the \textsc{MsMarco}~\cite{msmarco} passage ranking benchmark (including SPLADE-1M and SPLADE-FULL) and the \textsc{NQ}~\cite{nq} (Natural Questions) benchmark, all trained with the \textsc{Splade} model;  
(ii) Chinese dataset AntSparse, real business data from Ant Group trained with the BGE-M3~\cite{bge-m3} model, which has higher dimensionality due to the larger Chinese vocabulary;  
(iii) Random datasets with non-zero entry dimensions and values drawn from a uniform distribution.  
For each dataset, Table~\ref{tab:datasets} reports the average number of non-zero entries per vector ($\mathrm{avg}~\Vert \vec{x}_i \Vert$), average vectors per inverted list ($\mathrm{avg}~l$), and sparsity. The \textit{sparsity} of $\mathcal{D}$ is: $\textit{sparsity} = 1 - \frac{\sum_{\vec{x} \in \mathcal{D}} \Vert \vec{x} \Vert}{\Vert \mathcal{D} \Vert \cdot d}$.

We compare \textsc{Sindi} with five SOTA algorithms: \textsc{Seismic}, \textsc{PyANNS}, SOSIA, BMP, and HNSW. Below is a description of each algorithm: 
\begin{itemize}[leftmargin=*,labelsep=0.5em]
	\item \textbf{\textsc{Seismic}}~\cite{seismic}: A sparse vector index based on inverted lists. 
	\item \textbf{SOSIA}~\cite{sosia}: A sparse vector index using min-hash signatures. 
	\item \textbf{BMP}~\cite{bmp, sp}: A dynamic pruning strategy for learning sparse vector retrieval. It divides the original dataset into fine-grained blocks and generates a maximum value vector for each block to evaluate whether the block should be queried.
	\item \textbf{HNSW}~\cite{hnsw}: A graph-based index originally for dense vectors; we modify the data format and distance computation to support sparse vectors. 
	\item \textbf{\textsc{PyANNS}}~\cite{pyanns}: The open-source champion of the BigANN Benchmark 2023 Sparse Track. It is built on HNSW and incorporates quantization, query pruning, and rerank strategies. 
	\revisedtext{
		\item \textbf{SHNSW}~\cite{shnsw}:The open-source runner-up of the BigANN Benchmark 2023 Sparse Track. It is a graph-based index (HNSW variant) that utilizes memory optimization and early termination.
		\item \textbf{SINNAMON}~\cite{sinnamon}: An inverted index using hashing to compress vectors into dense sketches for SIMD scoring.
	}
\end{itemize}

\noindent\textbf{Parameter Settings.} We use the optimal parameters for each algorithm to ensure a fair comparison, the detailed configuration are in Table~\ref{tab:construction_params}. Parameter choices either follow the recommendations from the original authors or are determined via grid search, the details are provided in Appendix C.

\begin{table*}[t!]\vspace{-2ex}
	\centering
	\caption{Dataset Statistics and Characteristics}
	\label{tab:datasets}
	\begin{tabular}{lcccccccccc}
		\toprule
		\textbf{Dataset} & $\Vert \mathcal{D} \Vert$ & $\mathrm{avg}~\Vert x_i \Vert$ & $nq$ & $\mathrm{avg}~\Vert q \Vert$ & $d$ & \textit{\textbf{Sparsity}} & \textbf{Size (GB)} & $\mathrm{avg}~l$ & \textbf{Model} & \textbf{Language} \\
		\midrule
		SPLADE-1M     & 1,000,000   & 126.3  & 6980   & 49.1  & 30108    & 0.9958  & 0.94   & 4569.2   & \textsc{Splade}    & English \\
		SPLADE-FULL   & 8,841,823   & 126.8  & 6980   & 49.1  & 30108    & 0.9958  & 8.42   & 40447.3  & \textsc{Splade}    & English \\
		AntSparse-1M   & 1,000,000   & 40.1   & 1000   & 5.8   & 250000   & 0.9998  & 0.31   & 902.6    & BGE-M3    & Chinese \\
		AntSparse-10M  & 10,000,000  & 40.1   & 1000   & 5.8   & 250000   & 0.9998  & 3.06   & 6560.7   & BGE-M3    & Chinese \\
		NQ            & 2,681,468   & 149.4  & 3452   & 47.0  & 30510    & 0.9951  & 3.01   & 13914.7  & \textsc{Splade}    & English \\
		RANDOM-5M     & 5,000,000   & 150.0  & 5000   & 50.4  & 30000    & 0.9950  & 5.62   & 25000.0  & -         & -       \\
		\bottomrule
	\end{tabular}
\end{table*}

\begin{table}[htbp]
	\caption{\small Construction Parameter Settings (\revisedtext{\textbf{bold}} denotes selected optima).}
	\label{tab:construction_params}
	\centering
	\begin{tabularx}{0.48\textwidth}{p{2.3cm}X}
		\toprule
		\textbf{Algorithm} & \textbf{Construction Parameter Settings} \\
		\midrule
		\multirow{2}{*}{\textsc{Sindi}} 
		& SPLADE-FULL: $\alpha \in$ \{0.2, 0.3, 0.4, $\mathbf{0.5}$, 0.6, 0.7, 0.8\} \\
		& AntSparse\_10M: $\alpha \in$ \{0.7, 0.75, 0.8, $\mathbf{0.85}$, 0.9, 0.95, 1\} \\
		\midrule
		\multirow{2}{*}{\textsc{Seismic}}
		& SPLADE-FULL: $\lambda = 6000$, $\beta = 0.067$, $\alpha = 0.4$ \\
		& AntSparse\_10M: $\lambda \in$ \{5000, 6000, 7000, 8000, 9000, $\mathbf{10000}, 50000\}$, $\beta = 0.1$, $\alpha = 0.5$ \\
		\midrule
		\multirow{2}{*}{\textsc{Shnsw}}
		& SPLADE-1M: $M = 16$, efConstruction = 200 \\
		& SPLADE-FULL: $M = 16$, efConstruction = 1000 \\
		\midrule
		\textsc{Sinnamon} & $m = 31$, $h = 2$ \\
		\midrule
		\textsc{Hnsw} \& \textsc{PyANNS} & max\_degree = 32, ef\_construction = 1000 \\
		\midrule
		\textsc{Sosia} & $m = 150$, $l = 40$ \\
		\midrule
		\textsc{Bmp} & $l = 16$ \\
		\bottomrule
	\end{tabularx}
\end{table}

\noindent\textbf{Performance Metrics.} We evaluate index construction time, index size, recall, and queries per second (QPS) for all baselines. For a query $\vec{q}$, let $R = \{\vec{x}_1, \dots, \vec{x}_k\}$ denote the AMIPS results, and $R^\ast = \{\vec{x}^\ast_1, \dots, \vec{x}^\ast_k\}$ denote the exact MIPS results. Recall is computed as $\mathrm{Recall} = \frac{\Vert R \, \cap \, R^\ast \Vert}{\Vert R^\ast \Vert}$.
We specifically report Recall@50 and Recall@100. Since approximate methods trade off efficiency and accuracy, we also report \textit{throughput}, defined as the number of queries processed per second.

\noindent\textbf{Environment.} Experiments are conducted on a server running the CentOS, with an Intel Xeon Platinum 8269CY CPU @ 2.50GHz and 512\,GB memory. We implement \textsc{Sindi} in C++, compiled with \texttt{g++} 10.2.1. For hardware acceleration, \textsc{Sindi} supports both AVX-512 intrinsics and compiler auto-vectorization. To ensure reproducibility and peak performance, the following experiments are conducted using the manual AVX-512 intrinsics implementation. Detailed implementations are provided in Appendix A.

\begin{figure*}[t!]\centering\vspace{-2ex}
	\subfigure{
		\scalebox{0.6}[0.6]{\includegraphics{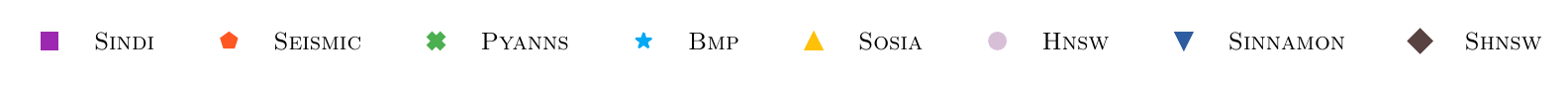}}}\hfill\\\vspace{-4ex}
	\addtocounter{subfigure}{-1}
	\subfigure[][{\scriptsize SPLADE-1M Recall@50}]{
		\scalebox{0.18}[0.18]{\includegraphics{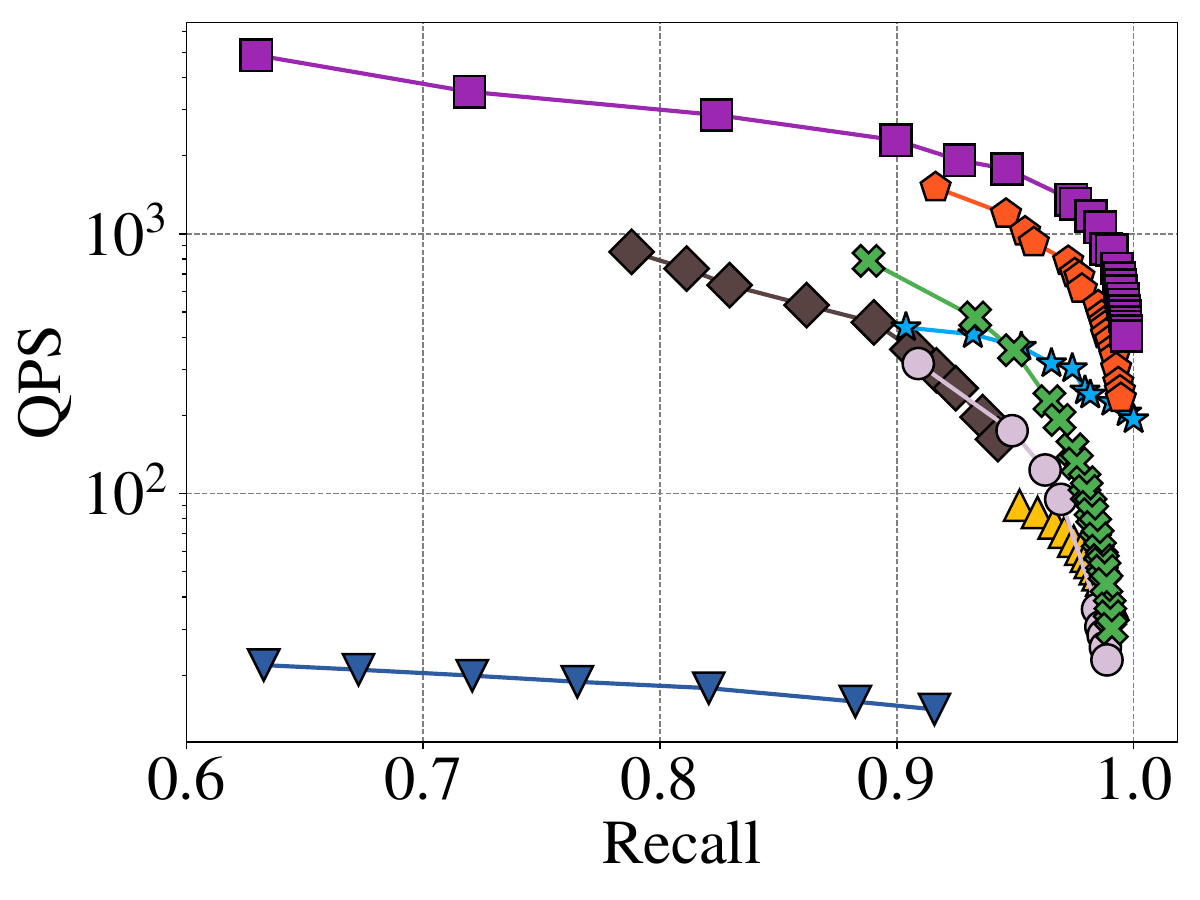}}
		\label{subfig:splade-1m-top50}}
	\hfill
	\subfigure[][{\scriptsize SPLADE-1M Recall@100}]{
		\scalebox{0.18}[0.18]{\includegraphics{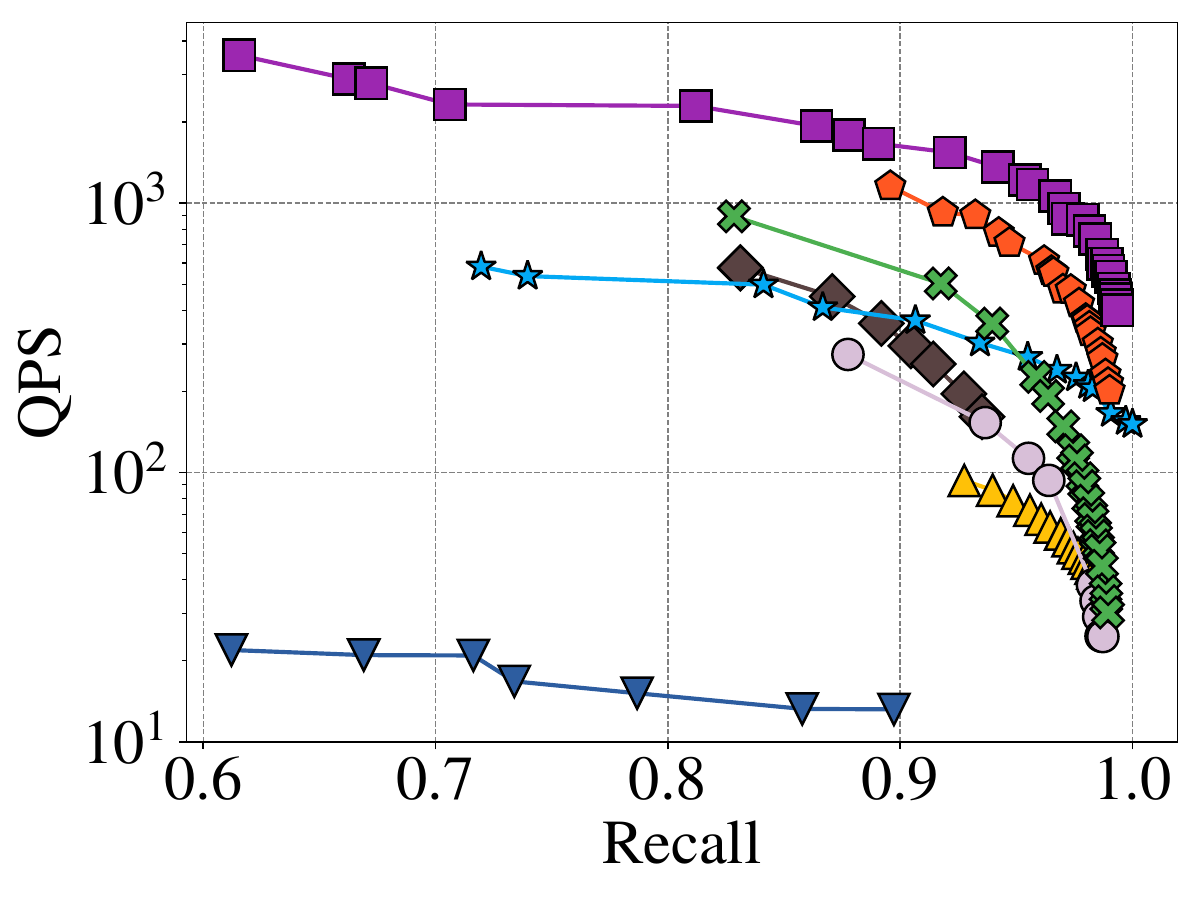}}
		\label{subfig:splade-1m-top100}}
	\hfill
	\subfigure[][{\scriptsize SPLADE-FULL Recall@50}]{
		\scalebox{0.18}[0.18]{\includegraphics{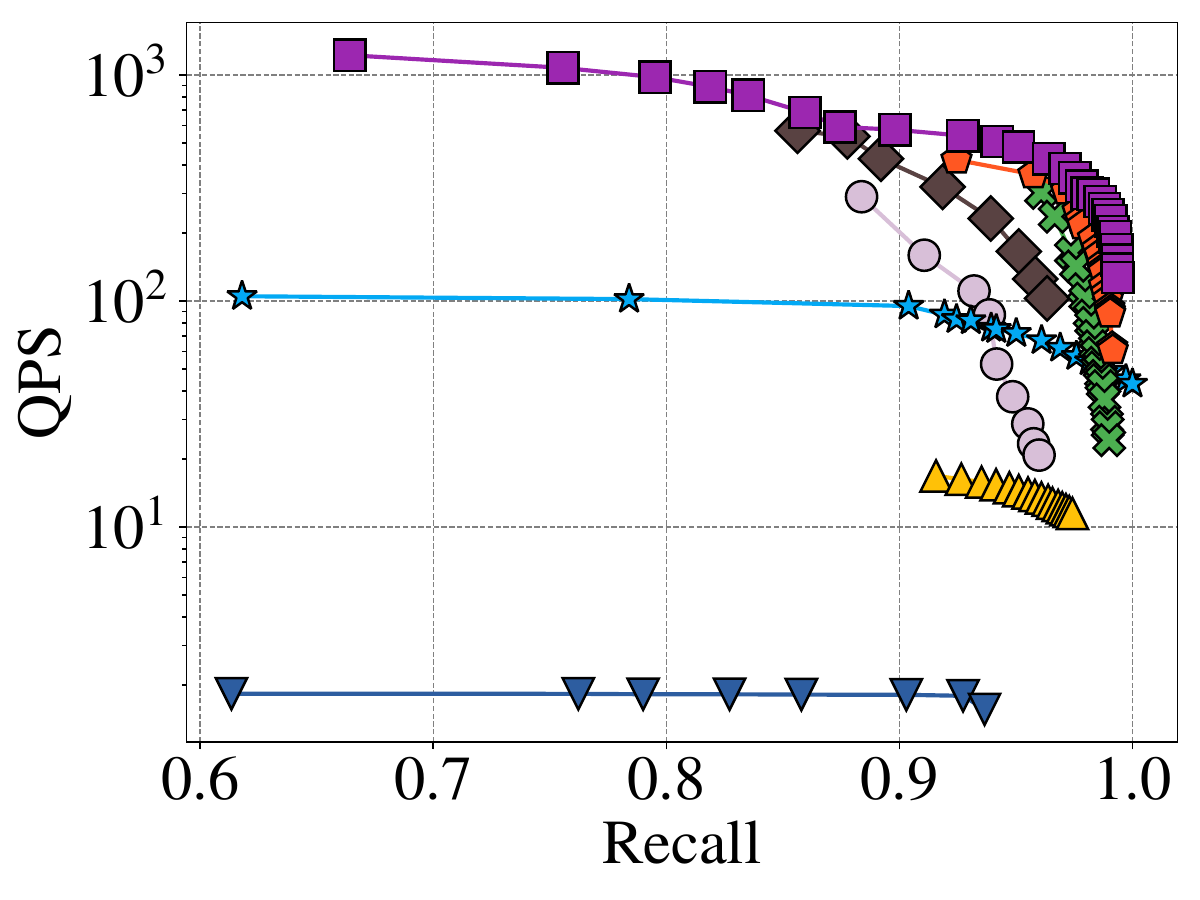}}
		\label{subfig:splade-full-top50}}
	\hfill
	\subfigure[][{\scriptsize SPLADE-FULL Recall@100}]{
		\scalebox{0.18}[0.18]{\includegraphics{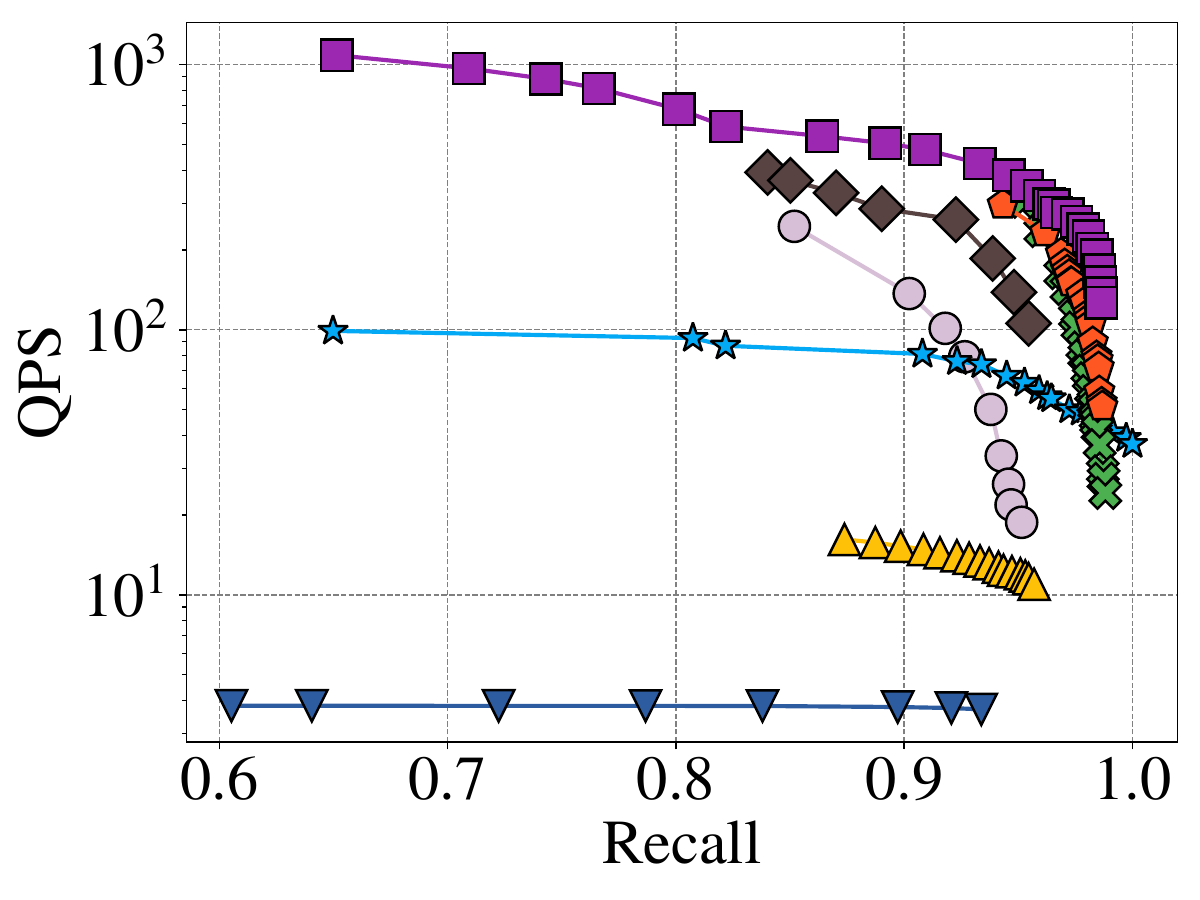}}
		\label{subfig:splade-full-top100}}
	\vspace{-2ex} 
	
	\subfigure[][{\scriptsize NQ Recall@50}]{
		\scalebox{0.18}[0.18]{\includegraphics{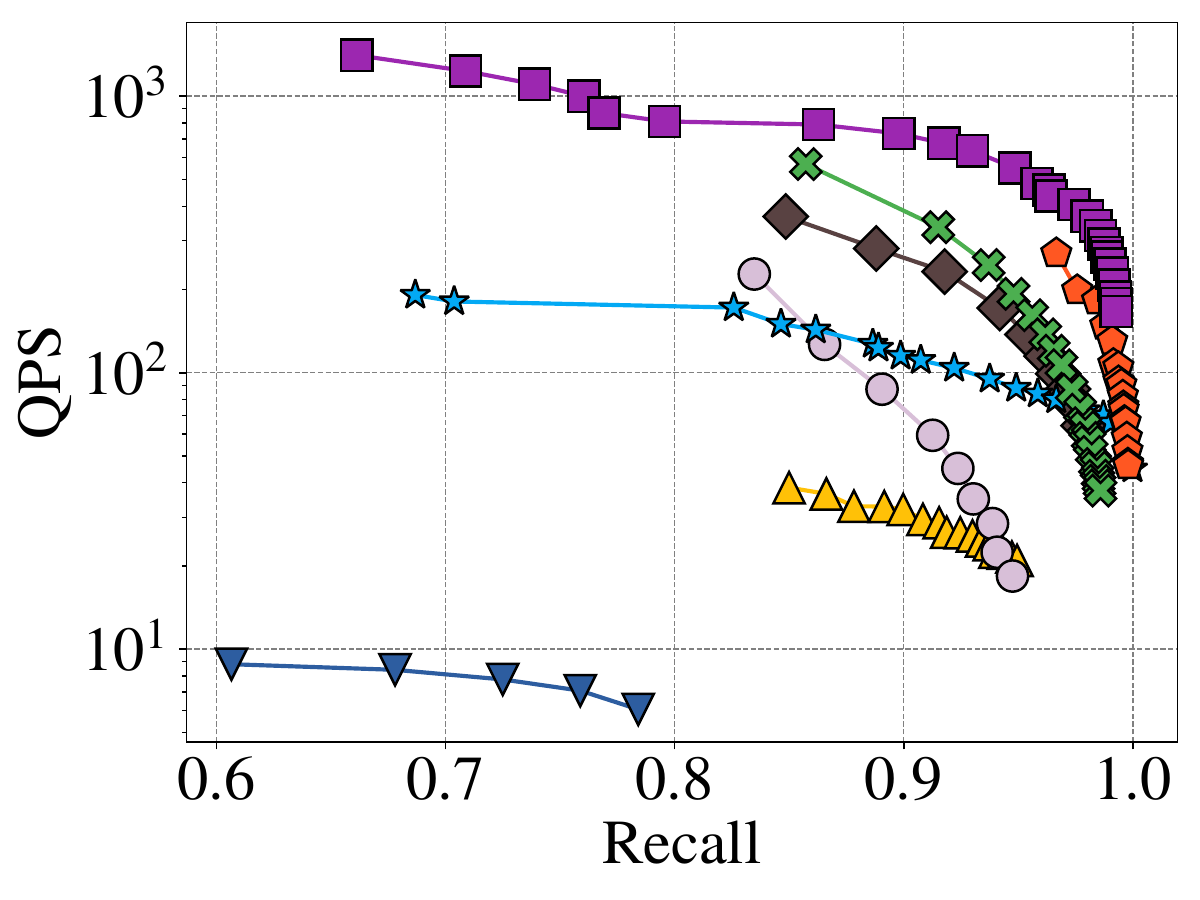}}
		\label{subfig:nq-top50}}
	\hfill
	\subfigure[][{\scriptsize NQ Recall@100}]{
		\scalebox{0.18}[0.18]{\includegraphics{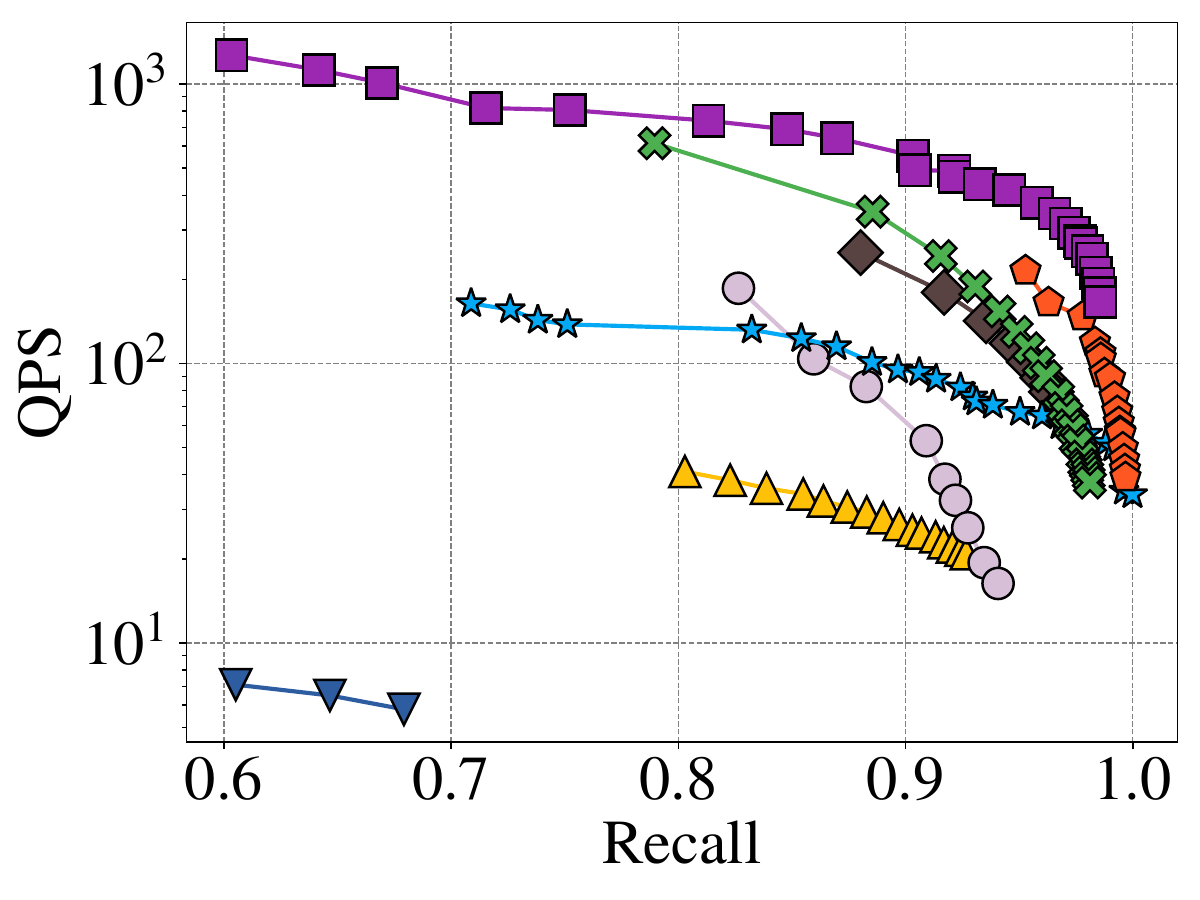}}
		\label{subfig:nq-top100}}
	\hfill
	\subfigure[][{\scriptsize RANDOM-5M Recall@50}]{
		\scalebox{0.18}[0.18]{\includegraphics{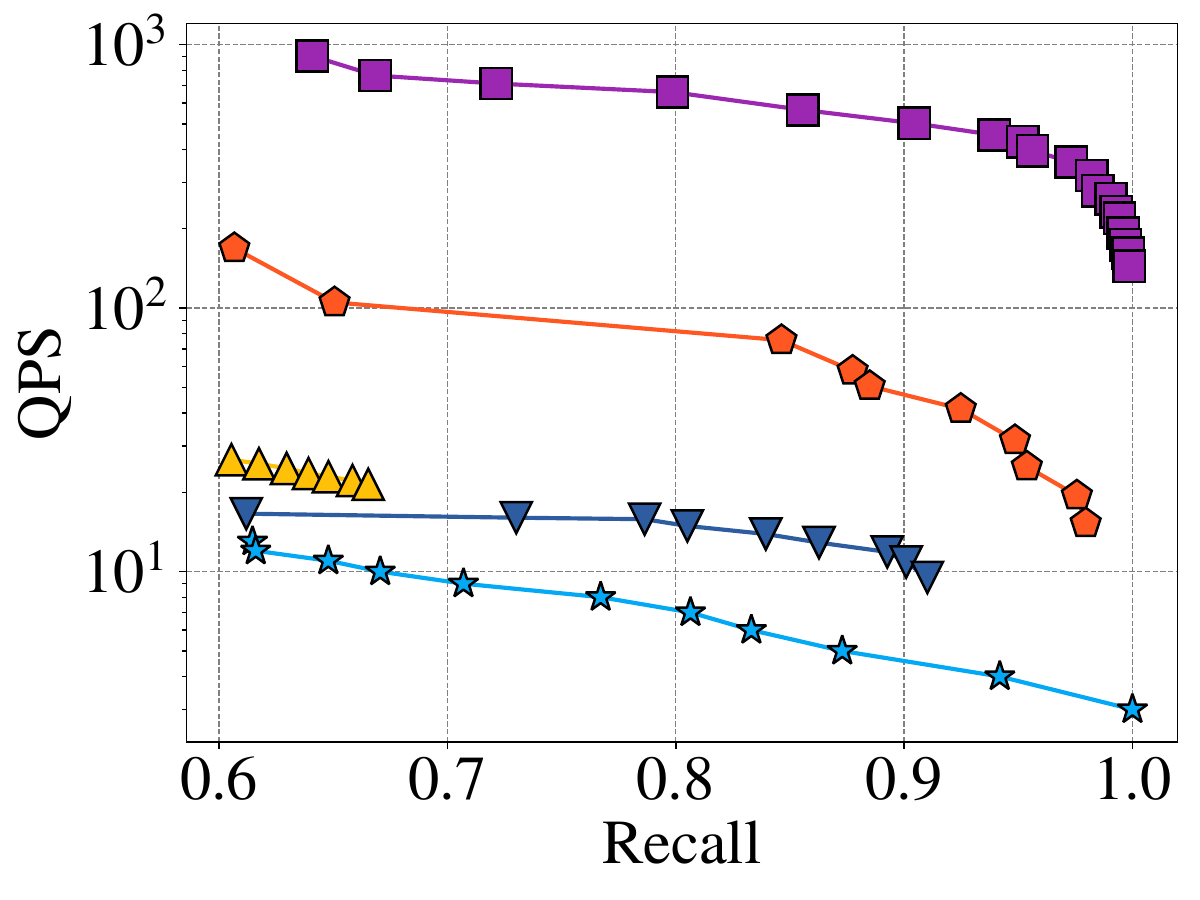}}
		\label{subfig:random_5m-top50}}
	\hfill
	\subfigure[][{\scriptsize RANDOM-5M Recall@100}]{
		\scalebox{0.18}[0.18]{\includegraphics{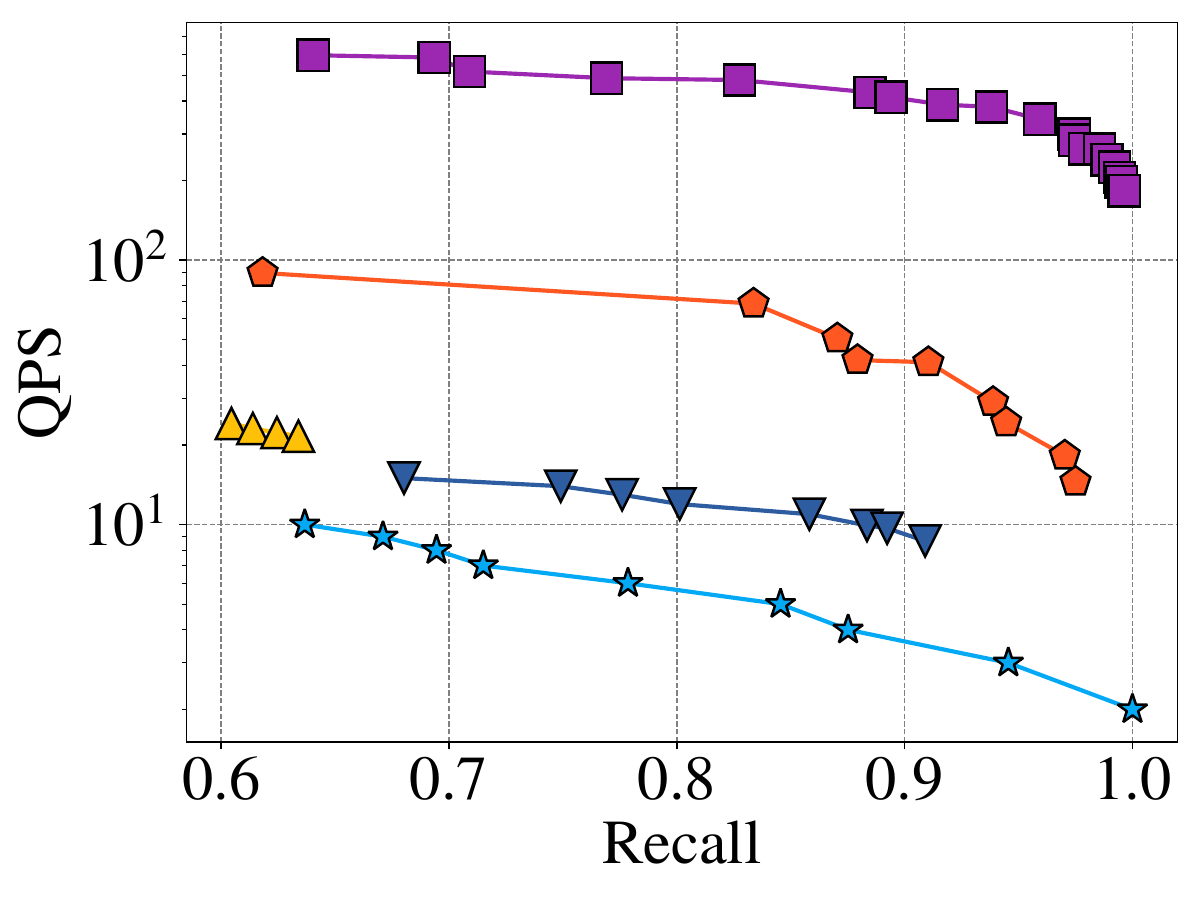}}
		\label{subfig:random_5m-top100}}
	\vspace{-2ex} 
	
	\subfigure[][{\scriptsize AntSparse-1M Recall@50}]{
		\scalebox{0.18}[0.18]{\includegraphics{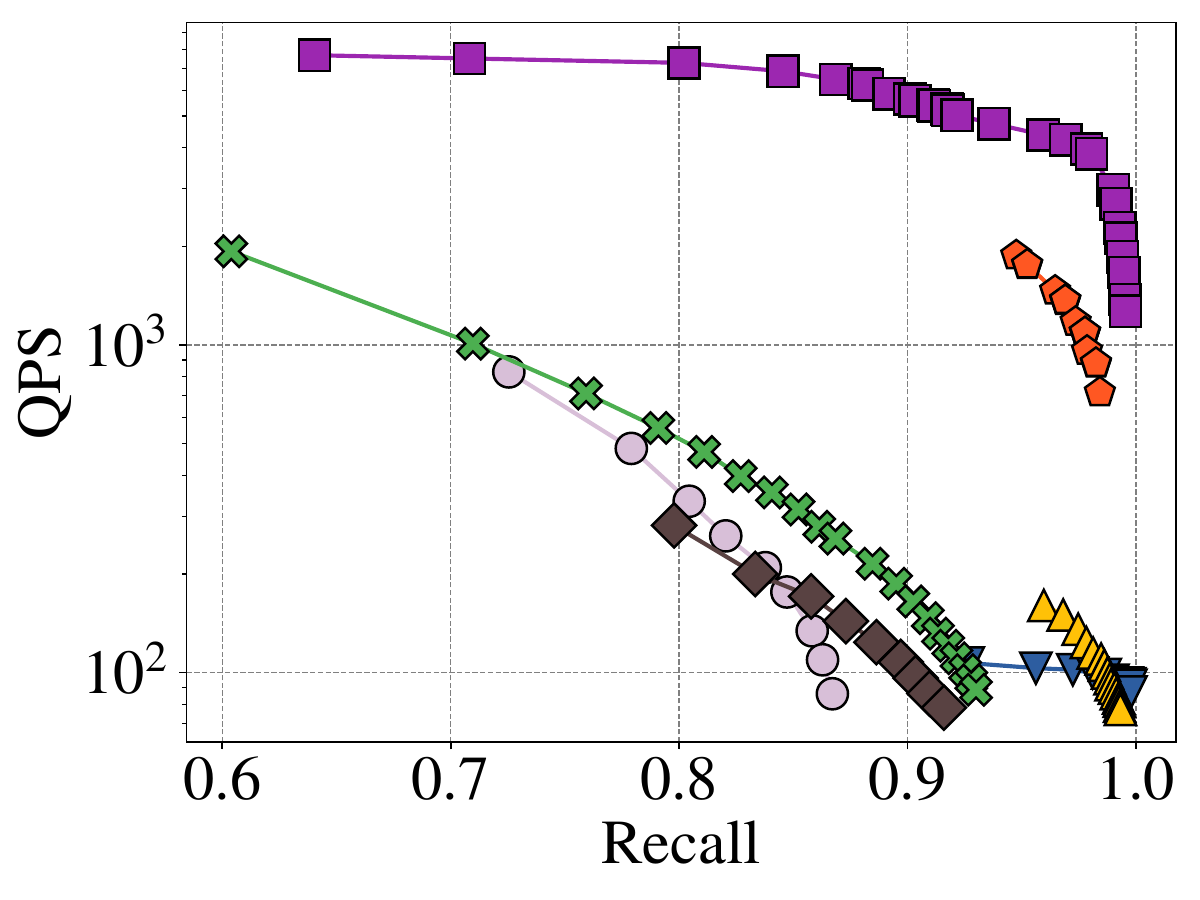}}
		\label{subfig:AntSparse-1m-top50}}
	\hfill
	\subfigure[][{\scriptsize AntSparse-1M Recall@100}]{
		\scalebox{0.18}[0.18]{\includegraphics{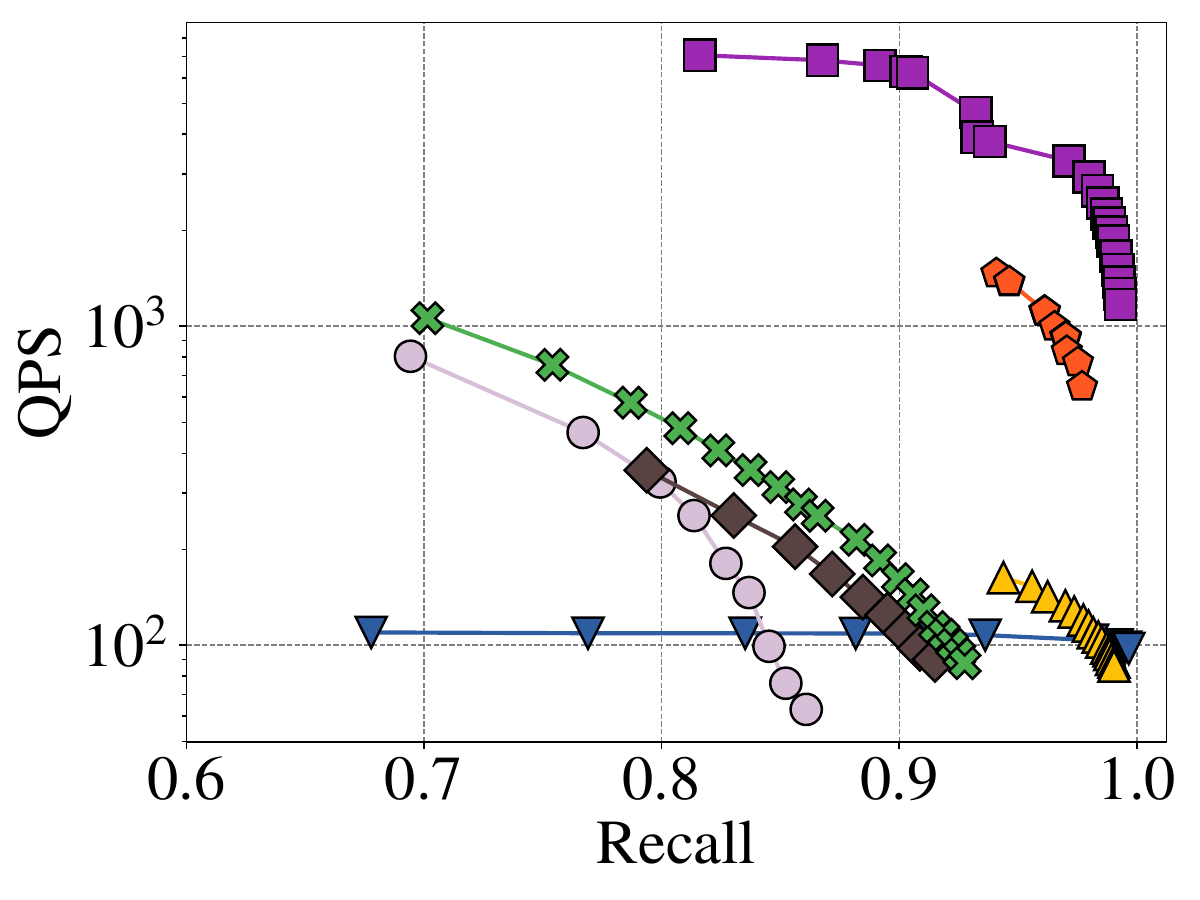}}
		\label{subfig:AntSparse-1m-top100}}
	\hfill
	\subfigure[][{\scriptsize AntSparse-10M Recall@50}]{
		\scalebox{0.18}[0.18]{\includegraphics{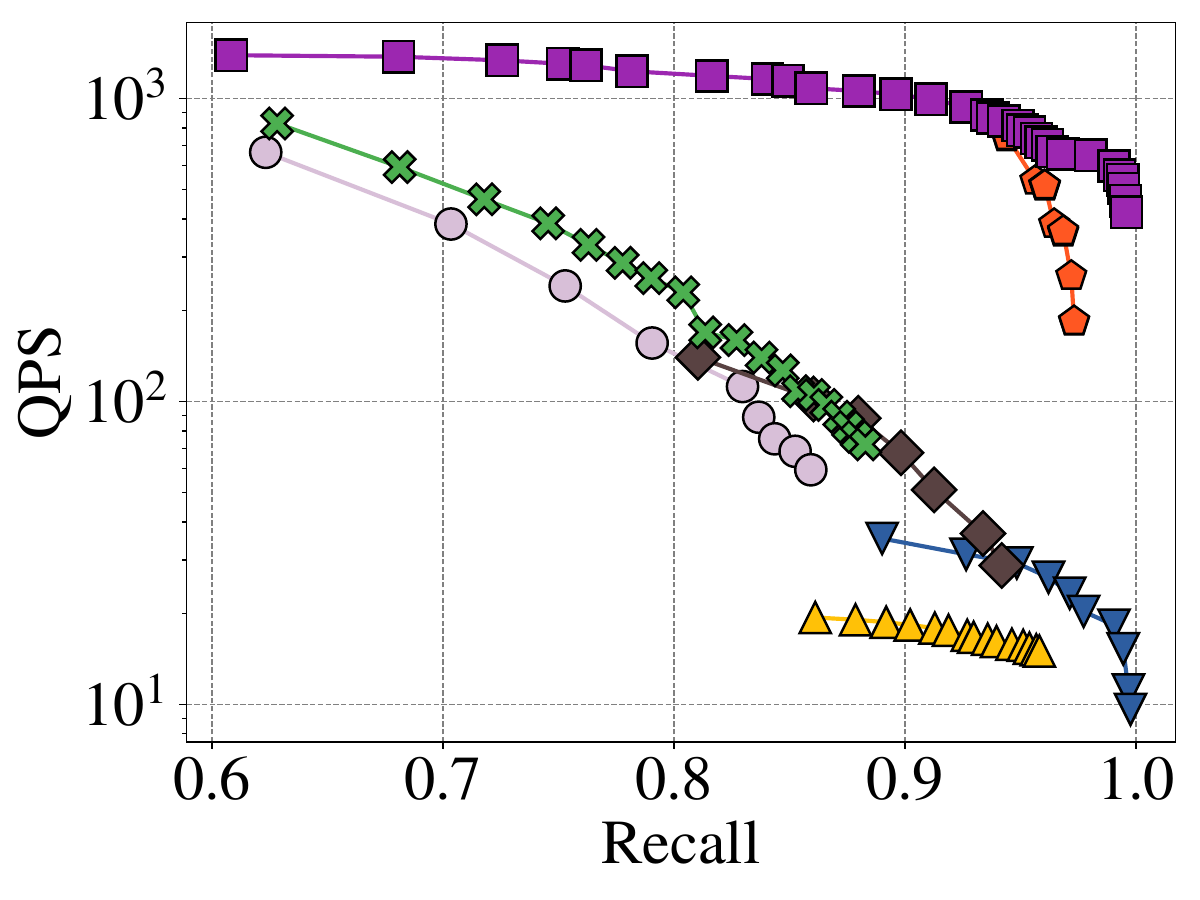}}
		\label{subfig:AntSparse-10m-top50}}
	\hfill
	\subfigure[][{\scriptsize AntSparse-10M Recall@100}]{
		\scalebox{0.18}[0.18]{\includegraphics{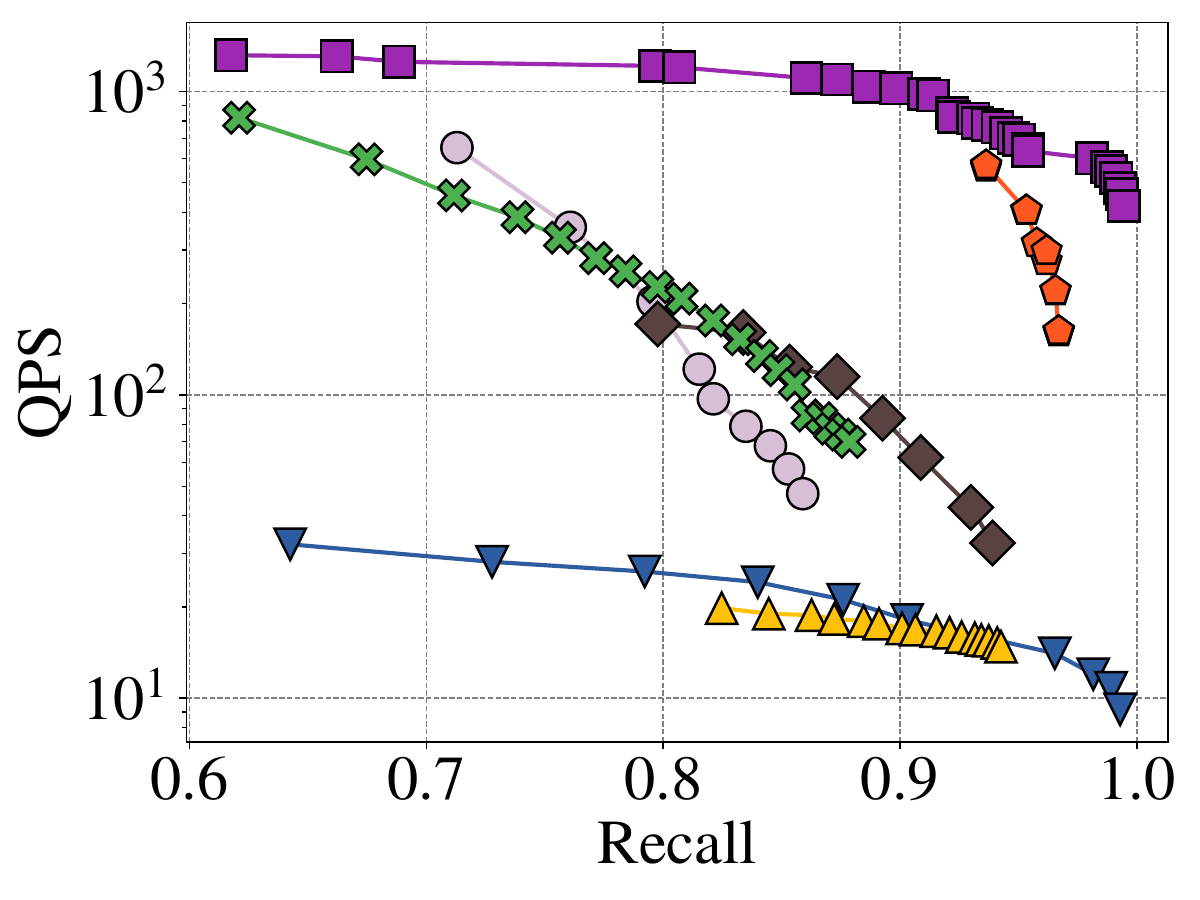}}
		\label{subfig:AntSparse-10m-top100}}
	
	\caption{\revisedtext{\small Overall Performance.}}
	\vspace{-2ex} 
	\label{fig:overview}
\end{figure*}

\subsection{Overall Performance}
\subsubsection{Recall and QPS}
Figure~\ref{fig:overview} shows the relationship between recall (Recall@50 and Recall@100) and single-threaded QPS for all algorithms. For each method, we report the best results across all tested parameter configurations.

On both English and Chinese datasets, \textsc{Sindi} achieves the highest QPS at the same recall levels. When Recall@50 is 99\%, on SPLADE-1M, the QPS of \textsc{Sindi} is $2.0\times$ that of \textsc{Seismic} and $26.4\times$ that of \textsc{PyANNS}$;$ on SPLADE-FULL, it is $4.16\times$ and $10.0\times$ higher, respectively. When Recall@100 is 98\% on SPLADE-FULL, \textsc{Sindi} attains $1.9\times$ the QPS of \textsc{Seismic} and $3.2\times$ that of \textsc{PyANNS}.

On the Chinese AntSparse-10M dataset encoded by the BGE-M3~\cite{bge-m3} model, \textsc{Sindi} also performs best. Fixing Recall@50 at 97\%, its QPS is $2.5\times$ that of \textsc{Seismic}, the recall of \textsc{PyANNS} is limited due to the high \textit{sparsity} of the dataset.

On RANDOM-5M, generated uniformly at random, the uniform distribution yields very few common non-zero dimensions between vectors, causing graph-based methods (\textsc{PyANNS}, HNSW, SHNSW) to suffer severe connectivity loss and low recall. The clustering effectiveness of \textsc{Seismic} is also sensitive to data distributions, resulting in marked degradation. In contrast, \textsc{Sindi} is unaffected by data distribution and achieves the best performance, with QPS exceeding \textsc{Seismic} by \textit{an order of magnitude}.

Overall, these results demonstrate that \textsc{Sindi} consistently delivers state-of-the-art performance across datasets with diverse languages, models, and distributions.

\subsubsection{Index Size and Construction Time}

\revisedtext{Figure~\ref{fig:IT_IS} summarizes the index size (all including datasets and search index storage) and construction time of \textsc{Sindi}, \textsc{Seismic}, and \textsc{PyANNS} on the two largest datasets (SPLADE-FULL and AntSparse-10M), showing that \textsc{Sindi} has the lowest \revisedtext{construction time}.}

\textsc{Seismic}, which stores summary vectors for each block, yields the largest index size; on AntSparse-10M, its size is $3.8\times$ that of \textsc{Sindi}. By contrast, the graph index construction of \textsc{PyANNS} requires numerous distance computations to find neighbors, resulting in a construction time $71.5\times$ that of \textsc{Sindi} on SPLADE-FULL. In comparison, \textsc{Sindi} mainly sorts each vector’s non-zero entries for pruning, keeping the \revisedtext{computational overhead} low and enabling rapid index building.

\subsection{Parameters}

\subsubsection{The Impact of $\alpha$}

\begin{figure}[t!]\vspace{-2ex}
	\centering
	\subfigure{
		\scalebox{0.55}[0.55]{\includegraphics{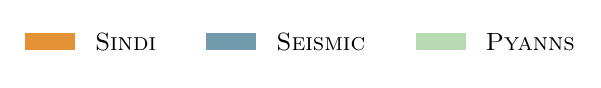}}}\hfill\\\vspace{-4ex}
	\addtocounter{subfigure}{-1}
	\subfigure[][{\scriptsize SPLADE-FULL}]{
		\scalebox{0.19}[0.18]{\includegraphics{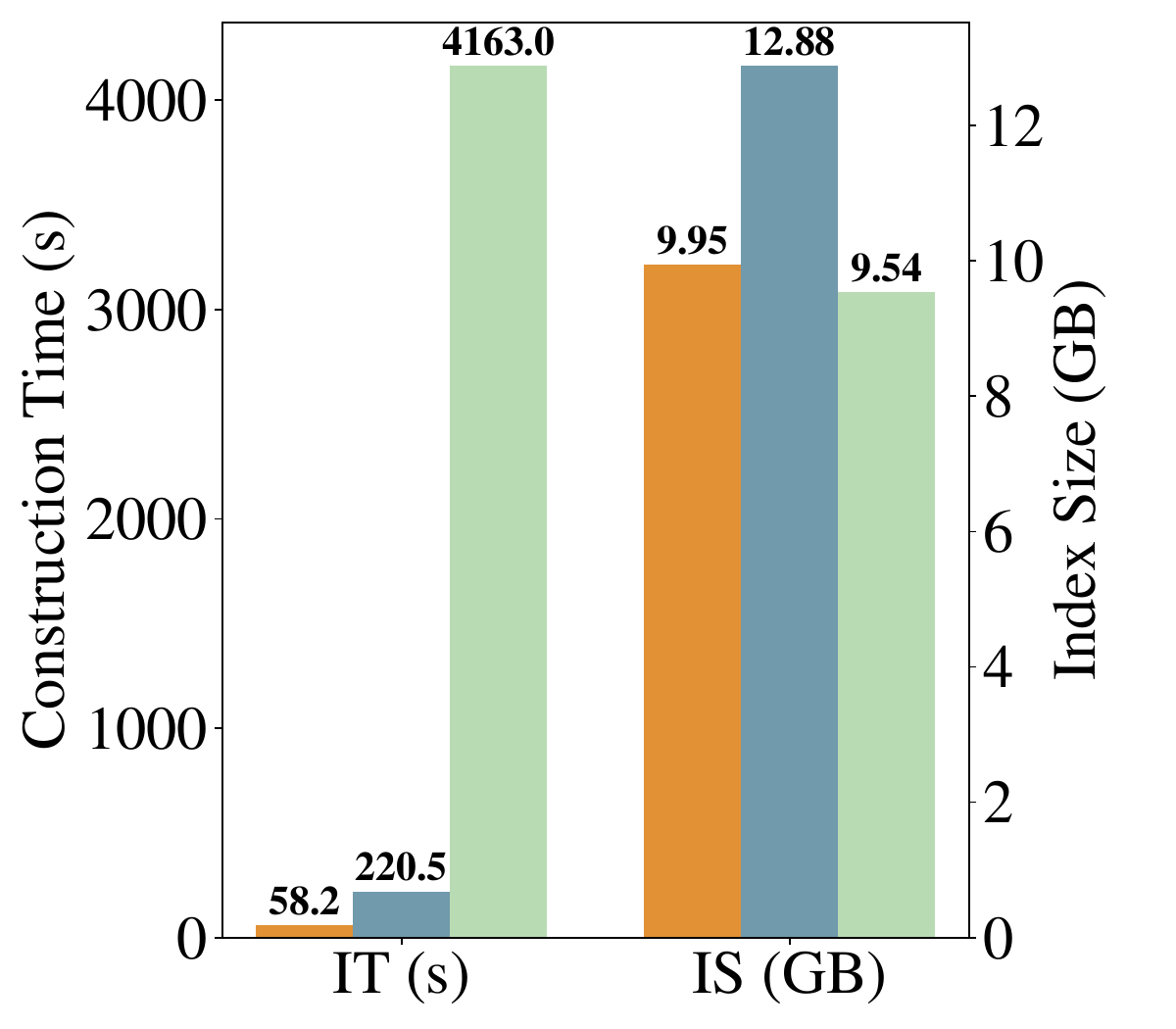}}
		\label{subfig:splade-full-it-is}}
	\hfill
	\subfigure[][{\scriptsize AntSparse-10M}]{
		\scalebox{0.19}[0.18]{\includegraphics{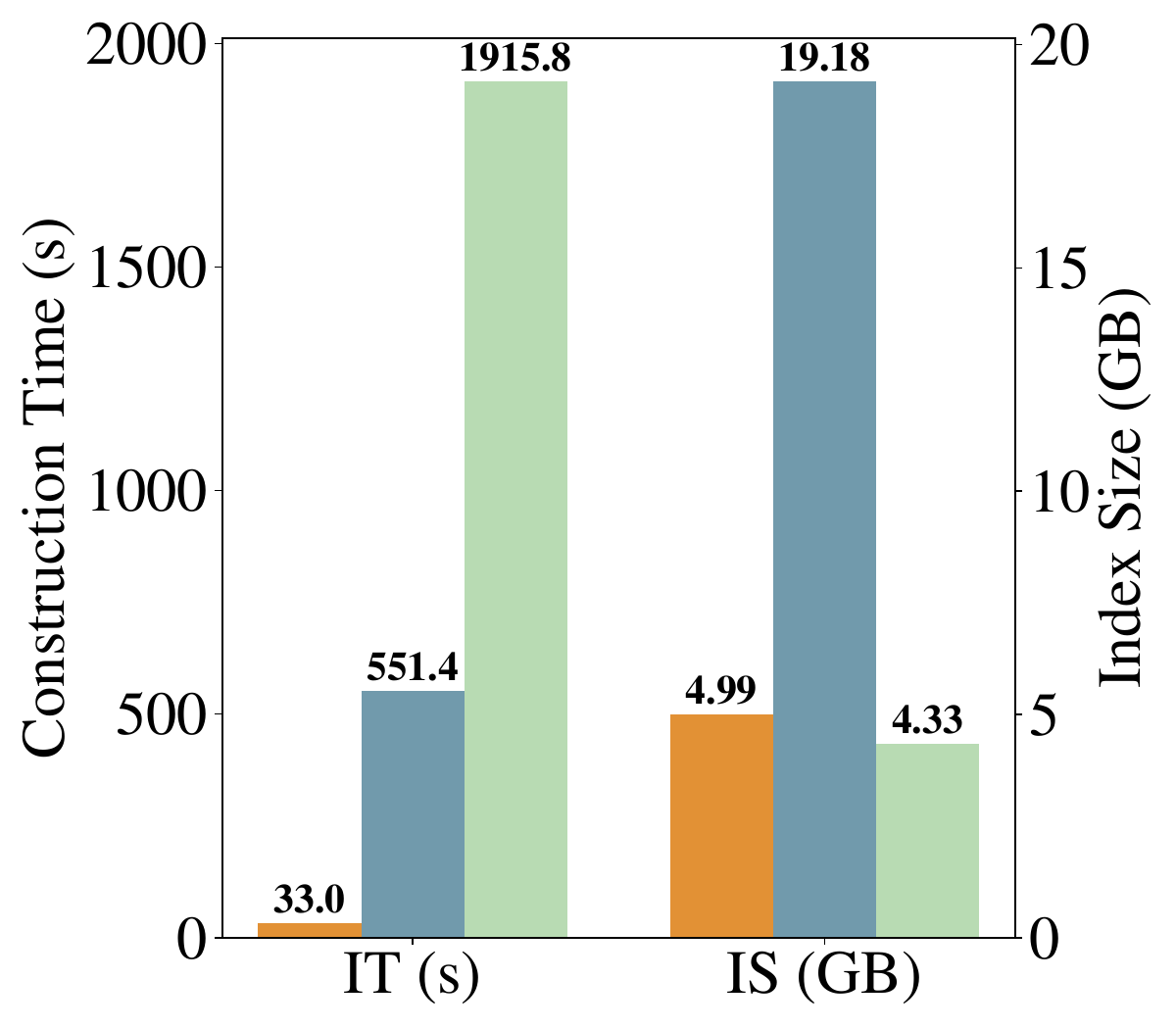}}
		\label{subfig:AntSparse-10m-it-is}}
	\vspace{-1ex} 
	
	\caption{\small Index Size and Construction Time for Different Datasets and Algorithms.}
	\vspace{-2ex} 
	\label{fig:IT_IS}
\end{figure}

\begin{figure}[t!]\vspace{-2ex}
	\centering
	\subfigure{
		\scalebox{0.5}[0.5]{\includegraphics{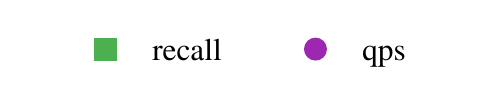}}}\hfill\\\vspace{-4ex}
	\addtocounter{subfigure}{-1}
	\subfigure[][{\scriptsize SPLADE-FULL Recall@50}]{
		\scalebox{0.19}[0.19]{\includegraphics{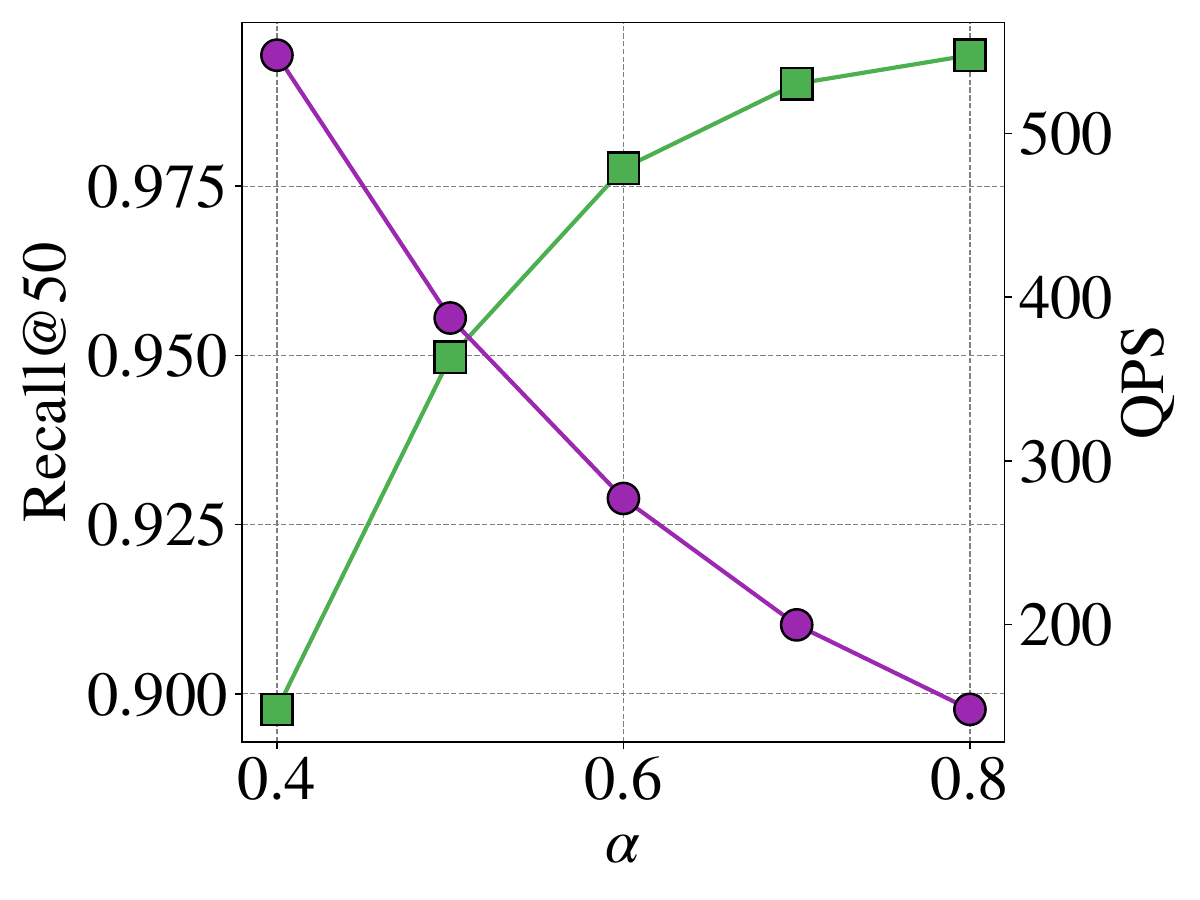}}
		\label{subfig:splade-full-top50-alpha}}
	\hfill
	\subfigure[][{\scriptsize AntSparse-10M Recall@50}]{
		\scalebox{0.19}[0.19]{\includegraphics{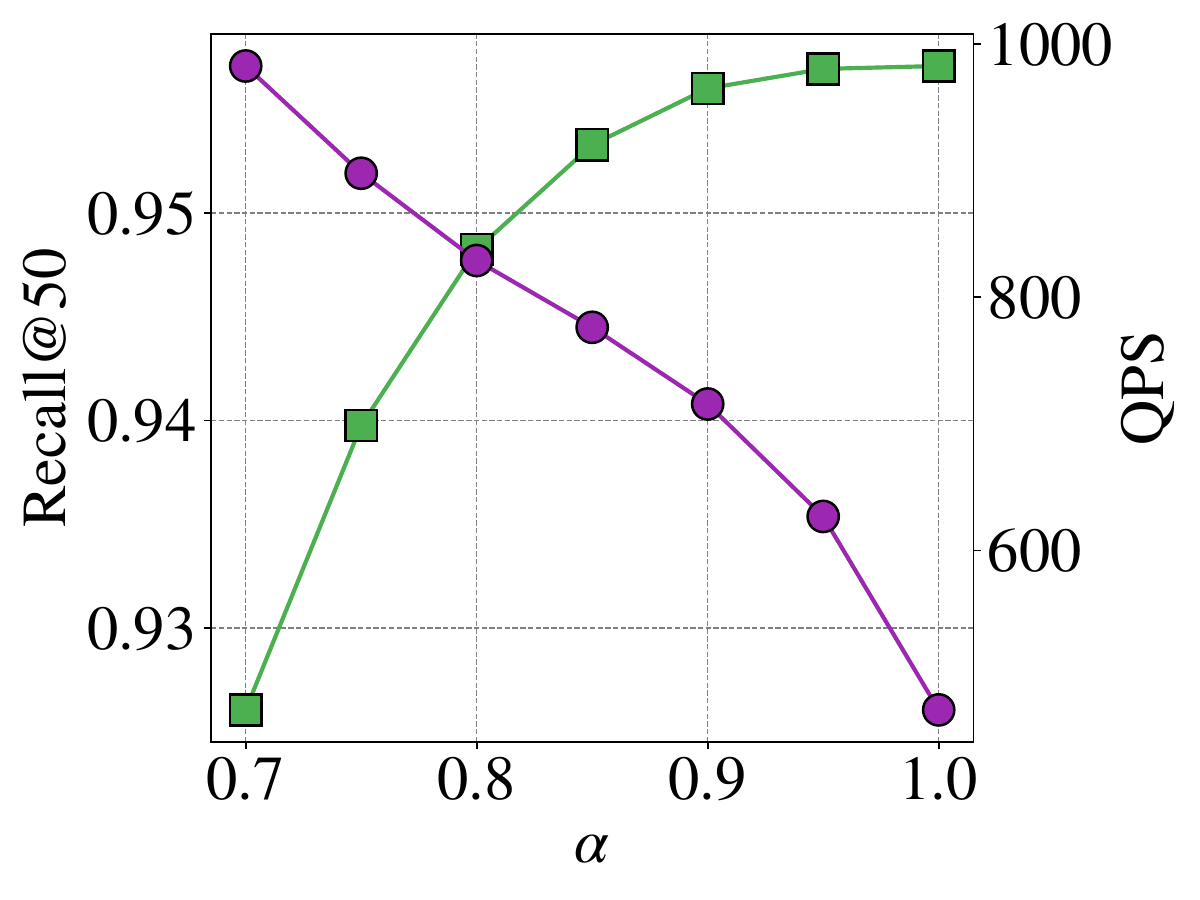}}
		\label{subfig:AntSparse-10m-top50-alpha}}
	\vspace{-1ex} 
	
	\caption{\small The Impact of $\alpha$.}
	\vspace{-2ex} 
	\label{fig:alpha}
\end{figure}

This experiment examines how the document pruning parameter $\alpha$, controlling the retained high-mass non-zero entries, affects \textsc{Sindi}'s performance. A larger $\alpha$ retains more entries per vector, potentially increasing recall but raising search costs. We vary $\alpha$ from 0.4 to 0.8 on SPLADE-FULL and 0.7 to 1.0 on AntSparse, keeping $\beta$ and $\gamma$ fixed.

Figure~\ref{fig:alpha} shows that, on \textsc{MsMarco}, recall rises and QPS drops as $\alpha$ increases, with both trends flattening at higher $\alpha$. In the lower $\alpha$ range, recall improves rapidly with moderate QPS loss, while in the higher range, recall gains slow and QPS stabilizes. On AntSparse, recall also grows more slowly at large $\alpha$, but QPS declines more steeply.

The slower recall gain at high $\alpha$ is due to the \textit{saturation effect} in Section~\ref{subsec:pruning}: once enough non-zero entries are kept, further additions barely reduce inner product error. The sharper QPS drop on AntSparse comes from its lower variance of non-zero values, which leads to more retained entries for the same $\alpha$ increase, and thus more postings to scan.

\subsubsection{The Impact of \textit{sparsity}}

We evaluate the performance of \textsc{Sindi} under varying dataset \textit{sparsity} levels. For a fixed $\mathrm{avg}\,\Vert \vec{x} \Vert$, larger $d$ yields higher \textit{sparsity}. To examine its impact, we generate five synthetic datasets ($\Vert \mathcal{D} \Vert = 1\text{M}$, $\mathrm{avg}\,\Vert \vec{x} \Vert = 120$) with $d \in \{10\mathrm{k},\,30\mathrm{k},\,50\mathrm{k},\,70\mathrm{k},\,100\mathrm{k}\}$, thereby increasing \textit{sparsity} gradually. Each dataset is generated uniformly at random, where both the positions and the values of non-zero entries follow a uniform distribution. Figure~\ref{fig:sparsity} compares \textsc{Sindi} and \textsc{Seismic} at Recall@50 = 90\% and Recall@50 = 99\%.

As \textit{sparsity} increases, both \textsc{Sindi} and \textsc{Seismic} achieve higher QPS at the same recall because the IVF index produces shorter inverted lists (average length $\mathrm{avg} \, l$ decreases),  reducing the number of candidate non-zero entries $\lVert q \rVert \, l$ to be visited. Since all true nearest neighbors still reside in the probed lists, recall remains unaffected.

\textsc{Sindi} consistently outperforms \textsc{Seismic} by maintaining approximately 10$\times$ higher QPS across all \textit{sparsity} levels. This demonstrates \textsc{Sindi}'s efficiency and robustness to different data distributions. 

We also report the sensitivity analysis of reordering depth $\gamma$, which is shown in Appendix E.

\subsection{Ablation}

\subsubsection{The Impact of Pruning Method}

\begin{figure}[t!]\vspace{-2ex}
	\centering
	\subfigure{
		\scalebox{0.5}[0.5]{\includegraphics{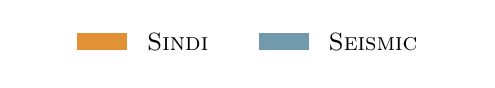}}}\hfill\\\vspace{-4ex}
	\addtocounter{subfigure}{-1}
	\subfigure[][{\scriptsize Recall@50=90\%}]{
		\scalebox{0.22}[0.22]{\includegraphics{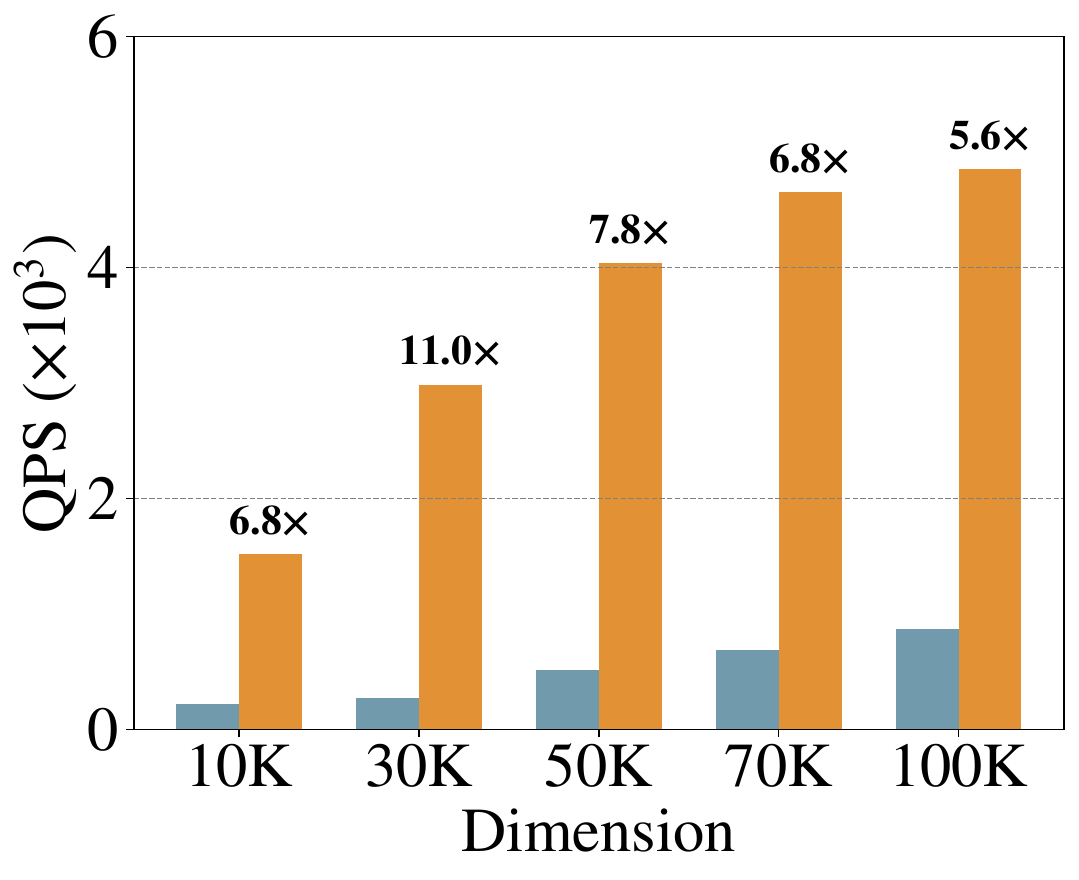}}
		\label{subfig:random-1M-recall90}}
	\hfill
	\subfigure[][{\scriptsize Recall@50=99\%}]{
		\scalebox{0.22}[0.22]{\includegraphics{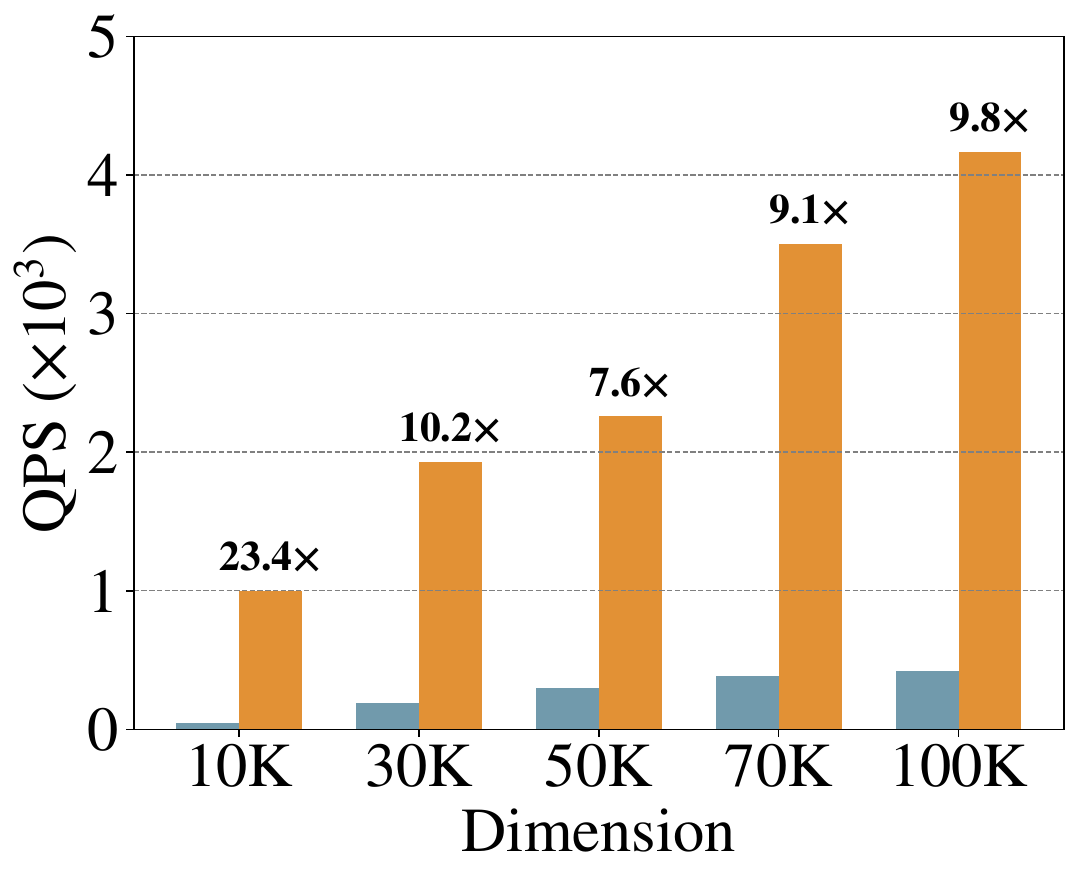}}
		\label{subfig:random-1M-recall99}}
	\vspace{-1ex} 
	
	\caption{\small QPS of \textsc{Sindi} and \textsc{Seismic} on RANDOM-1M Dataset with Different Sparsity}
	\vspace{-2ex} 
	\label{fig:sparsity}
\end{figure}

\begin{figure}[t!]\vspace{-1ex}
	\centering
	\subfigure{
		\scalebox{0.4}[0.4]{\includegraphics{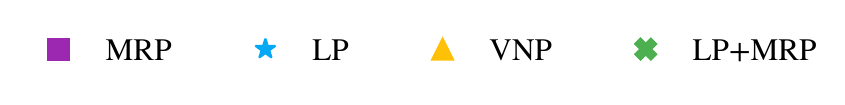}}}\hfill\\\vspace{-4ex}
	\addtocounter{subfigure}{-1}
	\subfigure[][{\scriptsize SPLADE-FULL Recall@10}]{
		\scalebox{0.19}[0.19]{\includegraphics{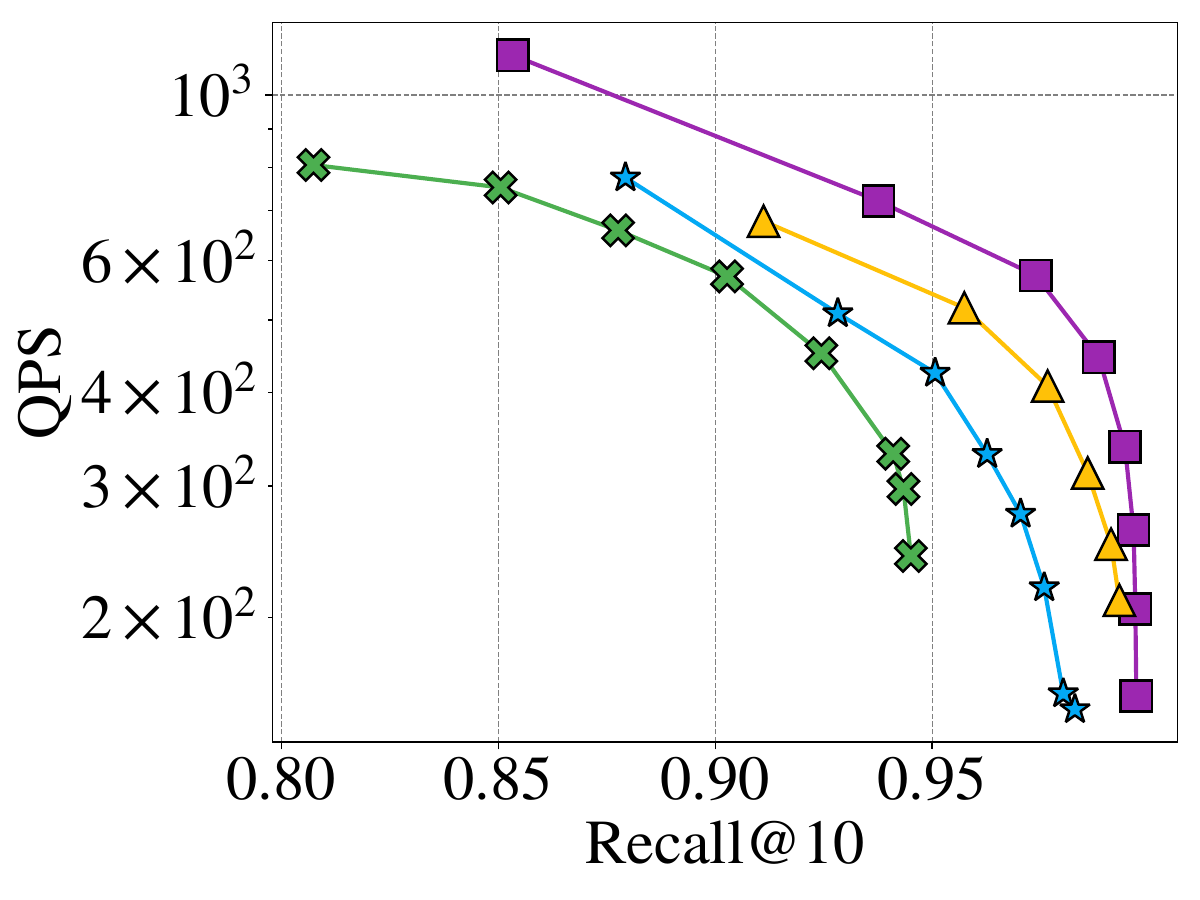}}
		\label{subfig:splade-full-top10-prune}}
	\hfill
	\subfigure[][{\scriptsize AntSparse Recall@10}]{
		\scalebox{0.19}[0.19]{\includegraphics{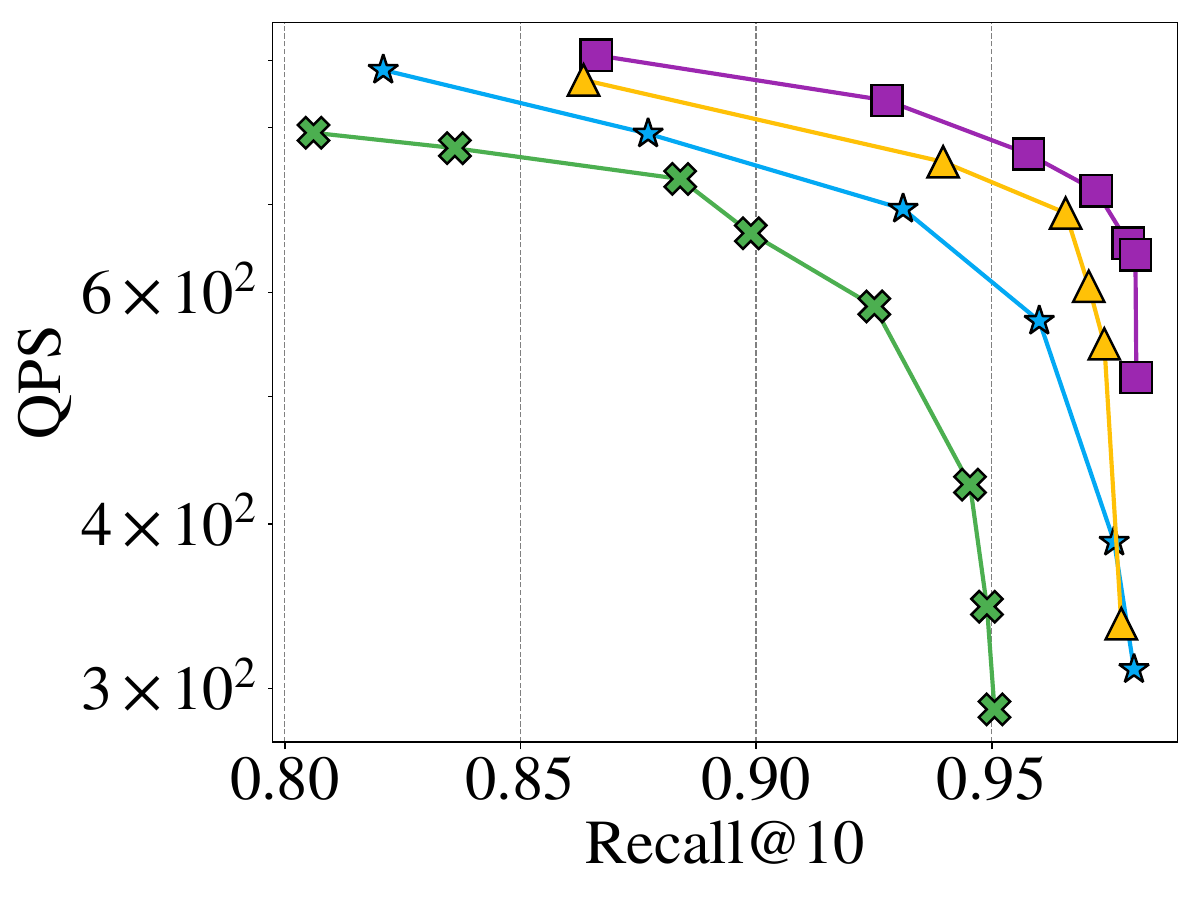}}
		\label{subfig:wholenet-10M-top10-prune}}
	\vspace{-1ex} 
	
	\caption{\small Recall@10 vs QPS on \textsc{MsMarco} and AntSparse of \textit{Mass Ratio Pruning}, \textit{List Pruning} and \textit{Vector Number Pruning}.}
	\vspace{-2ex} 
	\label{fig:prune_cmp}
\end{figure}

Figure~\ref{fig:prune_cmp} illustrates the performance of different pruning strategies on the SPLADE-FULL and AntSparse datasets. All strategies are evaluated under the same $\beta$ and $\gamma$ settings, while varying $\alpha$ to measure Recall and QPS. \textit{Mass Ratio Pruning} achieves the best overall performance, followed by \textit{Vector Number Pruning} and \textit{List Pruning}, with the lowest performance observed when combining \textit{List Pruning} and \textit{Mass Ratio Pruning}. 

The advantage of \textit{Mass Ratio Pruning} lies in its ability to preserve the non-zero entries that contribute most to the inner product, thereby retaining more true nearest neighbors during the coarse recall stage. In contrast, \textit{List Pruning} limits the posting list size for each dimension, which can result in two issues: some lists become too short and keep mainly small-value entries, while others remain too long and remove large-value entries. As a result, \textit{List Pruning} is less suitable for \textsc{Sindi}. \textsc{Seismic}, however, uses \textit{List Pruning} because it computes the full inner product for all vectors in the lists, avoiding large accuracy losses. When \textit{List Pruning} is combined with \textit{Mass Ratio Pruning}, even more non-zero entries are discarded, further reducing recall.

\subsubsection{The Impact of Reorder}

\revisedtext{To investigate the effectiveness of the reordering strategy, we evaluated \textsc{Sindi}'s Recall@50 and query time with and without reordering on the SPLADE-FULL and AntSparse-10M datasets. Regarding parameter selection, $\beta$  and $\gamma$ were fixed; these values serve as representative samples drawn from the optimal intervals identified via grid search. Conversely, to assess robustness across varying index densities, we varied the document prune ratio $\alpha$ from 0.3 to 0.6 on SPLADE-FULL and from 0.7 to 1.0 on AntSparse-10M.}

The results are shown in Figure~\ref{fig:reorder}. For the reordering strategy, the query time includes both accumulation time and reordering time, whereas the non-reordering strategy includes only accumulation time. Since $\gamma$ is fixed, the reorder time remains relatively constant. As $\alpha$ increases, the accumulation time also increases accordingly. Although reorder time accounts for only a small portion of the total query time, it yields a substantial recall improvement. For example, on SPLADE-FULL with $\alpha = 0.6$, the accumulation time is 17099 ms and the reorder time is 3553 ms ($\approx 17.2\%$ of the total), yet recall improves from 0.71 to 0.97. This demonstrates the clear benefit of incorporating the reordering strategy.

The reordering strategy is effective for two main reasons. First, it focuses computation on a limited set of non-zero entries that contribute most to the inner product, greatly reducing unnecessary operations. Second, the partial inner products derived from high-value entries largely preserve the true ranking order, ensuring that small $\gamma$ is sufficient to contain the true nearest neighbors, thereby improving both efficiency and accuracy.

\subsection{Scalability}

\begin{figure}[t!]\vspace{-2ex}
	\centering
	\subfigure{
		\scalebox{0.24}[0.24]{\includegraphics{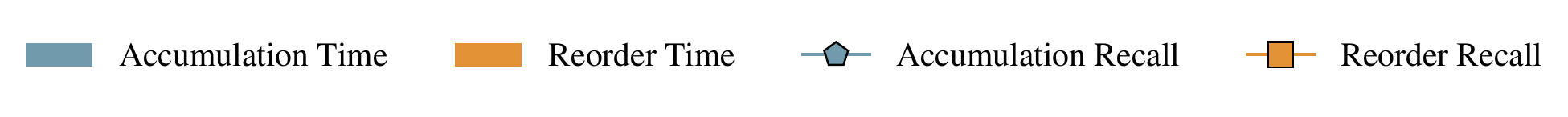}}}\hfill\\\vspace{-3ex}
	\addtocounter{subfigure}{-1}
	\subfigure[][{\scriptsize SPLADE-FULL}]{
		\scalebox{0.18}[0.18]{\includegraphics{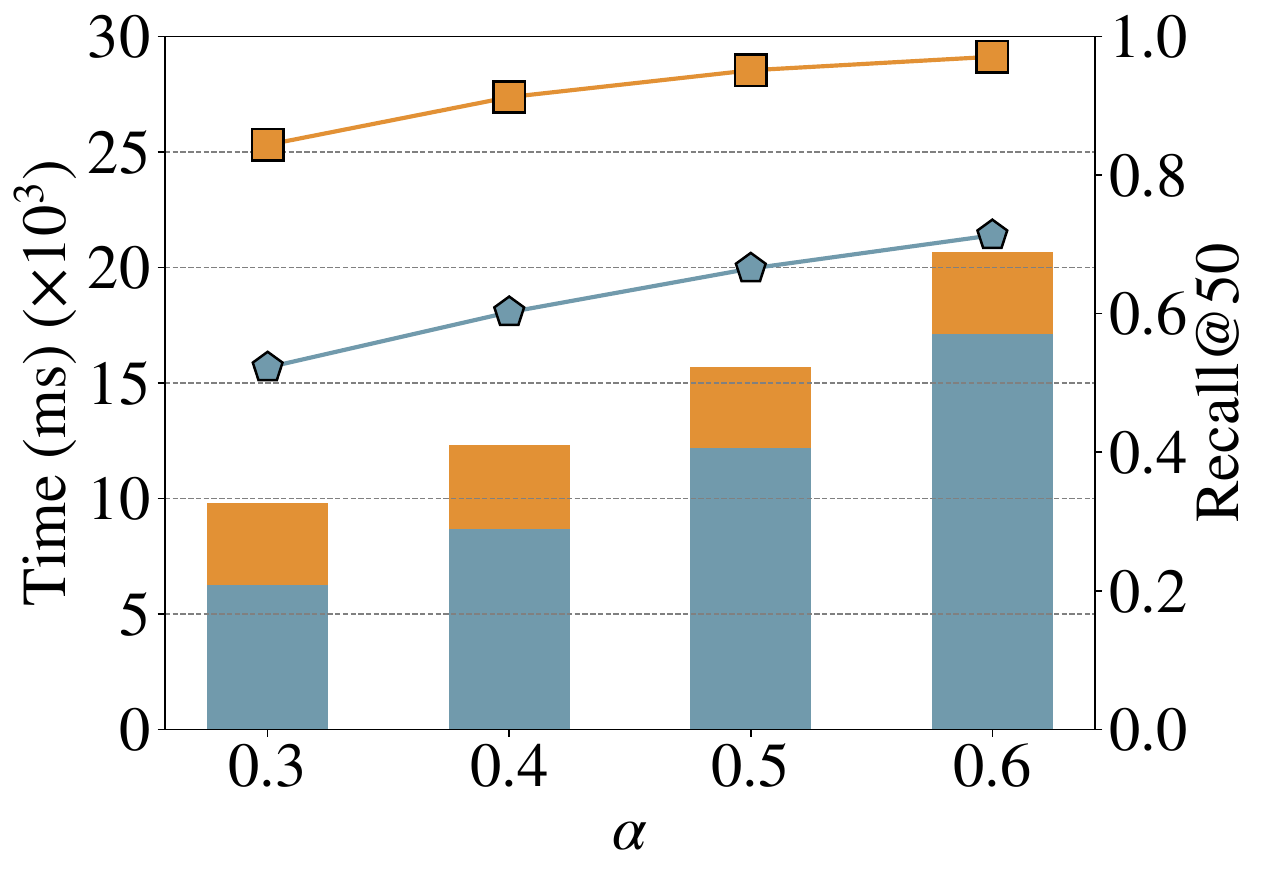}}
		\label{subfig:splade-full-top50-reorder}}
	\hfill
	\subfigure[][{\scriptsize AntSparse-10M}]{
		\scalebox{0.18}[0.18]{\includegraphics{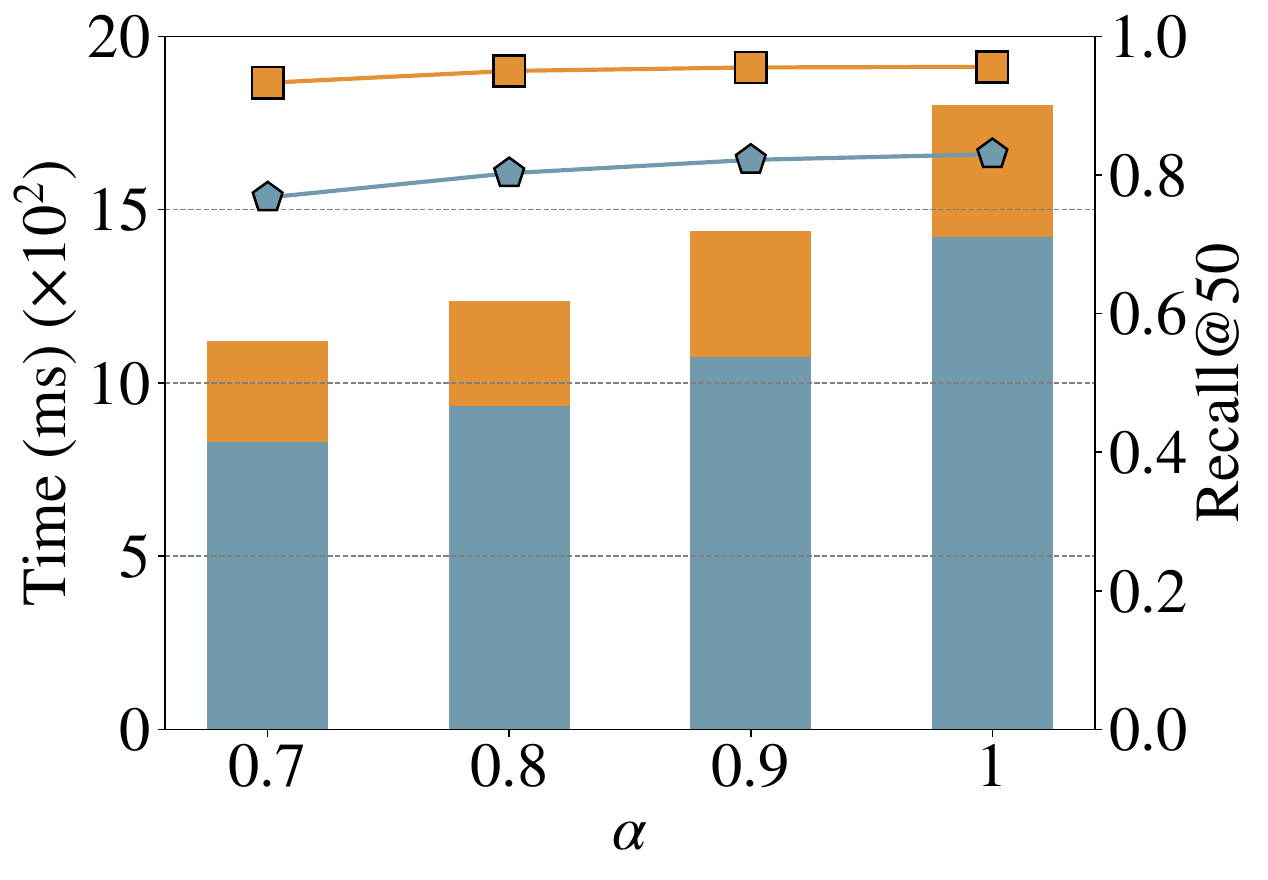}}
		\label{subfig:antsparse-top50-reorder}}
	\vspace{-1ex} 
	
	\caption{\small  Reorder vs. Non-Reorder on SPLADE-FULL and AntSparse-10M Datasets: Time Cost and Recall@50 with Varying $\alpha$.}
	\vspace{-2ex} 
	\label{fig:reorder}
\end{figure}

To further evaluate the scalability of the \textsc{Sindi} algorithm, we conducted a multi-threaded performance test on two large-scale datasets: SPLADE-FULL and AntSparse-10M. We measured QPS at different recall targets, with \(\mathrm{Recall}@50 \in \{0.95, 0.97, 0.99\}\) for SPLADE-FULL and \(\mathrm{Recall}@50 \in \{0.90, 0.95, 0.99\}\) for AntSparse-10M, while varying the number of CPU cores from 2 to 10, as shown in Figure~\ref{fig:thread}. 

On AntSparse-10M at \(\mathrm{Recall}@50 = 0.90\), using 2 CPU cores yields 1979.49 QPS ($\approx 989.75$ QPS per core), whereas 10 cores achieve 8374.01 QPS ($\approx 837.40$ QPS per core), indicating per-core efficiency drops by less than 16\% when scaling from 2 to 10 cores. On SPLADE-FULL at \(\mathrm{Recall}@50 = 0.99\), using 2 CPU cores yields 453.46 QPS ($\approx 226.73$ QPS per core), whereas 10 cores achieve 2142.05 QPS ($\approx 214.21$ QPS per core), corresponding to a per-core efficiency drop of about 5.5\%. Similar scaling behavior is observed across other recall targets for both datasets.

These results show that \textsc{Sindi} maintains high multi-core efficiency across datasets and accuracy targets, confirming its suitability for deployment in scenarios requiring both high recall and high throughput, with minimal parallelization overhead.

We also report the performance comparison with Lucene-based systems in Appendix F.

\begin{figure}[t!]\vspace{-2ex}
	\centering
	\subfigure{
		\scalebox{0.33}[0.33]{\includegraphics{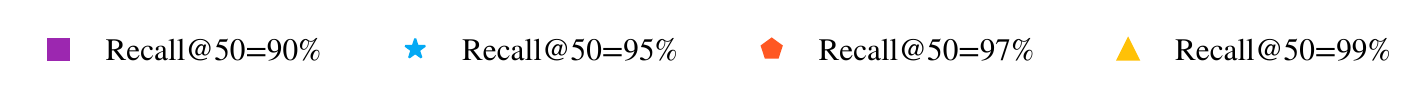}}}\hfill\\\vspace{-3ex}
	\addtocounter{subfigure}{-1}
	\subfigure[][{\scriptsize SPLADE-FULL}]{
		\scalebox{0.18}[0.18]{\includegraphics{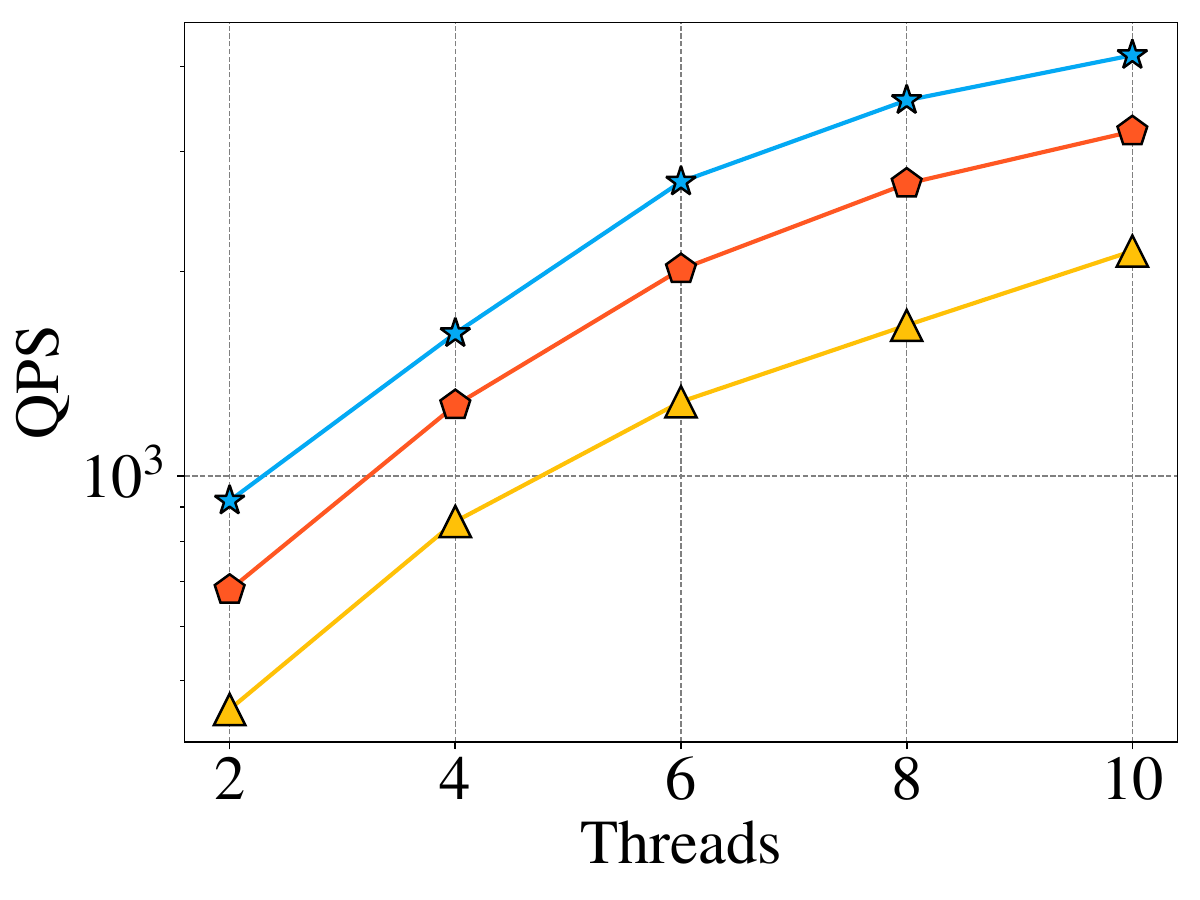}}
		\label{subfig:splade-full-top50-thread}}
	\hfill
	\subfigure[][{\scriptsize AntSparse-10M}]{
		\scalebox{0.18}[0.18]{\includegraphics{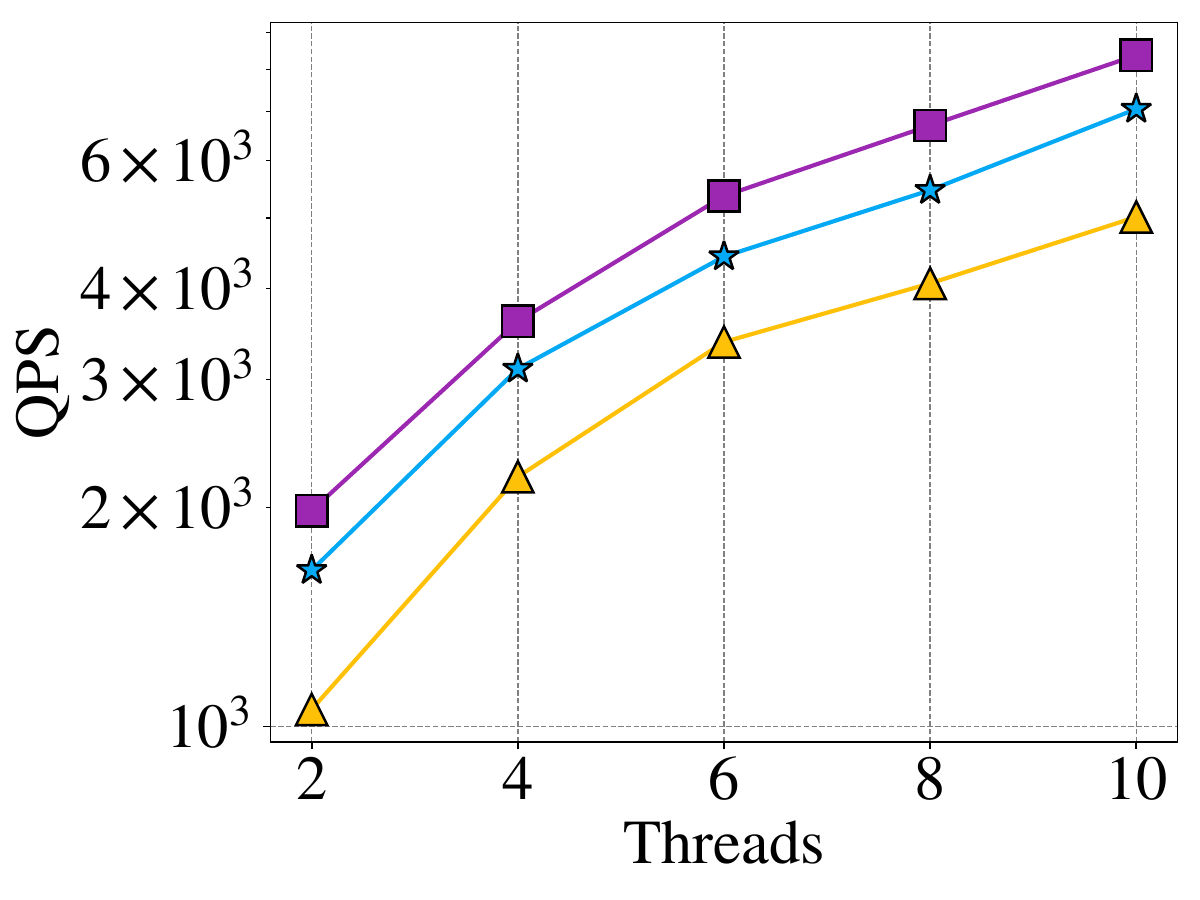}}
		\label{subfig:antsparse-top50-thread}}
	\vspace{-1ex} 
	
	\caption{\small Multi-threaded Scalability of \textsc{Sindi} (QPS) at Different Recall@50 Targets on SPLADE-FULL and AntSparse-10M Datasets.}
	\vspace{-2ex} 
	\label{fig:thread}
\end{figure}

\section{Related Work}

	\textbf{Inverted Index-based Methods.}
	Algorithms like BMW~\cite{bmw}, BMP~\cite{bmp}, and \textsc{Seismic}~\cite{seismic} utilize blocking strategies. BMW~\cite{bmw}'s dynamic pruning often degenerates to brute force due to smooth weight distributions, causing frequent cache misses. We put the comparisons with BMW in Appendix D. BMP~\cite{bmp} improves this by pre-sorting block bounds but incurs prohibitive sorting costs.  \textsc{Seismic}~\cite{seismic} utilizes geometric blocking for skipping but remains bottlenecked by inefficient exact scoring and random memory access to raw vectors.
	
	\textbf{Graph-based Methods.}
	\textsc{PyANNS}~\cite{pyanns} and SHNSW~\cite{shnsw} adapt HNSW~\cite{hnsw} for sparse data via quantization and data co-location. However, performance degrades on highly sparse datasets where limited node-query overlap causes greedy routing failures. Furthermore, graph index construction remains significantly more expensive than inverted indices.
	
	\textbf{Hashing-based Methods.}
	SINNAMON~\cite{sinnamon} and SOSIA~\cite{sosia} explore hashing pathways. SINNAMON~\cite{sinnamon} uses dense sketches for SIMD streaming but lacks static pruning, forcing the evaluation of extensive candidates. SOSIA~\cite{sosia} employs MinHash for LSH-based estimation; achieving high accuracy requires numerous hash functions, incurring significant overhead, while its randomized storage forces expensive random memory access.

\section{Conclusion}
\label{sec:conclusion}

In this work, we propose \textsc{Sindi}, an inverted index for sparse vectors that eliminates redundant inner-product computations.  
By storing non-zero entries directly in the postings, \textsc{Sindi} removes both document ID lookups and random memory accesses, and leverages SIMD acceleration to maximize CPU parallelism.  
It further introduces \textit{Mass Ratio Pruning}, which effectively preserves high-value entries, and a reordering strategy whose refinement step ensures high accuracy.
Experiments on multilingual, multi-scale real-world datasets demonstrate that \textsc{Sindi} delivers state-of-the-art performance in both recall and throughput.

\section{acknowledgment}
\label{sec:ack}
This work is supported by the ``Pioneer'' and ``Leading Goose'' R\&D Program of Zhejiang (No. 2026C02A1236), Ant Group, Fundamental Research Funds for the Central Universities, Shanghai Central and Local Science and Technology Development Fund Project (Grant No. YDZX20253100002004), the New Cornerstone Science Foundation through the XPLORER PRIZE. Corresponding author: Peng Cheng.

\section{AI-Generated Content Acknowledgement}
No content of this paper is generated by Generative AI tools and technologies, such as ChatGPT.

\bibliographystyle{IEEEtran}
\bibliography{add}

 \clearpage
 \appendix
\subsection{SIMD Implementation and Evaluation}
\textsc{Sindi}'s handwritten AVX-512 implementation relies on four key strategies:
\begin{itemize}
	\item \textbf{Latency Hiding for Sparse Access:} \textsc{Sindi} extend SIMD to sparse data using \texttt{gather/scatter} instructions; to hide their significant memory access latency, \textsc{Sindi} employ \texttt{2x loop unrolling} to schedule independent gather operations back-to-back.
	\item \textbf{Data Prefetching:} \textsc{Sindi} employ \texttt{prefetch} directives to proactively hint the CPU to fetch data for subsequent loop iterations into the L1 cache.
	\item \textbf{Fused Multiply-Add (FMA):} The core computation utilizes FMA instructions, which perform multiplication and addition in a single, high-throughput operation.
	\item \textbf{Tail Masking Process:} \textsc{Sindi} uses AVX-512's mask registers to efficiently process the final few elements of a list that do not fill a full SIMD vector.
\end{itemize}

To evaluate effect of SIMD to \textsc{Sindi}, a comparative analysis was performed between handwritten SIMD implementations and naive C++ code compiled with auto-vectorization. Experiments were conducted on an Intel Xeon Platinum 8269CY CPU using GCC 10.2.1. Tested configurations include:

\begin{itemize}
	
	\item \textbf{unoptimized}: Naive C++ code without compiler optimizations.
	
	\item \textbf{compiler-opt}: Naive C++ code compiled with auto-vectorization and optimization flags (\texttt{-Ofast, -O3, -free-vectorize, -fopen-simd, -funroll-loops}).
	
	\item \textbf{sse}: Handwritten code using the respective SIMD intrinsics ''{\texttt{-msse -msse2 -msse3 -mssse3 -msse4 -msse4a -msse4.1 -msse4.2}}''.
	
	\item \textbf{avx}: Handwritten code using the respective SIMD intrinsics ''{\texttt{-mavx}}''.
	
	\item \textbf{avx2}: Handwritten code using the respective SIMD intrinsics ''{\texttt{-mavx2 -mfma}}''.
	
	\item \textbf{avx512}: Handwritten code using the respective SIMD intrinsics ''{\texttt{-mavx512f -mavx512pf -mavx512er -mavx512cd -mavx512vl -mavx512bw -mavx512dq -mavx512ifma -mavx512vbmi}}''.
	
\end{itemize}

The end-to-end QPS results on the SPLADE-1M and SPLADE-FULL datasets are presented in Table \ref{tab:base_1m_qps_speedup} and Table \ref{tab:base_full_qps_speedup}.

\begin{table*}[t]
	\centering
	\caption{\textit{QPS and Speedup relative to scalar baseline on SPLADE-1M dataset.}}
	\label{tab:base_1m_qps_speedup}
	\resizebox{\textwidth}{!}{
		\begin{tabular}{ccccccc}
			\toprule
			$\lambda$ & unoptimized & sse & avx & avx2 & compiler-opt & avx512 \\
			\midrule
			100000 & 1407.3 & 1413.6 (1.00$\times$) & 1513.1 (1.08$\times$) & 1584.2 (1.13$\times$) & 1711.1 (1.22$\times$) & \textbf{1776.4 (1.26$\times$)} \\
			300000 & 1141.6 & 1153.1 (1.01$\times$) & 1276.7 (1.12$\times$) & 1335.3 (1.17$\times$) & 1439.0 (1.26$\times$) & \textbf{1482.4 (1.30$\times$)} \\
			500000 & 983.2 & 1010.6 (1.03$\times$) & 1104.2 (1.12$\times$) & 1183.6 (1.20$\times$) & 1240.1 (1.26$\times$) & \textbf{1275.9 (1.30$\times$)} \\
			700000 & 909.5 & 957.7 (1.05$\times$) & 1024.8 (1.13$\times$) & 1085.5 (1.19$\times$) & 1182.3 (1.30$\times$) & \textbf{1203.4 (1.32$\times$)} \\
			1000000 & 828.9 & 858.9 (1.04$\times$) & 936.2 (1.13$\times$) & 975.2 (1.18$\times$) & 1060.8 (1.28$\times$) & \textbf{1071.3 (1.29$\times$)} \\
			\bottomrule
		\end{tabular}
	}
\end{table*}

\begin{table*}[ht]
	\centering
	\caption{\textit{QPS and Speedup relative to scalar baseline on SPLADE-FULL dataset.}}
	\label{tab:base_full_qps_speedup}
	\resizebox{\textwidth}{!}{
		\begin{tabular}{ccccccc}
			\toprule
			$\lambda$ & unoptimized & sse & avx & avx2 & compiler-opt & avx512 \\
			\midrule
			100000 & 276.9 & 282.5 (1.02$\times$) & 298.5 (1.08$\times$) & 318.8 (1.15$\times$) & 339.1 (1.22$\times$) & \textbf{348.3 (1.26$\times$)} \\
			300000 & 218.6 & 228.7 (1.05$\times$) & 242.7 (1.11$\times$) & 252.3 (1.15$\times$) & 257.4 (1.18$\times$) & \textbf{268.1 (1.23$\times$)} \\
			500000 & 194.7 & 199.7 (1.03$\times$) & 206.3 (1.06$\times$) & 215.9 (1.11$\times$) & 220.6 (1.13$\times$) & \textbf{238.8 (1.23$\times$)} \\
			700000 & 182.5 & 186.7 (1.02$\times$) & 194.8 (1.07$\times$) & 198.6 (1.09$\times$) & 206.7 (1.13$\times$) & \textbf{220.8 (1.21$\times$)} \\
			1000000 & 172.4 & 178.5 (1.04$\times$) & 181.3 (1.05$\times$) & 184.1 (1.07$\times$) & 192.0 (1.11$\times$) & \textbf{205.8 (1.19$\times$)} \\
			\bottomrule
		\end{tabular}
	}
\end{table*}

Results indicate that handwritten AVX-512 consistently yields peak performance (up to 1.32$\times$ speedup), validating the low-level optimizations. Notably, the compiler-optimized version (\texttt{compiler-opt}) achieves efficiency comparable to AVX-512 (e.g., 1.22$\times$ vs. 1.26$\times$ on SPLADE-1M). This confirms that the memory design of SINDI is inherently amenable to vectorization, allowing efficient deployment on platforms lacking AVX-512 via standard compiler flags.

\subsection{Sensitivity Analysis of Window Size $\lambda$ across Hardware and Datasets}
\label{subsec: window_tuning}

To assess the portability of the window size parameter $\lambda$ and validate the robustness of the proposed double power-law cost model, we reproduced the sensitivity analysis (originally Example 6) across three distinct server configurations (specifications detailed in Table \ref{tab:hardware_specs}). Experiments were conducted on both SPLADE (English) and AntSparse (Chinese) datasets using full-precision SINDI. We utilized Intel VTune Profiler to measure memory-bound metrics, specifically distinguishing between random accesses to the distance array and cache eviction overheads during sub-list switching. For each dataset $\mathcal{D}$, we sampled 8 values of $\lambda$ on a logarithmic scale ranging from $1K$ to $||\mathcal{D}||$.

\begin{table}[t]
	\centering
	
	\caption{Hardware Specifications for Sensitivity Analysis}
	\label{tab:hardware_specs}
	
	\begin{tabularx}{\columnwidth}{@{} l X X X @{}} 
		\toprule
		\textbf{Component} & \textbf{Server A} & \textbf{Server B} & \textbf{Server C}\\ 
		\midrule
		\textbf{CPU} & Intel Xeon Platinum 8269CY @ 2.50GHz & Intel Xeon Platinum 8163 @ 2.50GHz & Intel Xeon CPU E5-2650 v2 @ 2.60GHz\\ 
		\textbf{L1d Cache} & 32K & 32K & 32K\\ 
		\textbf{L1i Cache} & 32K & 32K & 32K\\ 
		
		\textbf{L2 Cache} & 1,024K & 1,024K & 256K\\ 
		\textbf{L3 Cache} & 36,608K & 33,792K & 20,480K\\ 
		\textbf{Memory} & 512 GB & 502 GB & 125 GB\\
		\bottomrule
	\end{tabularx}
\end{table}

Figure \ref{fig:window_appendix} illustrates the correlation between theoretical predictions and measured QPS, alongside the memory-bound proportions. Table \ref{tab:lambda_results} summarizes the fitted theoretical optima ($\lambda^*$), the measured optimal intervals, and the corresponding QPS stability ranges. Due to memory limitations, Server C can only test datasets with a size of 1M.

\begin{table}[t]
	\centering
	\caption{Theoretical $\lambda^*$ and Optimal Window Interval with corresponding QPS Interval}
	\label{tab:lambda_results}
	\begin{tabularx}{\columnwidth}{@{} ll r >{\centering}X >{\centering\arraybackslash}X @{}}
		\toprule
		\textbf{Dataset} & \textbf{Server} & $\mathbf{\lambda^*}$ & \textbf{$\lambda$ interval} & \textbf{QPS interval}\\ 
		\midrule
		SPLADE-FULL   & Server A & 120,656 & [15k, 200k] & [21, 24] \\ 
		AntSparse-10M & Server A & 51,024  & [14k, 200k] & [90, 94]\\ 
		
		SPLADE-FULL   & Server B & 92,936  & [15k, 200k] & [25, 26] \\ 
		AntSparse-10M & Server B & 47,555  & [14k, 200k] & [74, 78]\\ 
		
		SPLADE-1M     & Server C & 20,794  & [7k, 50k]   & [171, 187] \\ 
		AntSparse-1M  & Server C & 11,308  & [7k, 50k]   & [594, 607]\\ 
		\bottomrule
	\end{tabularx}
\end{table}

\begin{figure}[!t]
	\centering
	\scalebox{0.4}{\includegraphics{figures/window/bar.pdf}}
	\vspace{1ex}
	
	\subfigure[{\scriptsize Server-A SPLADE-FULL}]{%
		\scalebox{0.15}{\includegraphics{figures/revision/window/f63_base_full_window.pdf}}
		\label{subfig:f63-splade-full-window}
	}
	\subfigure[{\scriptsize Server-A AntSparse-10M}]{%
		\scalebox{0.15}{\includegraphics{figures/revision/window/f63_wholenet_10M_window.pdf}}
		\label{subfig:f63-antsparse-10m-window}
	}
	\vspace{1ex}
	
	\subfigure[{\scriptsize Server-B SPLADE-FULL}]{%
		\scalebox{0.15}{\includegraphics{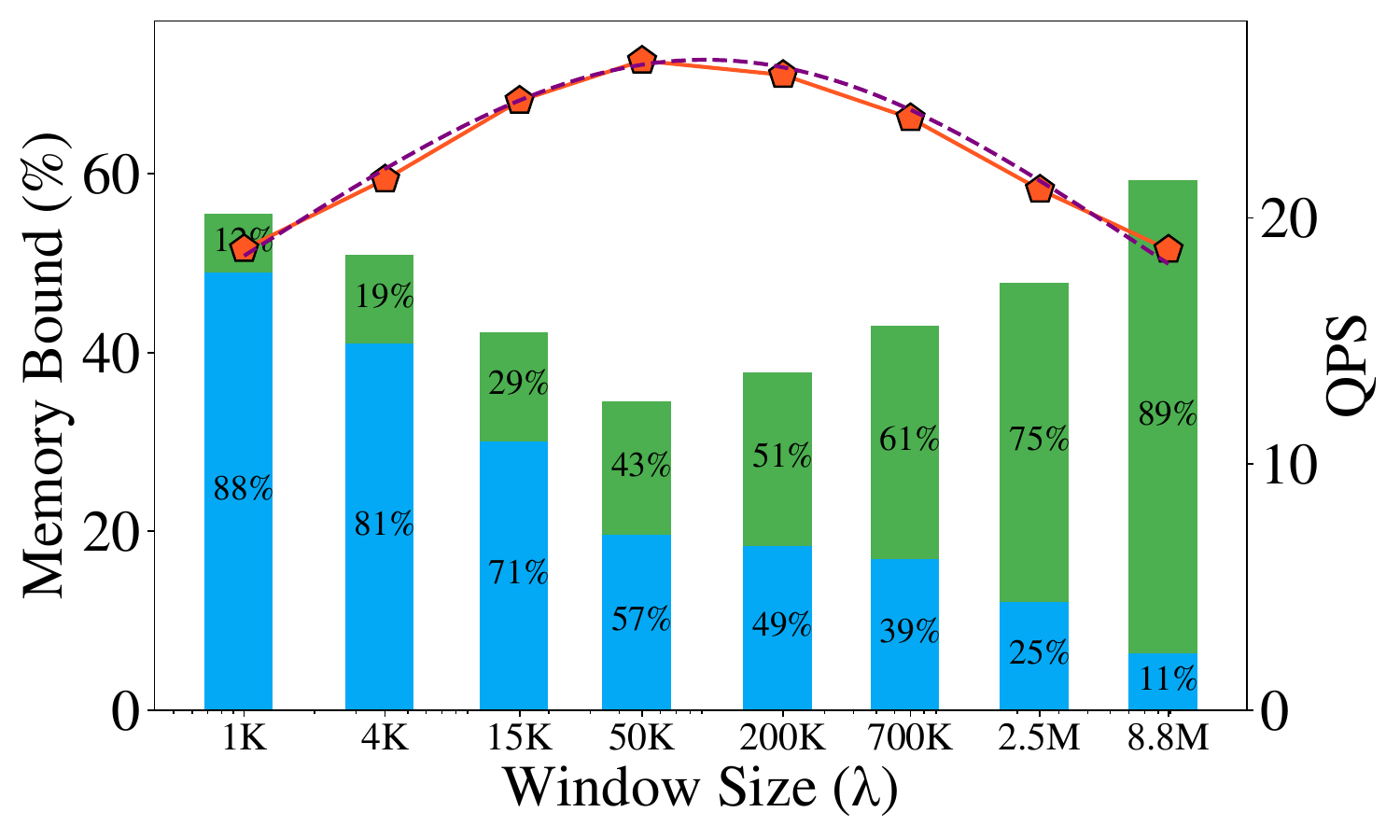}}
		\label{subfig:f53-splade-full-window}
	}
	\subfigure[{\scriptsize Server-B AntSparse-10M}]{%
		\scalebox{0.15}{\includegraphics{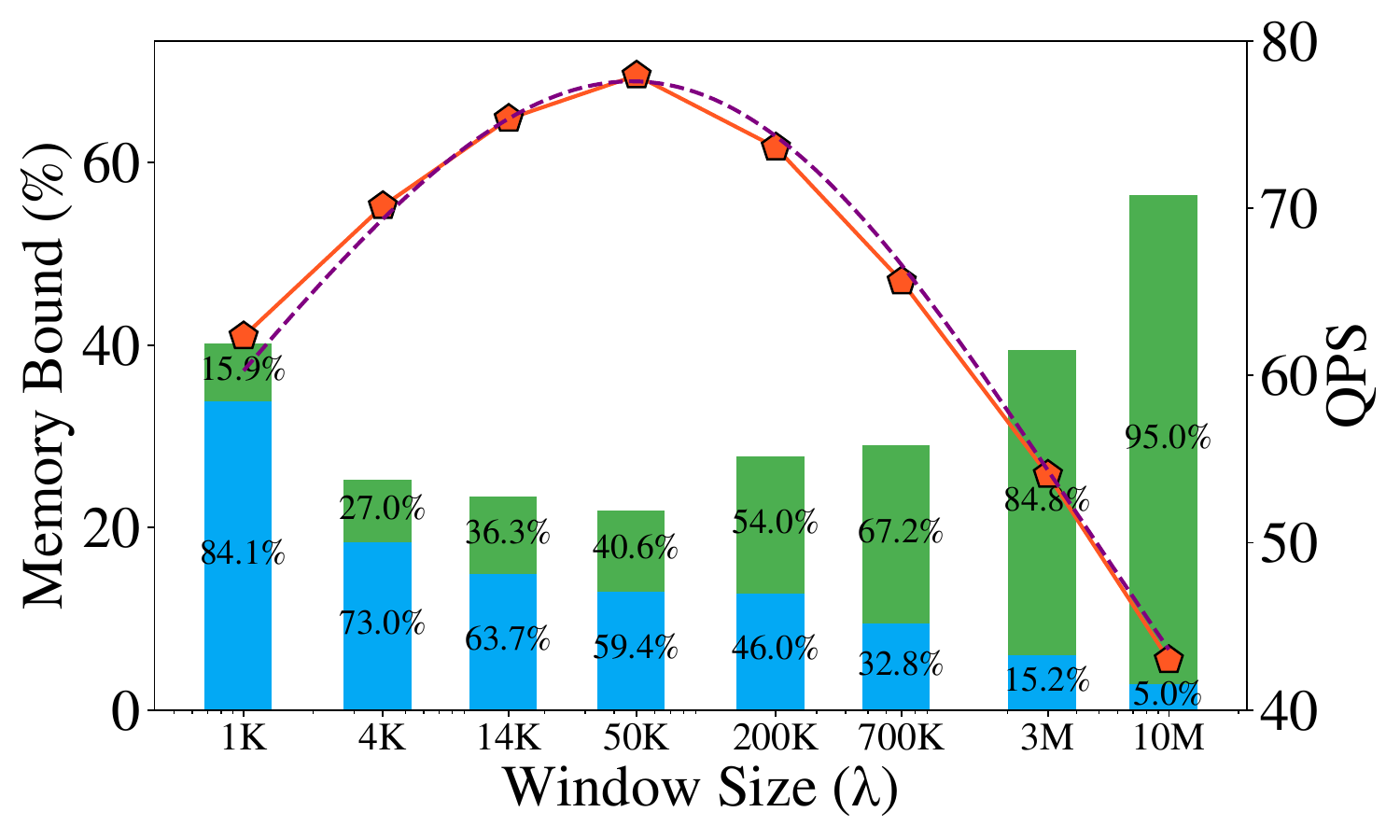}}
		\label{subfig:f53-antsparse-10m-window}
	}
	\vspace{1ex}
	
	\subfigure[{\scriptsize Server-C SPLADE-1M}]{%
		\scalebox{0.15}{\includegraphics{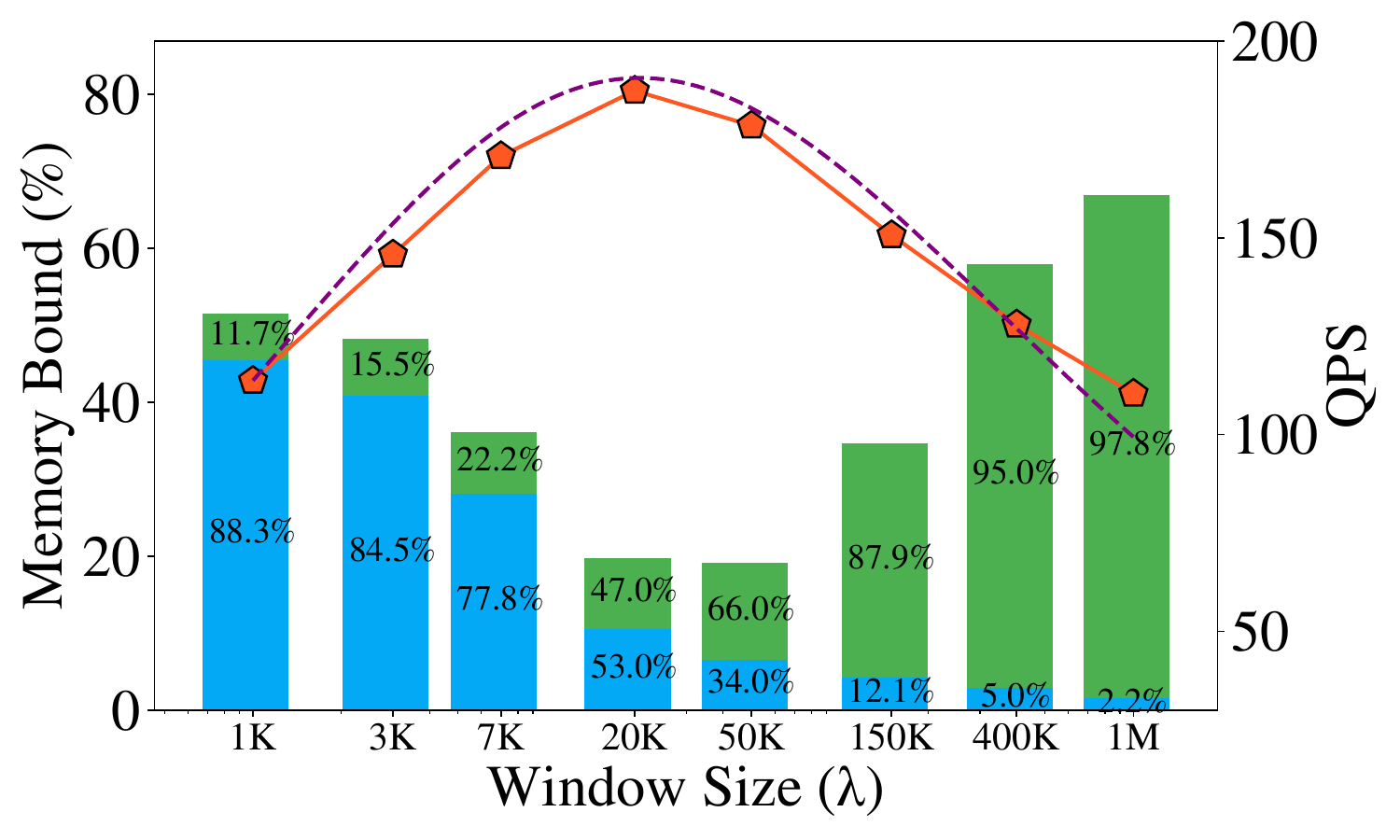}}
		\label{subfig:tbase1-splade-1m-window}
	}
	\subfigure[{\scriptsize Server-C AntSparse-1M}]{%
		\scalebox{0.15}{\includegraphics{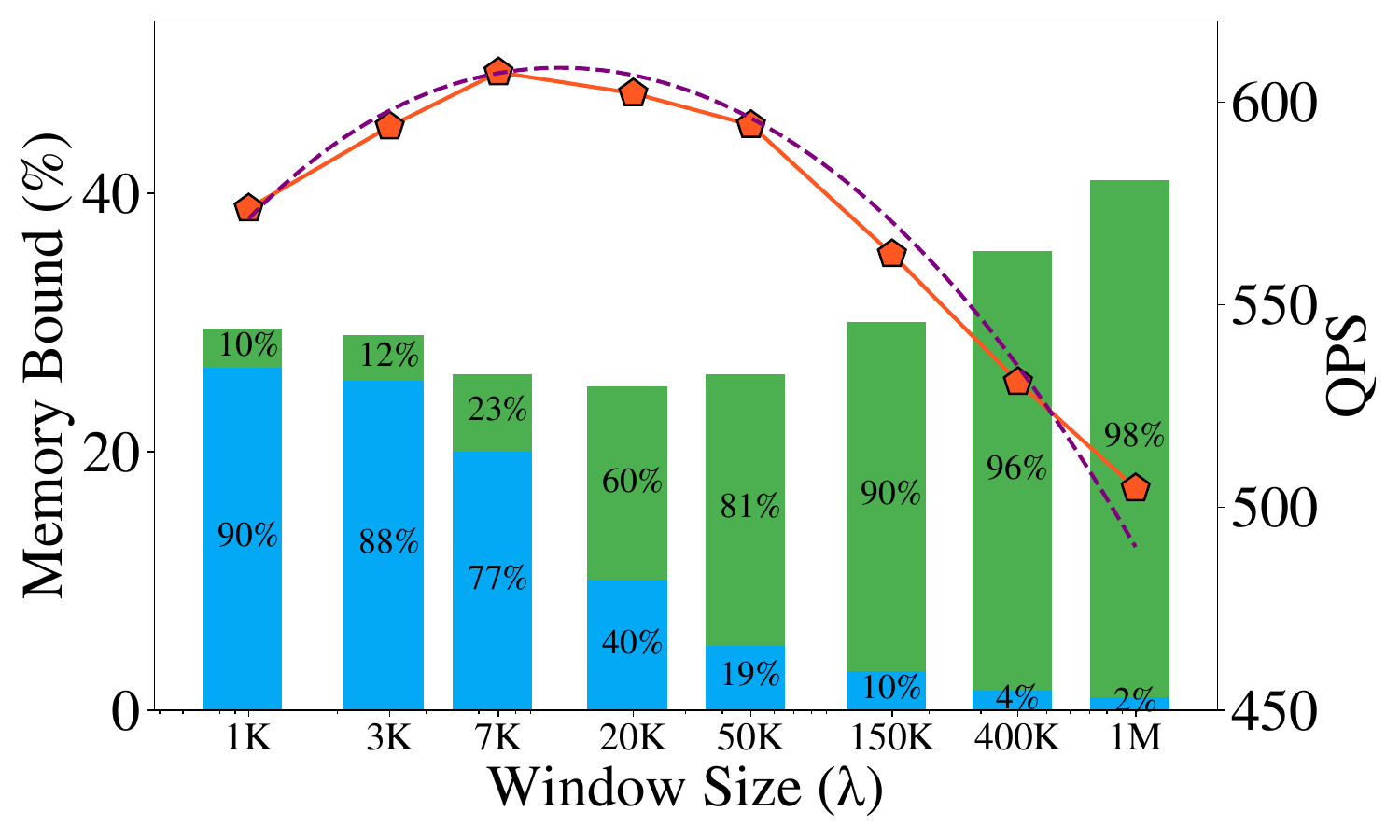}}
		\label{subfig:tbase1-antsparse-1m-window}
	}
	
	\vspace{1ex}
	\caption{Impact of Window Size on Query Throughput and Memory Accesses across different hardware configurations.}
	\label{fig:window_appendix}
	\vspace{-2ex}
\end{figure}

\textbf{Analysis of Influencing Factors.} The optimal window size $\lambda^*$ is primarily determined by two factors: hardware cache capacity and dataset sparsity.
\begin{itemize}
	\item \textbf{Cache Capacity:} Larger L3 caches accommodate larger distance arrays, reducing random access penalties ($c_{rand}$) and shifting the equilibrium $\lambda^*$ to higher values. For instance, Server A (36MB L3) exhibits a larger $\lambda^*$ than Server C (20MB L3) on the SPLADE dataset.
	\item \textbf{Dataset Sparsity:} Higher sparsity (as seen in AntSparse) results in shorter inverted lists, lowering the memory bandwidth cost of switching windows ($c_{evict}$). Consequently, highly sparse data favors smaller optimal windows. Table \ref{tab:lambda_results} consistently shows smaller $\lambda^*$ values for AntSparse compared to SPLADE across all servers.
\end{itemize}

\textbf{Practical Parameter Recommendations.} Our analysis reveals high parameter robustness, with query throughput remaining stable across order-of-magnitude ranges. Based on these findings, we recommend a simple **binary search method** to efficiently locate the extremum, although a general ``rule of thumb'' range (e.g., $\lambda \in [10,000, 120,000]$) is usually sufficient due to the flat performance plateau. For sparser datasets or cache-limited server, we recommend $\lambda \in [10,000, 50,000]$; for denser datasets or cache-rich server, we recommend $\lambda \in [50,000, 120,000]$.

\subsection{Baseline Parameter Optimization and Configuration}
\label{subsec:param_settings}

To ensure the fairness and optimality of the comparative evaluation, \textsc{Sindi} adopted a hybrid parameter selection strategy derived from three sources: recommendations from original publications, default settings from recognized benchmarks (e.g., BigANN), and extensive self-tuning via grid search. Table \ref{tab:final_parameters} summarizes the final parameter settings used for all algorithms across datasets, while Table \ref{tab:param_source} details the specific methodology and search ranges employed. \textsc{PyANNS} and HGRAPH's performance under RANDOM dataset is very poor due to random dataset's query and base vector's inner product almost zero.

\textit{The selection logic is threefold:}
\begin{itemize}
	\item \textbf{Literature Recommendations:} For algorithms such as \textsc{Seismic} and \textsc{Sosia}, we adopted the optimal parameters recommended in their respective original papers, as these settings were derived from extensive grid searches.
	\item \textbf{Benchmark Defaults:} For \textsc{PyANNS}, we utilized the configuration from the BigANN Benchmark, where it secured the top position in the Sparse Track. These settings represent the current state-of-the-art standard for SPLADE-based retrieval.
	\item \textbf{High-Recall Optimization:} For scenarios without established defaults (e.g., \textsc{Bmp} or \textsc{Seismic} on the new AntSparse datasets), we performed grid searches specifically optimizing for Query Per Second (QPS) throughput within the high-recall region (Recall@50 $\in [0.9, 1.0]$).
\end{itemize}

\begin{table*}[htbp]
	\centering
	\caption{Final Parameter Settings for All Algorithms and Datasets}
	\label{tab:final_parameters}
	\renewcommand{\arraystretch}{1.3}
	\small
	\begin{tabularx}{\textwidth}{l|X|c|c|c|c}
		\toprule
		\textbf{Dataset} & \textbf{\textsc{Seismic}} ($\lambda, \beta, \alpha$) & \textbf{\textsc{PyANNS}} & \textbf{SOSIA} & \textbf{BMP} & \textbf{HNSW} \\
		\hline
		\textbf{SPLADE-1M} & $\lambda=1400, \beta=140, \alpha=0.4$ & \multirow{5}{*}{\makecell{$R=32$ \\ $L=1000$}} & \multirow{6}{*}{\makecell{$l=50$ \\ $m=150$}} & \multirow{6}{*}{$b=16$} & \multirow{5}{*}{\makecell{$R=32$ \\ $L=1000$}} \\
		\cline{1-2}
		\textbf{SPLADE-FULL} & $\lambda=6000, \beta=400, \alpha=0.4$ & & & & \\
		\cline{1-2}
		\textbf{NQ} & $\lambda=5250, \beta=525, \alpha=0.5$ & & & & \\
		\cline{1-2}
		\textbf{AntSparse-1M} & $\lambda=2000, \beta=200, \alpha=0.5$ & & & & \\
		\cline{1-2}
		\textbf{AntSparse-10M} & $\lambda=10000, \beta=1000, \alpha=0.5$ & & & & \\
		\cline{1-3} \cline{6-6}
		\textbf{RANDOM-5M} & $\lambda=6000, \beta=600, \alpha=0.4$ & -- & & & -- \\
		\bottomrule
	\end{tabularx}
\end{table*}

\begin{table*}[htbp]
	\centering
	\caption{Parameter Selection Methodology and Tuning Ranges}
	\label{tab:param_source}
	\renewcommand{\arraystretch}{1.3}
	\small
	\begin{tabularx}{\textwidth}{l|l|X}
		\toprule
		\textbf{Algorithm} & \textbf{Dataset} & \textbf{Methodology \& Search Ranges} \\
		\hline
		\multirow{6}{*}{\textbf{\textsc{Seismic}}} & SPLADE-1M & \textbf{Self-Tuning:} Grid search $\lambda \in [1000, 2000]$ (step 200). \\
		\cline{2-3}
		& SPLADE-FULL & \textbf{Paper Rec.} : Grid search $\lambda \in [1500, 7500]$ (step 500), $\beta \in [150, 750]$ (step 50), $\alpha \in [0.1, 0.5]$ (step 0.1). \\
		\cline{2-3}
		& NQ & \textbf{Paper Rec.} : Grid search $\lambda \in \{4500, ...\}$, $\beta \in \{300, ...\}$, $\alpha \in \{0.3, 0.4, 0.5\}$. \\
		\cline{2-3}
		& RANDOM-5M & \textbf{Self-Tuning:} Grid search $\lambda \in [3000, 10000]$ (step 1000). \\
		\cline{2-3}
		& AntSparse-1M & \textbf{Self-Tuning:} Grid search $\lambda \in [800, 2000]$ (step 200). \\
		\cline{2-3}
		& AntSparse-10M & \textbf{Self-Tuning:} Grid search $\lambda \in [2000, 10000]$ (step 1000). \\
		\hline
		\multirow{2}{*}{\textbf{\textsc{PyANNS}}} & SPLADE & \textbf{Benchmark Rec.} : BigANN Sparse Track winner config. \\
		\cline{2-3}
		& Others & \textbf{Self-Tuning:} Grid search $R \in \{16, 32, 64\}, L \in \{1000, 2000\}$. \\
		\hline
		\textbf{SOSIA} & All & \textbf{Paper Rec.} : Grid search $l \in \{1, 10, 20, 40, 80\}$, $m \in \{25, 50, 100, 150, 200\}$. \\
		\hline
		\textbf{BMP} & All & \textbf{Self-Tuning:} Grid search $b \in \{16, 32, 64\}$. \\
		\hline
		\textbf{HNSW} & All & \textbf{Default:} Inherits optimal settings from \textsc{PyANNS}. \\
		\bottomrule
	\end{tabularx}
\end{table*}

\subsection{Comparison with Exact Retrieval Baseline (BMW)}

Block-Max WAND (BMW) is a classic dynamic pruning algorithm designed for efficient exact retrieval on full-text search. It relies on storing upper-bound scores for blocks of postings to skip non-competitive documents. To validate the efficiency of \textsc{Sindi}'s hardware-aware design against this algorithmic skipping approach, we compared the query throughput (QPS) of BMW using the PISA engine against the full-precision version of \textsc{Sindi} using the PISA engine. The results for Top-50 and Top-100 retrieval are presented in Table \ref{tab:full-precision-qps-combined}.

\begin{table}[htbp]
	\centering
	\caption{Performance Comparison (QPS) of BMW and Full-Precision SINDI}
	\label{tab:full-precision-qps-combined}
	\begin{tabular}{lcccc}
		\toprule
		& \multicolumn{2}{c}{\textbf{Top-50}} & \multicolumn{2}{c}{\textbf{Top-100}} \\
		\cmidrule(lr){2-3} \cmidrule(lr){4-5}
		\textbf{Dataset} & {\textbf{BMW}} & {\textbf{SINDI}} & {\textbf{BMW}} & {\textbf{SINDI}} \\
		\midrule
		SPLADE-1M & 5.22 & \textbf{301.30} & 5.33 & \textbf{304.06} \\
		SPLADE-FULL & 0.68 & \textbf{33.78} & 0.66 & \textbf{32.85} \\
		AntSparse-1M & 436.55 & \textbf{1029.59} & 473.20 & \textbf{1022.98} \\
		AntSparse-10M & 75.23 & \textbf{104.56} & 66.98 & \textbf{103.36} \\
		RANDOM-5M & 6.93 & \textbf{124.57} & 7.19 & \textbf{127.12} \\
		NQ & 1.72 & \textbf{90.89} & 1.67 & \textbf{90.85} \\
		\bottomrule
	\end{tabular}
\end{table}

\textbf{Analysis.} \textsc{Sindi} significantly outperforms BMW on SPLADE datasets (approx. $50\times$ speedup). This is because SPLADE vectors exhibit smooth score distributions that render BMW's block-skipping ineffective. Moreover, SPLADE's long queries (nearly 50) leads to many branch prediction errors and cache misses. However, on the AntSparse dataset, the performance gap narrows. AntSparse is characterized by extremely high dimensionality (250k) and very short queries (avg. 6 terms). In this sparse regime, BMW's ability to skip posting lists becomes more advantageous. \textsc{Sindi} maintains superior query throughput due to its cache-friendly memory layout and lack of branch misprediction overhead.

\subsection{Sensitivity Analysis of Reordering Depth $\gamma$}

To evaluate the sensitivity of retrieval performance to the reordering candidate size $\gamma$, experiments were conducted on the \textbf{SPLADE-FULL} and \textbf{AntSparse-10M} datasets. The parameter $\gamma$ was varied across broad ranges under three distinct query pruning ratios ($\beta$). Figure \ref{fig:gamma} illustrates the resulting QPS-Recall trajectories.

\textbf{Analysis.} The curves exhibit a consistent trend across all configurations: recall increases with $\gamma$ and eventually reaches a saturation plateau. It is observed that this saturation occurs more rapidly as the query pruning ratio $\beta$ increases. Extending $\gamma$ beyond the saturation point offers negligible improvements in recall but causes a linear degradation in QPS due to the increased cost of full-precision computations.

\begin{figure}[!t]
	\centering
	\scalebox{0.35}{\includegraphics{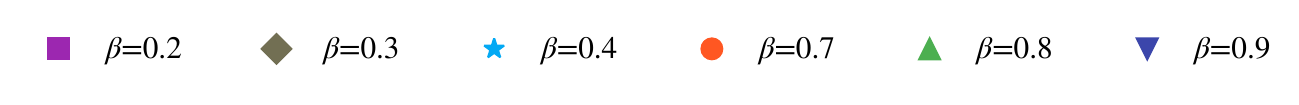}}
	\vspace{1ex}  
	
	\subfigure[{\scriptsize SPLADE-FULL}]{%
		\scalebox{0.18}{\includegraphics{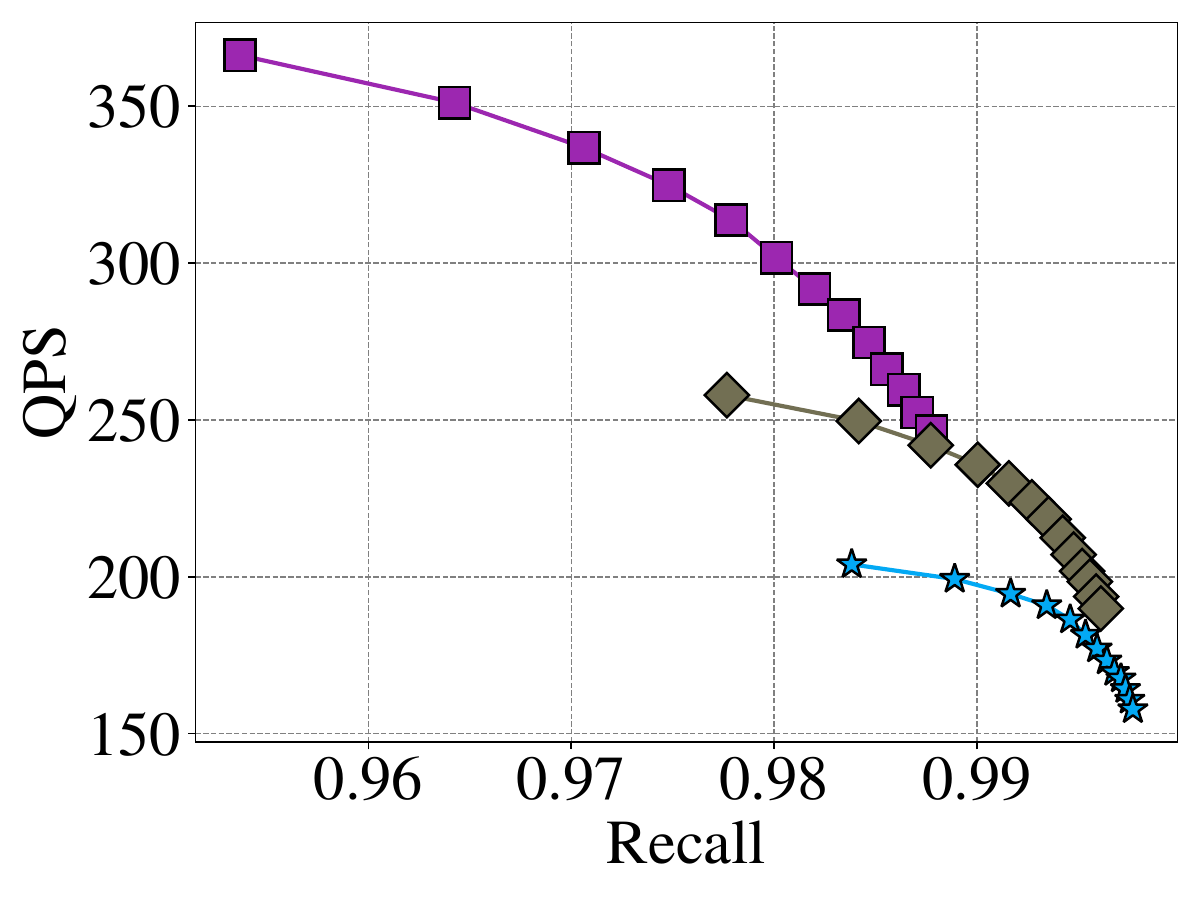}}
		\label{subfig:splade-full-gamma}
	}
	\subfigure[{\scriptsize AntSparse-10M}]{%
		\scalebox{0.18}{\includegraphics{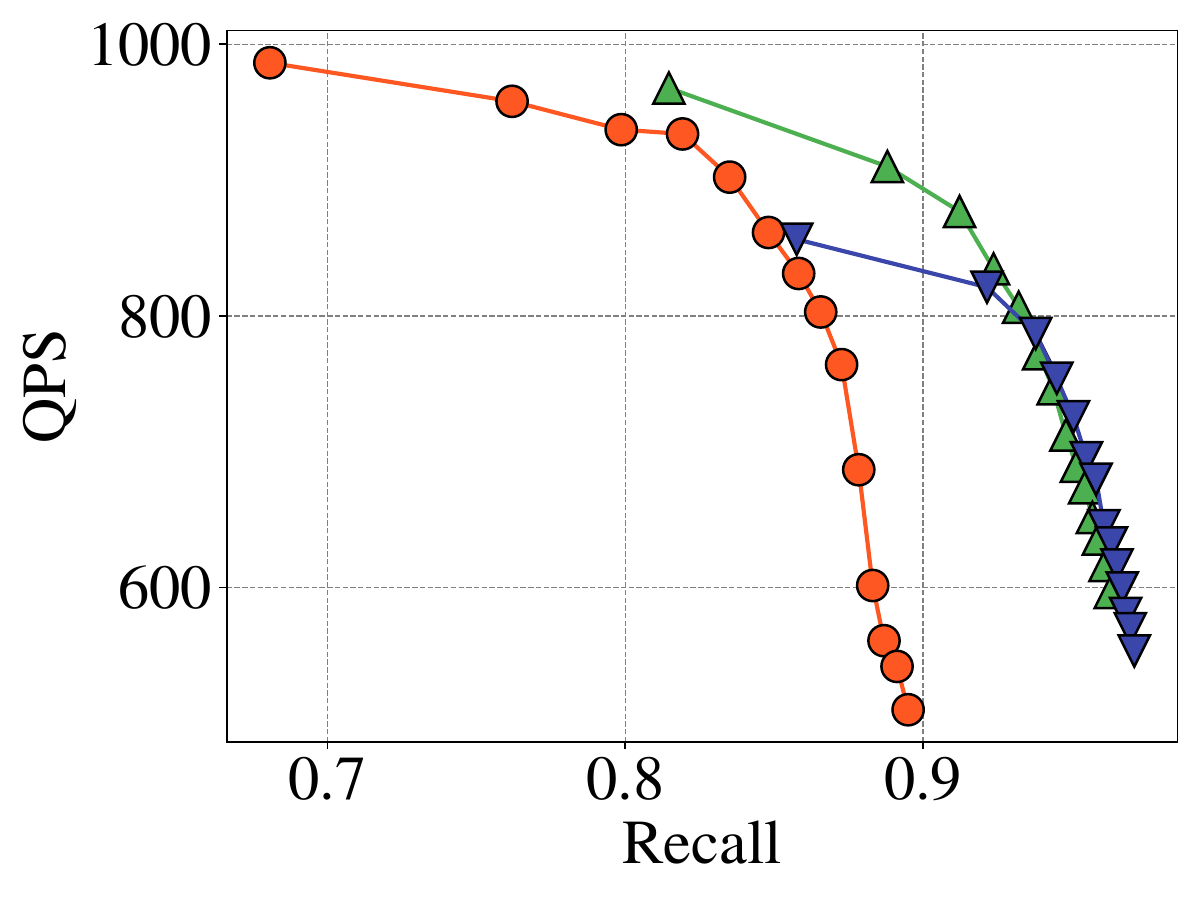}}
		\label{subfig:antsparse-10m-gamma}
	}
	
	\vspace{1ex}
	\caption{\small Sensitivity Analysis of Reordering Candidate Size $\gamma$.}
	\label{fig:gamma}
	\vspace{-2ex}  
\end{figure}

\subsection{Performance Comparison with Lucene-based System}
\label{app:lucene_comparison}

To demonstrate the fundamental advantages of \textsc{Sindi} over standard ``IR-style'' indexing, a system-level comparison was conducted against Elasticsearch, the most widely adopted production system based on the Lucene core. Evaluations were performed in two distinct scenarios to ensure robustness: a standalone system comparison and an industrial integration assessment.

\begin{enumerate}
	\item \textbf{Standalone System Comparison:} \textsc{Sindi} (integrated into TBaseSearch~\cite{tbase}) was compared against Elasticsearch under identical configurations (4 threads, top-$k=10$). \textsc{Sindi} achieved \textbf{1,775 QPS}, approximately \textbf{3$\times$ higher} than Elasticsearch's 575 QPS, while concurrently reducing average latency by \textbf{67\%} (2.23 ms vs. 6.92 ms). These metrics confirm the superior core efficiency of the proposed architecture.
	
	\item \textbf{Industrial Integration (OceanBase):} To validate performance within complex commercial infrastructure, OceanBase~\cite{oceanbase} (a distributed relational database integrating \textsc{Sindi}) was evaluated against Elasticsearch on the SPLADE-FULL dataset using Intel Xeon Platinum 8269CY CPUs. As shown in Figure~\ref{fig:oceanbase}, the integration consistently outperforms the baseline:
	\begin{itemize}
		\item \textbf{Full Recall (No Pruning):} OceanBase achieves \textbf{78 QPS} compared to Elasticsearch's 33 QPS, representing a \textbf{2.4$\times$} speedup even without approximation strategies.
		\item \textbf{With Pruning:} Performance gains are approximately \textbf{3$\times$} without reordering, extending to \textbf{2$\times$--5$\times$} when reordering is enabled.
	\end{itemize}
\end{enumerate}

As shown in Table \ref{tab:perf_comparison_remote}, the results indicate that \textsc{Sindi} is fundamentally more efficient than standard IR-style implementations. While Lucene effectively utilizes skipping (Block-Max WAND), it remains limited by scalar processing. In contrast, the value-storing design of \textsc{Sindi} facilitates cache-friendly access and \textbf{SIMD acceleration}, delivering order-of-magnitude improvements over traditional methods.

\begin{table}[htbp]
	\centering
	\caption{Performance Comparison between \textsc{Sindi} and Elasticsearch}
	\label{tab:perf_comparison_remote}
	\begin{tabular}{lcc}
		\toprule
		\textbf{Metric} & \textbf{TbaseSearch} & \textbf{Elasticsearch} \\ \midrule
		Memory Usage (RSS) & 1.09 GB & 1.65 GB \\
		Number of Queries & 1,000 & 1,000 \\
		Top-$k$ & 10 & 10 \\
		Client Threads & 4 & 4 \\ \midrule
		Average Recall & 0.9988 & 0.9950 \\
		QPS (Queries Per Second) & 1,775.11 (\textbf{+208.4\%}) & 575.61 \\
		Avg Latency (ms) & 2.23 (\textbf{-67.8\%}) & 6.92 \\
		Min Latency (ms) & 0.58 (\textbf{-82.2\%}) & 3.26 \\
		P99 Latency (ms) & 4.22 (\textbf{-63.7\%}) & 11.62 \\
		\bottomrule
	\end{tabular}
\end{table}

\newpage
\begin{figure}[htbp]
	\centering
	\scalebox{0.4}{\includegraphics{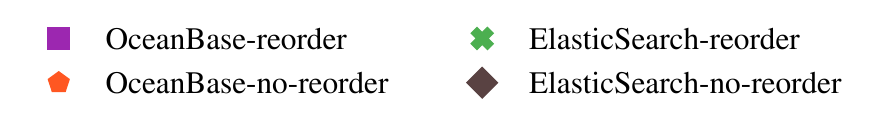}} \\
	\vspace{0.5ex} 
	
	\scalebox{0.3}{\includegraphics{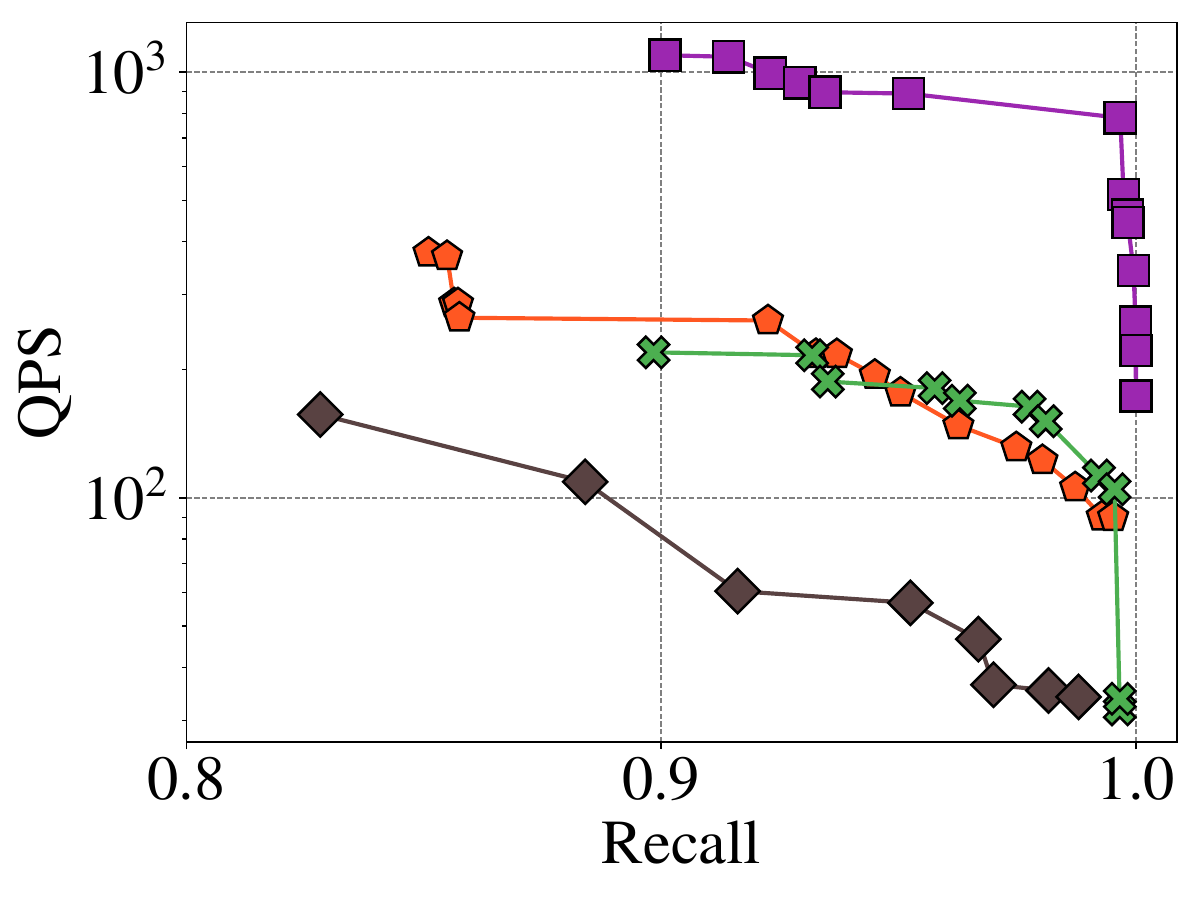}}
	
	\caption{Performance comparison between OceanBase (integrating \textsc{Sindi}) and Elasticsearch on the SPLADE-FULL dataset.}
	\label{fig:oceanbase}
	\vspace{-2ex}
\end{figure}

\end{document}